\documentclass[12pt]{article}
\pdfoutput=1
\usepackage[square,comma,numbers,sort&compress]{natbib}
\usepackage{graphicx,epstopdf,amssymb,amsfonts,amsmath,amsthm,array,
mathrsfs,amscd,enumitem}
\usepackage{float}
\usepackage{xcolor}
\usepackage{appendix}
\usepackage{fnpct}
\usepackage{lmodern}
\usepackage[T1]{fontenc}
\usepackage[utf8]{inputenc}
\usepackage{hyperref}
\newcommand{\pbs}[1]{\let\temp=\\#1\let\\=\temp}
\numberwithin{equation}{section}
\def\be{\begin{equation}}\def\ee{\end{equation}}
\providecommand{\abs}[1]{\lvert#1\rvert}
\def\cvp{\raise 2pt\hbox{,}} 
 \def\tr{\mathop{\text{tr}}\nolimits}
\def\im{\mathop{\text{Im}}\nolimits}
\def\re{\mathop{\text{Re}}\nolimits}  

 \def\d{{\rm d}} 
\def\la{\lambda}\def\La{\Lambda}
 \def\uN{\text{U}(N)}

\DeclareMathOperator{\diff}{Diff}\DeclareMathOperator{\met}{Met}
\DeclareMathOperator{\wig}{Wig}
\DeclareMathOperator{\imm}{Imm}\DeclareMathOperator{\himm}{hImm}
\DeclareMathOperator{\emb}{Emb}\DeclareMathOperator{\salmet}{SAL}

\DeclareMathOperator{\Sch}{\text{Sch}}

\DeclareMathOperator{\rank}{\text{rk}}
\DeclareMathOperator{\Iso}{\text{Iso}}\DeclareMathOperator{\Aut}{\text{Aut}}
\DeclareMathOperator{\Ai}{\text{Ai}}\DeclareMathOperator{\Bi}{\text{Bi}}

\def\disk{\mathscr D}

\def\diffSp{\diff_{+}(\Sone)}

\def\Htwo{\text{H}^{2}}\def\Sone{\text{S}^{1}}\def\Stwo{\text{S}^{2}}\def\Etwo{\text{E}^{2}}

\def\PslR{\text{PSL}(2,\mathbb R)}

\def\PslR{\text{PSL}(2,\mathbb R)}

\def\Htwo{\text{H}^{2}}
\def\Etwo{\text{E}^{2}}\def\Stwo{\text{S}^{2}}

\def\la{\lambda}
\def\La{\Lambda}

\def\disk{\mathscr D}

\def\met{\text{Met}}

\def\SOPu{$\text{SOP}_{\text u}\ $}

\newtheorem{definition}{Definition}[section]
\newtheorem{lemma}{Lemma}[section]
\newtheorem{proposition}{Proposition}[section]
\newtheorem{hypothesis}{Hypothesis}[section]
\newtheorem{theorem}{Theorem}[section]
\newtheorem{corollary}{Corollary}[section]
\newtheorem{conjecture}{Conjecture}[section]

\def\pla#1#2#3{{\it Phys.\ Lett.\ }{\bf A #1} (#2) #3}
\def\plb#1#2#3{{\it Phys.\ Lett.\ }{\bf B #1} (#2) #3}
\def\npb#1#2#3{{\it Nucl.\ Phys.\ }{\bf B #1} (#2) #3}

\def\prl#1#2#3{{\it Phys.\ Rev.\ Lett.\ }{\bf #1} (#2) #3}
\def\jhep#1#2#3{{\it J. High Energy Phys.\ }{\bf #1} (#2) #3}
\def\prd#1#2#3{{\it Phys.\ Rev.\ }{\bf D #1} (#2) #3}

\def\pre#1#2#3{{\it Phys.\ Rev.\ }{\bf E #1} (#2) #3}

\def\atmp#1#2#3{{\it Adv.\ Theor.\ Math.\ Phys.\ }{\bf #1} (#2) #3}
\def\cmp#1#2#3{{\it Comm.\ Math.\ Phys.\ }{\bf #1} (#2) #3}
\def\pr#1#2#3{{\it Phys.\ Rep.\ }{\bf #1} (#2) #3}

\def\ijmpa#1#2#3{{\it Int.\ J.\ Mod.\ Phys.\ }{\bf A #1} (#2) #3}
\def\mpla#1#2#3{{\it Mod.\ Phys.\ Lett.\ }{\bf A #1} (#2) #3}
\def\ap#1#2#3{{\it Ann.\ of Phys.\ }{\bf #1} (#2) #3}

\def\rmp#1#2#3{{\it Rev.\ Mod. Phys. }{\bf #1} (#2) #3}

\def\imath#1#2#3{{\it Invent math }{\bf #1} (#2) #3}

\def\jpa#1#2#3{{\it J.\ Phys.\ }{\bf A #1} (#2) #3}

\def\ahpa#1#2#3{{\it Ann.\ I.~H.~Poincar\'e }{\bf A #1} (#2) #3}
\def\cqg#1#2#3{{\it Class.\ Quant.\ Grav.\ }{\bf #1} (#2) #3}

\begin{document}
%
%
{\pagestyle{empty}
\parskip 0in
\

\vfill
\begin{center}
{\LARGE Random Disks of Constant Curvature:}

\bigskip

{\LARGE the Lattice Story}



\vspace{0.4in}


Frank F{\scshape errari}

\medskip
{\it Service de Physique Th\'eorique et Math\'ematique\\
Universit\'e Libre de Bruxelles (ULB) and International Solvay Institutes\\
Campus de la Plaine, CP 231, B-1050 Bruxelles, Belgique}

%

\smallskip
{\tt frank.ferrari@ulb.be}
\end{center}
\vfill\noindent

We lay the groundwork for a UV-complete formulation of the Euclidean Jackiw-Teitelboim two-dimensional models of quantum gravity when the boundary lengths are finite, emphasizing the discretized approach. The picture that emerges is qua\-li\-ta\-ti\-ve\-ly new. For the disk topology, the problem reduces to counting so-called self-overlapping curves, that are closed loops that bound a distorted disk, with an appropriate multiplicity. We build a matrix model that does the correct counting. The theories in negative, zero and positive curvatures have the same UV description but drastically different macroscopic properties. The Schwarzian theory emerges in the limit of very large and negative cosmological constant in the negative curvature model, as an effective theory valid on distance scales much larger than the curvature length scale. In positive curvature, we argue that large geometries are ubiquitous and that the theory exists only for positive cosmological constant. Our discussion is pedagogical and includes a review of several relevant topics.

\vfill

\medskip
%
\begin{flushleft}
\today
\end{flushleft}
\newpage\pagestyle{plain}
\baselineskip 16pt
\setcounter{footnote}{0}

}

\tableofcontents

\section{Introduction}

\subsection{Two-dimensional quantum gravity: from Liouville...}

The starting point of our work is to postulate that a Euclidean theory of two-dimensional quantum gravity must be a diffeomorphism invariant theory of random metrics on a two-dimensional manifold $\mathscr M_{h,b}$ of fixed topology, with genus $h$ and $b$ boundary components.\footnote{The study of Lorentzian signature metrics and/or the consideration of the sum over all topologies are of course essential in some applications and will be briefly discussed later.} A rigorous formulation of such theories can be obtained by following the Wilsonian paradigm, in the spirit of statistical physics. The models are first regularized by considering discretized surfaces made of flat polygons. The problem is then purely combinatorial: one enumerates all the possible ways to build a surface by gluing polygons together. This counting of ``maps'' was pioneered by Tutte in the 60s \cite{Tutte}. In 1978, following ideas first put forward by 't~Hooft \cite{tHooft}, Br\'ezin, Itzykson, Parisi and Zuber related the combinatorics of maps to the solution of formal matrix integrals in the large $N$ limit \cite{BIPZ}. This point of view greatly simplifies the counting problem.
The discretized theory being solved, one must then define a continuum limit in which the size of the polygons goes to zero. The existence of non-trivial and universal continuum limits relies on the existence of critical points \cite{Kazakov}, \cite{Dblescaleref}. 

Another approach, pioneered by Polyakov in 1981 \cite{Polyakov}, is to use a direct field theoretic formulation in the continuum, called Liouville theory. One uses conformal gauge and treats the conformal factor as an ordinary two-dimensional quantum field. Diffeomorphism invariance of the underlying quantum gravity theory implies that the resulting Liouville model must be a conformal field theory. This remains true even when gravity is coupled to an arbitrary, possibly non-conformal, matter quantum field theory. This approach has scored a number of impressive successes \cite{Liouvillesuccess1,Liouvillesuccess2}. When comparable, the results obtained in the continuum Liouville approach and in the discretized approach beautifully match. This subject has become so mature that it is by now a well-established field in probability theory, see e.g.\ \cite{mathLiouville} and references therein.

A ``topological'' version of two-dimensional quantum gravity can also be studied. Instead of considering the space of all metrics on a given surface $\mathscr M_{h,b}$, one restricts oneself to metrics of constant negative curvature and requires the boundaries to be geodesics of given lengths. The resulting moduli spaces of metrics are finite dimensional. The associated partition functions are thus given by ordinary integrals that compute the volumes of these moduli spaces with respect to the Weil-Petersson measure, as a function of the boundary lengths. These volumes can be computed by using the Mirzakhani recursion relations \cite{Mirzakhani}, making the link with previous famous work by Witten and Kontsevitch \cite{Schlimittopo}. The Mirzakhani recursion relations are themselves equivalent to the so-called topological recursion relations in matrix models \cite{Eynardtoporec}. 

Nice reviews and relevant references on these classic results may be found in \cite{gravityreviews}. 

\subsection{...to Jackiw-Teitelboim}

In recent years, a version of two-dimensional quantum gravity, called Jackiw-Teitel\-boim quantum gravity \cite{JTpapers}, or JT gravity, for short, has been extensively studied. JT gravity lies somewhere between topological gravity and Liouville gravity. One fixes the bulk scalar curvature and the boundary lengths to be constants, as in the case of topological gravity, but one waives the constraint that the boundaries must be geodesics.\footnote{An important consequence of our work will be to show that the boundaries in JT gravity are actually fractals with a non-trivial Hausdorff dimension. The boundary lengths that are fixed are thus ``quantum lengths,'' analogous to the diffusion time for Brownian paths, as will be explained later.} The boundary conditions are therefore similar to the FZZT branes in Liouville gravity \cite{FZZT}.\footnote{In the original FZZT construction, one fixes the boundary cosmological constant. Fixing the boundary length is obtained by a Laplace transform. The important point is that no other condition is imposed on the boundary.} JT gravity comes in three versions, depending on whether the bulk curvature is chosen to be negative, zero or positive. Most of the literature focuses on the negative curvature case, because of its direct relevance for holography. The three versions are identical at distance scales much smaller than the bulk curvature length scale: they have the same UV-completion. The UV-complete formulation of the models can thus be studied by focusing first on the simplest case of the zero curvature theory, in which all the essential ideas can be introduced in a somewhat simplified context. The cases of negative and positive curvature will also be discussed, in a second stage. At long distance scales, we shall explain that the three models display drastically different physics.

In order to avoid a confusion in the terminology, let us stress that for us the cosmological constant $\La$ is defined to be the constant multiplying the bulk area term in the action.\footnote{For the disk topology, and when the non-zero curvature models are formulated in terms of a dilaton field, as is traditional, $\La$ is proportional to the boundary value of the dilaton.}   In JT gravity, it thus has no direct relation with the sign of the bulk curvature, which is fixed independently. In particular, the negative curvature model is defined for any value of the cosmological constant, positive or negative, the near-hyperbolic limit most relevant for holography corresponding to $\La\rightarrow -\infty$. On the other hand, we shall argue that the positive curvature model is defined only for positive cosmological constant, $\La\geq 0$, whereas the zero curvature model is expected to exist as long as $\La\geq \La_{\text c}$, for a strictly negative critical cosmological constant $\La_{\text c}$.

From the physics point of view, JT gravity may seem to be too simple. Indeed, since the bulk curvature is fixed, the bulk geometry cannot fluctuate and thus the local bulk physics looks non-gravitational. Consequently, on closed manifolds $\mathscr M_{h,0}$, JT gravity is equivalent to topological gravity. 

However, the situation turns out to be much more interesting on manifolds with boundaries. The core of our work will be to show, following \cite{ferrari}, that the situation is in fact even more interesting and far more subtle than the existing literature suggests.

For instance, in the case of the disk topology, the space of metrics is directly related to the space of immersions of the source disk into a canonical target space of constant curvature, which is either the hyperbolic space $\Htwo$, the Euclidean plane $\Etwo$ or the two-sphere $\Stwo$, depending on whether we consider the case of negative, zero or positive curvature. The target space geometry itself is totally rigid, which translates mathematically the fact that the local bulk physics looks non-gravitational. However, there is a very large number of possible immersions. By pullback, this yields an infinite dimensional space of metrics on the source disk. Intuitively, the degrees of freedom of the theory seem to be encoded in the shape of the boundary of the immersed disk into the fixed canonical target spaces. One of the main goal of the present paper is to explain in detail the fluctuating boundary picture, from first principles. This will reveal some truly surprising and unexpected features. In particular, the standard lore according to which there is a one-to-one correspondence between the degrees of freedom of the theory and the shape of the boundary modulo the action of the isometries of the target space will turn out to be incorrect. We shall also see that typical immersed boundaries have many self-intersections and are fractal curves.

The vast majority of the literature on JT gravity has focused on the case of negative curvature in a particular limit for which the lengths of the boundaries and the absolute value of the negative bulk cosmological constant go to infinity, their ratios being kept fixed in the bulk curvature units. The intuition behind this limit is that it corresponds to taking the infinite UV cut-off limit in  the dual boundary description. This is often referred to as the ``near-hyperbolic limit'' or the ``Schwarzian limit.''\footnote{The terminology ``near AdS'' limit is also frequently encountered, but we will use the name ``hyperbolic space'' for the space of constant negative curvature in Euclidean signature and reserve the term ``anti-de Sitter'' for the case of Lorentzian signature.} The conjectured dominant configurations in the limit correspond to smoothly wiggling boundaries that fluctuate mildly on distance scales much larger than the curvature length scale. Mathematically, they are described by the so-called ``reparameterization ansatz'' for which the JT gravity path integral reduces to an integral over the group of diffeomorphisms of the circle governed by the Schwarzian action. The resulting model is exactly solvable \cite{SWloca}. Despite its simplicity, it has demonstrated its relevance in shedding light on some non-trivial questions in quantum gravity. It provides a holographic description \cite{holomod1,holomod2,holomod3,holomod4} of strongly coupled quantum mechanical systems \cite{Kitaev,tensors} that are relevant to describe near-extremal black holes, including, for instance, the process of black hole formation \cite{ferphases} or the physics of traversable wormholes \cite{MaldacenaWorm}. The Schwarzian description has allowed to address the black hole information paradox in a controlled set-up \cite{islands}. At the same time, some unexpected subtleties with the holographic framework itself have emerged, in relation with the non-factorization of amplitudes. In particular, the Schwarzian framework is not consistent with a quantum mechanical boundary description, but requires an interpretation in terms of an ensemble average over quantum theories \cite{SSS}. 

Our results imply that the Schwarzian theory is an effective theory, valid on distance scales much greater than the curvature length scale $L$. It emerges as a correct IR description of the negative curvature JT gravity model in the limit $\La\rightarrow -\infty$. On distances of the order or smaller than $L$, the structure of the boundaries become much more complicated than the smoothly wiggling boundaries described by the Schwarzian. In particular, the boundaries are self-intersecting fractal curves, and in the microscopic description of the models the length parameter is a ``quantum'' length $\beta_{\text q}$ that has an anomalous dimension. The usual smooth boundary length $\ell$ characterizing the configurations included in the Schwarzian description is itself an effective IR parameter that can be expressed in terms of the microscopic parameters $\beta_{\text q}$, $\La$ and $L$ in the $\La\rightarrow -\infty$ limit.

\subsection{Microscopic and finite cut-off Jackiw-Teitelboim}

The primary motivation for the work presented in this article and the companion papers \cite{ferrari,ferraJTconfgauge} is the observation that, despite the important applications mentioned above and the fact that the relevance of the model is widely acknowledged, the theoretical foundations of JT quantum gravity have been very little studied and are poorly understood until now. In particular, the status of the subject is far from the one achieved by Liouville gravity. 

Our goal is to provide some of the keys to a rigorous study of the theory, starting from first principles, both in the discretized and the continuous approaches. Reference \cite{ferrari} provides a synthetic expos\'e of the main findings. The present article focuses predominantly, but not exclusively, on the development of the discretized approach. The continuous approach, using conformal gauge, is further studied in the upcoming paper \cite{ferraJTconfgauge}. Ideas from the continuum conformal gauge formalism were use in \cite{Loopcalc} to study the JT gravity models coupled to a conformal field theory of central charge $c$, in the semi-classical limit $c\rightarrow -\infty$, at the one-loop order.

An important point is that we consider Jackiw-Teitelboim gravity for finite values of the microscopic parameters: the curvature length scale $L$, the bulk cosmological constant $\La$ and the quantum boundary lengths $\beta_{\text q,i}$.\footnote{Of course, in the non-zero curvature models, one may set the units so that $L=1$. It is sometimes convenient, for the physical discussion, to keep $L$ explicitly.} As we have already mentioned, the boundaries in JT gravity turn out to be fractal.\footnote{The precise statement is that typical metrics will have fractal boundaries, with probability one. Metrics with smooth boundaries of course exist, but form a subset of measure zero in the space of metrics.} Thus, they do not have smooth boundary length parameters. Such a notion only emerges on long distance scales in the negative curvature model, in the Schwarzian limit $\La\rightarrow -\infty$.

In the negative curvature theory, working for finite values of the parameters is interpreted as introducing a finite UV cut-off in the holographic dual. The Schwarzian limit $\La\rightarrow -\infty$, in which very large, nearly-hyperbolic geometries dominate\footnote{The precise definition of the near-hyperbolic geometries, described by the so-called reparameterization ansatz, is reviewed in section \ref{repaSec}.}, is thus a large UV cut-off limit from the boundary perspective. This interpretation comes from the usual UV/IR relation in holography. However, it is important to stress that, unlike in higher dimensions, the strict limit of infinite cut-off is trivial when the bulk geometry is two-dimensional. This follows from the fact that, in two dimensions, there is a unique complete metric of constant negative curvature on the disk, corresponding to hyperbolic space itself. Therefore, the theory with the cut-off strictly removed has no degree of freedom.\footnote{In Lorentzian signature, a similar result is valid, see in particular \cite{finitecutoffAdS}.} From the boundary perspective, the impossibility of completely removing the UV cut-off is the  consequence of the fact that there is no standard quantum mechanical and non-trivial conformal field theory in one dimension.\footnote{A simple argument is as follows. A conformal primary operator $\mathscr O$ of conformal dimension $\Delta>0$ must have a two-point function $\langle\mathscr O(t)\mathscr O(0)\rangle$ proportional to $1/|t|^{2\Delta}$, which thus diverges when $t\rightarrow 0$. But in quantum mechanics, there is no singularity when one multiplies operators at the same point and thus no divergence in the two-point functions when $t\rightarrow 0$.} One can find theories with conformally invariant IR fixed points, as the SYK model, but the IR physics can never be completely decoupled from the UV, which is not conformally invariant. In other words, conformal invariance is only approximate and is broken above some UV scale that is impossible to get rid off.

The lesson from this discussion is that it is essential to formulate the theory for finite values of the parameters. The distance scales of the order of the curvature length scale are  not decoupled and the formulation must be UV-complete from the bulk perspective. 

\subsection{UV-complete Jackiw-Teitelboim quantum gravity}

On which basis can we construct a UV-complete formulation of JT gravity? Is there a unique ``consistent'' UV completion? Why can't we simply assume that the SYK model, for example,  provides a satisfactory UV completion?

The axiom on which our work is based is that, in general, Euclidean two-di\-men\-sio\-nal quantum gravity is a diffeomorphism invariant theory of random metrics. In particular, Euclidean Jackiw-Teitelboim quantum gravity is a diffeomorphism invariant theory of random constant curvature metrics. Note that in dimensions strictly greater than two, this minimalist metric point of view may very well be too naive, but in two dimensions, there is no reason to believe this. On the contrary, the classical results on Liouville theory and our own results on JT gravity offer a rather compelling physical picture and a rigorous UV-complete formulation of quantum gravity in two dimensions.

The metric axiom rules out the SYK model, or other similar models, since these are not models of random metrics. There is another independent argument against the idea that SYK could be a satisfactory UV completion of JT gravity, based on the following observations. First, the UV completion must be the same for all three versions of JT, negative, zero or positive curvature, since the curvature length scale is  irrelevant at very short distances. Second, the Schwarzian description emerges as a long-distance effective theory only in the case of the negative curvature theory. As will be made clear in the following, the other models have a very different long-distance behaviour. Third, the existence of the IR Schwarzian description is the basic reason, if not the only reason, why the SYK model has been studied in relation with quantum gravity. It seems highly improbable that SYK could provide any insight into the UV structure of the zero or positive curvature models, in which the Schwarzian plays no role. Putting these three observations together, it seems extremely unlikely and unnatural that SYK could provide a UV completion of JT gravity. We shall argue that the Schwarzian theory is a universal long distance description of many models of closed random loops in hyperbolic space, independently of any relation with quantum gravity.

A completely different, non-minimalist but physically well-grounded route to find UV completions of JT gravity is to consider that the theory is the IR limit of a higher dimensional quantum gravity theory. This is natural when one studies the near-horizon limit of higher dimensional near-extremal black holes. The higher dimensional picture also suggests to consider general dilaton gravity models in two dimensions, for which the bulk curvature is not fixed. From this point of view, JT gravity corresponds to a linear dilaton potential, in which case the dilaton plays the role of a Lagrange multiplier field imposing the constant curvature constraint. The higher dimensional point of view is certainly worth exploring further, but this is far beyond the scope of the present paper. In any case, building a rigorous UV completion along these lines would require a detailed understanding of the higher dimensional quantum gravity models, at a level which does not seem to be available at the moment. We believe that an interesting guiding principle to do this could be to ensure that a matching with the purely two-dimensional UV-complete point of view developed in our work should be made.

Even if we restrict ourselves strictly to dimension two and we follow the metric axiom, is the UV completion unique? Formally, the metric axiom implies that the theory is defined by a path integral over constant curvature metric
\be\label{formalpath} \int\! D g\,\delta \bigl(R[g]-2\eta/L^{2}\bigr)\, e^{-\frac{\La}{16\pi} A[g]}\bigl(\cdots\bigr)\, ,\ee
where $D g$ is a formal path integral measure over metrics, the functional $\delta$-function imposes the constraint of constant curvature, with $\eta=\pm 1$ or $0$ depending on the version of JT one considers, $A[g]$ is the area associated with the metric $g$ and the dots indicate possible insertions of diffeomorphism invariant operators. The question of uniqueness  amounts to asking whether there are several inequivalent ways to make sense of the formal path integral \eqref{formalpath}, or, equivalently, if there are several distinct measures, in the rigorous mathematical sense, on the space of constant curvature metrics modulo diffeomorphisms. The answer to this question is yes. An obvious way to understand that there are many interesting measures is that the quantum gravity theory could be coupled to a non-trivial QFT, and integrating out the QFT degrees of freedom then produces a new interesting measure. Measures that probably cannot be interpreted in this simple way will also be constructed in \cite{ferraJTconfgauge}. To achieve uniqueness, one thus needs an additional guiding principle on top of the metric axiom. There are in fact at least three possible angles of attack.

One point of view is to start from the real-time quantization and the requirement that the Euclidean theory computes correctly suitably defined observables in the Hilbert space framework. In spirit, this approach is similar to the logic used in constructing the Euclidean path integral in quantum mechanics, which reproduces correctly the quantum mechanical predictions if and only if one uses the Wiener measure of the set of paths, which are thus Brownian almost surely. However, the rigorous real-time quantization of JT gravity remains an outstanding open problem at the microscopic, finite volume, level; see \cite{realJTpaper} for an interesting approach in the near-AdS framework. We hope that our work on the Euclidean formulation, in the present paper and in \cite{ferraJTconfgauge}, will actually provides important clues to achieve it in the future. 

A second point of view is to assume that the measure $D g$ in the path integral \eqref{formalpath} is the usual ultralocal DeWitt measure on the space of metrics and proceed from there. This is non-rigorous, since ultralocal measures are only formal objects, but experience with Liouville theory indicates that this approach is valuable in providing strong clues at to what the correct, rigorously defined measure is. We shall use this point of view in the companion paper \cite{ferraJTconfgauge}.

A third point of view, that we shall follow in the present paper, is to use a lattice formulation, discretizing the space of constant curvature metrics and then taking a continuum limit. This approach provides a compelling physical picture and is unique, in the sense that the continuum limit of the lattice models is universal, independent of the details of the discretization procedure. In the continuum, it yields a privileged measure on the space of constant curvature metrics. We conjecture that this measure is the same as the one obtained through the formal procedure starting from the ultralocal DeWitt measure \cite{ferraJTconfgauge}, and is also the correct measure for making the link with the real-time quantization of the models.

\subsection{Comparison with the previous literature}

To the best of our knowledge, possible microscopic, UV-complete, definitions of JT gravity have been discussed in only three papers, all focusing on the disk topology. 

i) In \cite{KitaevSuh}, the authors take seriously the idea that the fluctuating boundary behaves as if it were the trajectory of a ``boundary particle.'' The path integral is then defined in the usual quantum mechanical way, using a quadratic action which describes the continuum limit of a closed random walk in target space, which is hyperbolic space for the negative curvature theory studied in \cite{KitaevSuh}. This point of view has been used to investigate various properties of the models, see e.g.\ \cite{PWAlgebra}. 

ii) In \cite{StanfordSAP}, the authors use a very different starting point. They assume that, instead of arbitrary closed random walks, the fluctuating boundaries must be self-avoiding loops, that is to say, closed curves with no self-intersection. This is a highly non-trivial model of random paths, which has been much studied in the literature because it has many applications in statistical physics and in mathematics.

iii) The idea to consider that the boundary must be a non self-interesting curve was also used in \cite{VerlindeSAP}, albeit with different results than in \cite{StanfordSAP}, but making the link with an interesting ``bottom-up'' approach in which the finite cut-off JT theory is postulated to correspond to a $T\bar T$ irrelevant deformation of the Schwarzian infinite cut-off description \cite{ttbar}.

\emph{The picture that emerges from the investigations presented below is qualitatively and quantitatively different from these earlier proposals.} 

As we shall explain in detail below, a description in terms of random walks, as in \cite{KitaevSuh}, does not satisfy the metric axiom. The model, however, is quite interesting and does share some important qualitative features with the UV-complete formulation that we propose. The approaches in \cite{StanfordSAP} and \cite{VerlindeSAP}, on the other hand, do describe theories of random metrics, since one can easily show that to any self-avoiding loop is associated a unique metric on the disk. But the set of metrics so obtained is a subset of measure zero with respect to the relevant measure on the space of constant curvature metrics. Models based on self-avoiding loops may thus be seen as some kind of minisuperspace approximations.

As already mentioned, and in sharp contrast with the existing approaches based on the fluctuating boundary picture, one of the most surprising consequence of our analysis is that \emph{the degrees of freedom of the theory are not in one-to-one correspondence with the shape of the boundary.} This is the consequence of the fact that to a given boundary curve may be associated several distinct metrics on the disk, that is to say, metrics that are not related by a disk diffeomorphism and are thus truly physically inequivalent. This fact alone shows that many ideas about the model must be rethought.

\subsection{Plan of the paper}

Section \ref{diskviewSec} is devoted to a detailed discussion of the space of constant curvature metrics on the disk. In Sections \ref{dissurfrevSec} and \ref{dismetSec}, a microscopic, discretized description, is given in the flat case, using surfaces built by gluing polygons. Generating functions are defined in Section \ref{Genfun1Sec}. A continuum description in terms of immersions, treating in parallel all three cases of zero, positive and negative curvatures, is given in Section \ref{immersionSec}. Finally, in Section \ref{earlySec}, we propose an overview of earlier approaches, in order to make the paper as self-contained as possible and also to be able to better explain and highlight the novelties brought about by our own contribution. We discuss the reparameterization ansatz and the associated notion of ``wiggling boundary,'' the Schwarzian theory, the edge mode picture and the existing approaches to the finite cut-off models.

In Section \ref{bdviewSec}, we develop the description of the model in terms of the boundary curves. The fact that JT quantum gravity can be approached from two very different perspectives, either based on random metrics or focusing on the boundary fluctuations, is one of its most interesting property. The analysis of existing models indeed suggest that theories of random metrics, e.g.\ Liouville theory, and theories of random paths, e.g.\ self-avoiding paths or Brownian paths, have qualitatively distinct features and describe very different physics. After a general introduction in Section \ref{curvgeniSec}, we explain in \ref{bdcodeSec} that the correct model for the fluctuating boundary in JT is a model of \emph{random self-overlapping curves}. General properties of self-overlapping curves are explained in \ref{SOPpropSec}. In \ref{WhitneySec} and \ref{overlapSec} we discuss simple criteria that self-overlapping curves must satisfy, based on the Whitney index and the notion of number of overlaps. A more comprehensive analysis, in terms of Blank cuts and Blank words, in given in \ref{BlankSec}. In Section \ref{MilnorSec}, we introduce the fundamental notion of \emph{multiplicity}. We illustrate this notion on the simplest example of the Milnor curve. The multiplicity of a self-overlapping curve counts the number of distinct disks it bounds or equivalently the number of physically distinguishable, that is to say, diffeomorphism inequivalent, constant curvature metrics associated with it. The fact that the multiplicity can be strictly greater than one implies that the degrees of freedom of JT quantum gravity are not in one-to-one correspondence with the shape of the boundary. In Section \ref{Genfun2Sec}, we define the precise model of random self-avoiding polygons that is equivalent to the model of discretized random constant curvature metrics discussed in Section \ref{diskviewSec}. We also mention some possible generalizations for which self-overlapping curves can be open and we very briefly discuss the case of the annulus or higher topologies. Finally, we give some preliminary information on the discretized formulation in negative and positive curvatures, highlighting the subtleties and differences with the flat case.

The goal of Section \ref{continuumSec} is twofold. First, we provide a pedagogical review of some ideas that play a central role in the study of known models of random metrics (the Liouville theory, Section \ref{RevLSec}) and of random paths (Brownian, self-avoiding, convex) in \ref{randompathrev}. Our aim is to introduce some useful concepts and technical tools and also to highlight the similarities, but also some crucial differences between random metrics and random paths in the traditional set-ups, having in mind that in the case of JT the random metrics and random paths formulations are strictly equivalent. We then initiate the study of the self-overlapping polygon model in Section \ref{SOPJTSec}, both in the case of the uniform measure and in the case of the non-uniform measure taking into account the multiplicity, which is relevant for JT. We describe some interesting observables that are specific to these models and explain some conjectures on some critical exponents and the asymptotic properties of the area distribution law.

In Section \ref{mmSec}, we introduce a Hermitian matrix model that does the appropriate counting of metrics in flat JT gravity, its large $N$ free energy coinciding with the generating functions introduced in Sections \ref{diskviewSec} and \ref{bdviewSec}. The model is based on the idea of dually weighted graphs, that we review in \ref{duallywSec}. As explained in \ref{nogoSec} and \ref{VectorSec}, some modifications with respect to the usual formulation must be introduced in order to be able to describe JT gravity. In particular, we construct the appropriate JT gravity macroscopic loop operators and provide a matrix model formula for the multiplicity of a self-overlapping polygon.

Section \ref{nonzeroRSec} focuses on the cases of non-zero curvature. At short distance scales, the positive, negative and zero curvature models coincide, but at long distances they display drastically different physics. The large distance behaviours can be characterized by the asymptotics of the area probability densities, for which we present conjectures. For the negative curvature model, we explain the emergence of the Schwarzian description at long distances and make the link with the ballistic behaviour of random paths models in hyperbolic space. For the positive curvature model, we argue that large geometries will be ubiquitous due to the absence of an isoperimetric inequality. At the quantum level, this implies that the model exists only for positive or zero cosmological constant. We also discuss the lattice formulation of the non-zero curvature models and postulate a relation with a suitable analytic continuation of the matrix model discussed in Section \ref{mmSec}.

Finally, in Section \ref{OutSec1}, we briefly summarize some of the main results of our work. We then discuss in \ref{OutSec2} the consequences for the study of JT gravity on manifolds of arbitrary topology. We explain that the standard approach by Saad, Shenker and Stanford, which applies in the Schwarzian limit, cannot be generalized in any straightforward way to the case of the microscopic, finite cut-off description. We close the paper in Section \ref{OutSec3} by highlighting some interesting open questions that we would like to explore in future research project.

\section{\label{diskviewSec}Constant curvature metrics on the disk}

We present a construction of the constant curvature metrics on the disk, mainly focusing on the discretized version in the flat case, but also discussing the general case in the continuum. The discretized formulation is central, since it provides the basis for a rigorous microscopic formulation of the JT quantum gravity theory. We put our results into perspective by comparing them with existing approaches in the literature.

\subsection{\label{dissurfrevSec}Discretized surfaces and discretized metrics}

Let us start by briefly reviewing some well-known facts about the discretization of two-dimensional surfaces and metrics. 

A discretized surface is built by gluing polygons along their edges. The construction is purely combinatorial in nature: the data are the set of polygons and the pairing (gluing) of the edges. These same data define a matrix model fat Feynman graph as well. To see this, one may use one of two equivalent points of view. In the first, ``direct'' point of view, the $p$-gons correspond to faces of length $p$ in the matrix model graph and the polygon vertices and edges correspond to the Feynman graph vertices and edges respectively. In the second, ``dual'' point of view, the $p$-gons correspond to vertices of degree $p$ of the Feynman graph and a gluing is associated to an edge joining the vertices associated with the glued polygons. The vertices of the direct point of view are mapped to the faces of the dual point of view and vice-versa.

Surfaces with $k$ boundary components can be described as closed surfaces with $k$ marked faces that may have an arbitrary length, the edges belonging to the marked faces being interpreted as forming the surface boundary. The marked faces may be removed altogether. The Euler identity reads
\be\label{Euler} F-E+V = 2 - 2h -b\, ,\ee
where $F$ is the number of unmarked faces, $E$ the number of edges, $V$ the number of vertices, $h$ the genus and $b$ the number of boundary components (or marked faces). 

Assuming that the unmarked faces are made of $p$-gons for a fixed value of $p$,\footnote{The most general case, with different types of polygons, can be treated similarly, but we don't need it for our purposes.} and noting $r_{1},\ldots r_{b}$ the lengths (number of edges) of the boundaries, we have
\be\label{pgonid1}
2 E = p F+\sum_{i=1}^{b}r_{i}\, .\ee
Using Euler's \eqref{Euler}, we get
\be\label{pgonid2} 2(2-2h-b)+\sum_{i=1}^{b}r_{i} = (2-p)F+2V\, .\ee

A metric on a discretized surface is specified by introducing a length scale $\ell_{0}$. One decides that all the polygons building up the surface are flat and regular, with edge length $\ell_{0}$. Moreover, the gluing along the edges is smooth and flat. The cut-off $\ell_{0}$ may be interpreted as the Planck length. In this construction, the curvature is entirely localized on the vertices. The Ricci curvature in the vicinity of a bulk vertex (defined to be a vertex that does not belong to the boundary) of degree $q$, at which $q$ $p$-gons are glued together, reads
\be\label{localR} R(x) = R_{p,q}\frac{\delta(x-x_{V})}{\sqrt{g(x)}}\,\cvp\ee
where $x$ is a local coordinate system around the vertex situated at $x=x_{V}$, $g$ the determinant of the metric and 
\be\label{Rpqform} R_{p,q} = \frac{2\pi}{p}\bigl( 4-(p-2)(q-2)\bigl)\, .\ee
The extrinsic curvature in the vicinity of a boundary vertex of degree $q$ is similarly
\be\label{localK} k(s) = k_{p,q}\,\delta(s-s_{V})\, ,\ee
where $s$ is the curvilinear coordinate along the boundary and 
\be\label{Kpqform} k_{p,q} = \frac{\pi}{p}\bigl(2-(p-2)(q-2)\bigl)\, .\ee

Let $V_{q}$ and $v_{q}$ be the numbers of vertices of degree $q$ in the bulk and on the boundary, respectively. We have (also using \eqref{pgonid1})
\be\label{EVrel} 2 E = \sum_{q} q (V_{q} + v_{q}) = pF + \sum_{i=1}^{b}r_{i}\, .\ee
Moreover, the total number of boundary vertices must match the total boundary length,
\be\label{vbdcount} V_{\text b} = \sum_{q}v_{q} = \sum_{i=1}^{b}r_{i}\, .\ee
Using \eqref{Rpqform} and \eqref{Kpqform}, one can then check that the Gauss-Bonnet formula
\be\label{GaussBonnet} \int\!\d^{2}x\sqrt{g}\, R + 2\oint\!\d s\, k=4\pi (2-2h-b)\ee
is equivalent to the Euler formula \eqref{Euler}, or equivalently to \eqref{pgonid2}.

\subsection{\label{dismetSec}Discretized flat Jackiw-Teitelboim gravity}

A straightforward application of the ideas reviewed in the previous subsection allows to define a discretized version of Euclidean Jackiw-Teitelboim quantum gravity in the zero curvature case. Working this out explicitly and explaining some simple consequences is the subject of the present subsection.

The equation \eqref{Rpqform} shows that the constraint of zero curvature in the bulk implies that  $2p=(p-2)q$ for each bulk vertex. The solutions are $(p,q) = (3,6)$, $(4,4)$ or $(6,3)$, corresponding to the three regular tessellations of the plane, made of triangles, squares or hexagons. These three possibilities could be used in our construction and they are all expected to yield the same theory in the continuum limit. For illustrative purposes, we mainly focus on quadrangulations. Using hexagons may present some technical advantages that will be mentioned in Section \ref{mmSec}.

We choose to orient the boundaries in such a way that the outer normal is directed to the right, the surface itself being situated to the left. With this convention, a positive extrinsic curvature means that the boundary is making a left turn and a negative extrinsic curvature that it is making a right turn. Precisely, with quadrangulations, formula \eqref{Kpqform} shows that the extrinsic curvature at the boundary vertices is
\be\label{k4qform} k_{4,q} = 
\begin{cases}
\pi/2 & \text{if $q=2$ (left turn)},\\
0 & \text{if $q=3$ (straight line)},\\
-\pi/2  &\text{if $q=4$ (right turn)}.
\end{cases}
\ee
We see that vertices of degree two, three and four correspond to a left turn, going straight, or a right turn, respectively, and that the turns are made at right angle. We restrict ourselves to vertices of degree two, three or four on the boundary, since, intuitively, this is enough to obtain boundaries of arbitrary shapes in the continuum limit.\footnote{It may be sufficient to work with degrees two and four only. The case of degrees five or more will be briefly mentioned later.}

Let us denote by $V_{\text B}$ the number of bulk vertices and by $v_{\text L}$, $v_{\text S}$ and $v_{\text R}$ the number of vertices of degree two, three and four on the boundary. Eq.\ \eqref{pgonid1}, \eqref{EVrel} and \eqref{vbdcount} now take the form
\be\label{quadrelgraph} 2 E = 4 V_{\text B} + 2 v_{\text L} + 3 v_{\text S} + 4 v_{\text R} = 4 F + \sum_{i=1}^{b}r_{i} = 4 F + v_{\text L} + v_{\text S} + v_{\text R}\, .\ee
In particular
\be\label{vvv4} v_{\text L} + 2 v_{\text S} + 3 v_{\text R} = 4(F-V_{\text B})\ee
must be a multiple of four. Using \eqref{quadrelgraph}, Euler's identity \eqref{Euler} is found to be equivalent to 
\be\label{EulergenflatJT} v_{\text L}-v_{\text R} = 4(2-2h-b)\, .\ee
Taking into account \eqref{k4qform}, this is of course also equivalent to the zero-curvature Gauss-Bonnet formula $\oint k\d s = 2\pi(2-2g-b)$. Note that, together with \eqref{vvv4}, we get
\be\label{vSform} v_{\text{S}} = 2(F-V_{\text B}) - 2 v_{\text R} -2 (2-2h-b)\, .\ee
The number of straight vertices must be even.

On the disk, Eq.\ \eqref{EulergenflatJT} yields
\be\label{EulergenflatJTdisk} v_{\text L}-v_{\text R} = 4\, .\ee
This is a simple and useful constraint on the possible shapes of the boundary curve in JT gravity. We have derived it in the flat case, but we will see later that a version of this constraint is also valid in the positive and negative curvature cases; the difference is that in these latter cases it will not be a direct consequence of Gauss-Bonnet.

\begin{figure}
\centerline{\includegraphics[width=6in]{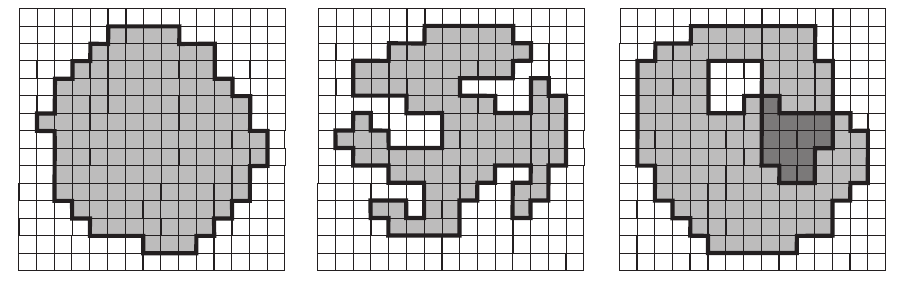}}
\caption{\label{disdiskbasicFig}Three examples of discretized disks drawn isometrically on a flat square lattice. The thick black lines represent the boundaries. The lattice tiles occupied by squares forming the disk, which are the faces of the associated fat graph in the matrix model direct representation, are shaded in grey. Dark grey tiles are occupied by two overlapping faces. Left inset: a disk with a smoothly wiggling boundary, akin to the ``reparameterization ansatz'' configurations much studied in the literature; center inset: a disk whose boundary is a typical self-avoiding polygon; right inset: a disk with a twofold overlapping region (shaded in dark grey) and a self-intersecting boundary.}
\end{figure}

What do discretized flat discs look like? We have depicted in Fig.\ \ref{disdiskbasicFig} three simple examples. All our drawings are made on an underlying square lattice. The basic building blocks of the lattice are called ``tiles.'' Each square face of a discretized disk is put on a lattice tile, building up a figure satisfying the constraints indicated above. The result is a graph for which all the faces, but the marked boundary, have length four, all the bulk vertices have degree four and all the boundary vertices have degree two, three or four. 

A crucial feature is that \emph{discretized disks may have overlaps}; see, for instance, the example depicted on the right inset of Fig.\ \ref{disdiskbasicFig}. As a consequence, \emph{the boundary curves may have self-intersections.} This is an absolutely fundamental feature, which is also relevant in the positive and negative curvature cases.

Intuitively, in the continuum limit, discretized disks are arbitrary \emph{distorted disks}, which are obtained by starting from a round disk and then allowing arbitrary smooth deformations, stretching and turning around. Under such deformations, some pieces of the original round disk may eventually cover other pieces, possibly multiple times. A deformation process, that yields a continuous configuration akin to the discretized disk depicted in the right inset of Fig.\ \ref{disdiskbasicFig}, is shown in Fig.\ \ref{deformFig}.

\begin{figure}
\centerline{\includegraphics[width=6in]{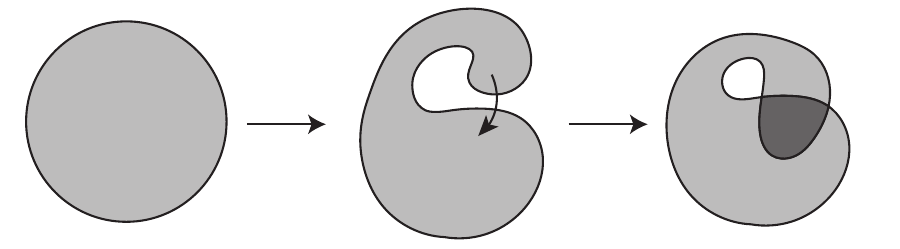}}
\caption{\label{deformFig}A distorted disk with an overlapping region (right inset), obtained by smoothly deforming a round disk (left inset). The deformation is a regular homotopy between the trivial embedding of the circle in Euclidean space and a more general immersion. It can also be interpreted, by pullback, as a smooth deformation of the Euclidean round metric on the disk to a more general flat metric.}
\end{figure}

Let us immediately dissipate a possible confusion. The graph associated with the disk is of course planar. As a consequence, it can be drawn on a plane in such a way that there is no overlap and, in particular, in such a way that the boundary has no self-intersection. This comes from a purely topological argument. The reason why we do get overlaps and self-intersections in our drawings is that we insist to depict the disks \emph{in an isometric way}, respecting the distances defined by the metric associated with the discretization, as explained in \ref{dissurfrevSec}. This is a very convenient and visual way to depict metrics of constant curvature, both in the discretized and in the continuum cases. This point of view will be further explored in Section \ref{immersionSec}. 

To close this subsection, we have depicted five forbidden configurations in Fig.\ \ref{forbidFig}. The first three insets represent cases with boundary vertices of orders strictly greater than four: order $q=5$, for which $k_{4,5}= -\pi$ (``U-turn''); $q=6$, for which $k_{4,6} = -3\pi/2$; and $q=7$, for which $k_{4,6} = -2\pi$. Allowing for such configurations could be interesting, since they may allow for higher critical points in the model. In the last two insets on the right are depicted two candidate boundary curves that turn out to be forbidden, because they do not bound distorted disks. It is indeed impossible to find discretized disks, as defined above, whose boundaries match these curves. The curve in the fourth inset has $v_{\text L}-v_{\text R}=8$ and can thus be ruled out by the simple criterion \eqref{EulergenflatJTdisk}. The curve in the rightmost inset satisfies the constraint $v_{\text L}-v_{\text R}=4$ but is also ruled out because, by checking the position of the would-be distorted disk with respect to the boundary, one realizes that it would have to fill up the entire plane. These simple examples hint at the fact that the set of allowed boundary curves is constrained in a rather non-trivial and non-local way. This will be discussed in much more detail in Section \ref{bdviewSec}.

\begin{figure}
\centerline{\includegraphics[width=6in]{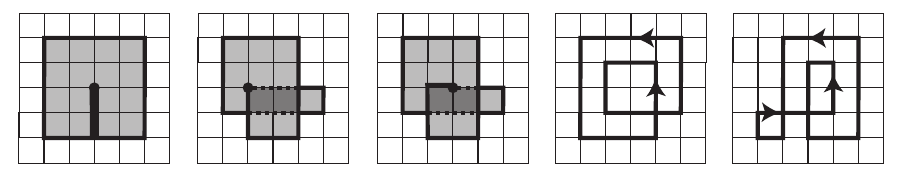}}
\caption{\label{forbidFig}Forbidden configurations corresponding to boundary vertices (black dots) of degrees 5, 6 and 7 (first, second and third insets respectively) and impossible boundary curves (right insets).}
\end{figure}

\subsection{\label{Genfun1Sec}Generating functions}

We denote by $W_{v,p}$ the number of distorted, flat, discretized disks with $F=p$ faces and $v=v_{\text L}+v_{\text S}+v_{\text R}$ boundary vertices. Note that \eqref{quadrelgraph} implies that $v$ is even\footnote{Actually, any closed path on a square lattice has even length, even paths that do not bound a distorted disks.} and thus we set
\be\label{vnrel} v = 2n\, .\ee
From \eqref{quadrelgraph} and \eqref{EulergenflatJTdisk}, the total number of edges and bulk vertices are
\be\label{nprel} E = 2p+n\, ,\quad V_{\text B} = p-n+1\, .\ee
%
To define $W_{2n,p}$ precisely, we identify two disks that are related by a lattice translation. From the point of view of the metric or gravitational interpretation, it would also be natural to identify configurations related by lattice rotations and reflexions, since these transformations clearly do not change the metric. However, for various reasons, this is not the most convenient choice to make, see in particular Section \ref{mmSec}. We thus decide to count as distinct non-identical configurations related by lattice rotations or reflexions. For generic configurations, and thus also in the continuum limit, this simply amounts to multiplying the generating functions by a factor of $4\times 2=8$. One may then rescale the generating functions by a global factor of $1/8$, but we will not do this since it is more natural to work with coefficients $W_{2n,p}$ that are integers.

\begin{figure}
\centerline{\includegraphics[width=6in]{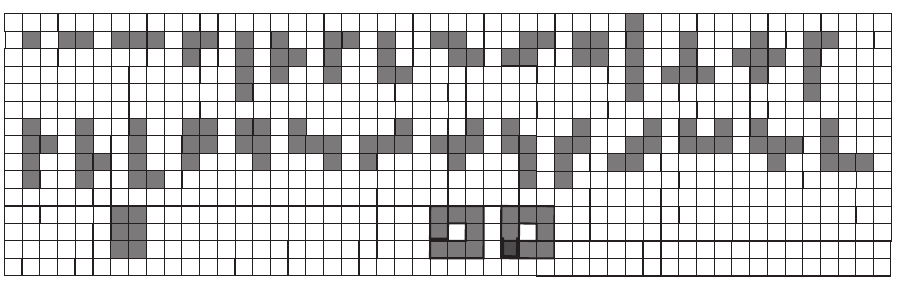}}
\caption{\label{figlowdisk}All the distorted disks, modulo lattice rotations, having five faces or less, or a boundary length 10 or less. In the lower part of the figure are also depicted the two simplest distorted disks whose boundaries are not self-avoiding polygons (contributing to the partition function at order $g^{18}t^{8}$) or that do have an overlap (contributing to the partition function at order $g^{20}t^{9}$).}
\end{figure}

All the distorted disks with low values of $n$ and $p$ are depicted in Fig.\ \ref{figlowdisk}. With our counting convention, the first configuration in the upper-left corner is counted one time, $W_{4,1}=1$, the second two times, $W_{6,2}=2$, the third two times and the fourth four times, $W_{8,3} = 6$, etc. The numbers $W_{2n,p}$ are conveniently encoded in the generating functions
\be\label{Zdef1} W(t,g) = \sum_{n,p\geq 1}W_{2n,p}t^{p}g^{2n}\, .\ee
One gets, from Fig.\ \ref{figlowdisk},
\be\label{tildeZnorot} W(t,g) = g^{4}t + 2 g^{6} t^{2} + 6 g^{8} t^{3} + (18g^{10}+g^{8})t^{4} + (8g^{10} + 55g^{12})t^{5} + 2g^{10}t^{6} + O(g^{12}t^{6})\, .\ee
This low order counting is not very interesting, because for low values of $n$ and $p$ the boundaries are standard self-avoiding polygons (SAPs) \cite{SAPcounting}. The counting starts to deviate from the SAP case at order $g^{18}t^{8}$, see Fig.\ \ref{figlowdisk}.

The parameters $t$ and $g$ are related to the bulk and boundary bare cosmological constants $\La_{0}$ and $\la_{0}$ by $t=e^{-\La_{0}\ell_{0}^{2}/(16\pi)}$ and $g=e^{-\la_{0}\ell_{0}}$, where $\ell_{0}$ is the Planck length.\footnote{We use the traditional convention in JT gravity of introducing a factor of $\frac{1}{16\pi}$ in the definition of the bulk cosmological constant.} The generating functions at fixed boundary length and fixed area are
\be\label{Zfixed} W_{2n}(t) = \sum_{p\geq 1}W_{2n,p}t^{p}\, ,\quad \tilde W_{p}(g) = \sum_{n\geq 1}W_{2n,p}g^{2n}\, .\ee
The full generating function \eqref{Zdef1} can be expressed as
\be\label{Zdef2} W(t,g) = \sum_{n\geq 1}W_{2n}(t)g^{2n} = \sum_{p\geq 1}\tilde W_{p}(g)t^{p}\, .\ee
For a given boundary length $2n$, the area is bounded above by $[n^{2}/4]$, which is $n^{2}/4$ or $(n^{2}-1)/4$ depending on whether $n$ is even or odd. This is the isoperimetric inequality on the square lattice. For a given area $p$, the largest possible boundary length is $2(p+1)$, which, by using \eqref{nprel}, is obtained when $V_{\text B}=0$, i.e.\ when the disk is tree-like, or branched polymer-like, with no bulk structure. This implies that the generating functions at fixed length and fixed area are polynomials, with
\be\label{degreeW} \deg W_{2n}(t) = [n^{2}/4]\, ,\quad \deg \tilde W_{p}(g) = 2 (p+1)\, .\ee
In particular,
\be\label{Z2npdsef} w_{2n} = W_{2n}(1) = \sum_{p\geq 1}W_{2n,p}\, ,\quad \tilde w_{p} = \tilde W_{p}(1) = \sum_{n\geq 1}W_{2n,p}\ee
are finite and count the total number of distorted disks with fixed boundary length $2n$ and fixed area $p$ respectively. From \eqref{tildeZnorot}, one finds
\be\label{wtildewlow}
\begin{split} & w_{2}=0\, ,\ w_{4}=1\, ,\ w_{6}=2\, ,\ w_{8} = 7\, ,\ w_{10} = 28\\
& \tilde w_{1}= 1\, ,\ \tilde w_{2} = 2\, ,\ \tilde w_{3}=6\, ,\ \tilde w_{4}=19\, ,\ \tilde w_{5}=63\, .
\end{split}\ee

The functions $W$, $W_{2n}$ or $\tilde W_{p}$ are a priori formal generating functions encoding the numbers $W_{2n,p}$. Their radii of convergence and the possibility to define a continuum limit will depend on the asymptotic behaviour at large $n$ and $p$. These aspects will be discussed in Section \ref{continuumSec}.

\subsection{\label{immersionSec}Constant curvature metrics on the disk in the continuum}

The above discussion has a natural extension in the continuum that we present in this subsection. The cases of zero, positive or negative curvatures are treated in parallel.

\subsubsection{Immersions and pull-back metrics}

Let us first recall the definition of an immersion and a few simple facts about the concept.
\begin{definition} An immersion $F:\mathscr M\rightarrow\mathscr N$ between two differentiable manifolds $\mathscr M$ and $\mathscr N$, $\dim\mathscr N\geq\dim\mathscr M$, is a differentiable function such that $\d F$ has maximum rank everywhere, that is to say, $\rank\d F_{|P} = \dim\mathscr M$.
\end{definition}
In local coordinate systems $(x^{\mu})$ on $\mathscr M$ and $(y^{a})$ on $\mathscr N$, an immersion $F$ is described by functions $y^{a}(x)$ such that the matrix $\partial y^{a}/\partial x^{\mu}$ has rank $\dim\mathscr M$ everywhere. 

We shall use this in the following way. If $F:\mathscr M\rightarrow\mathscr N$ is an immersion and $\mathscr N$ is endowed with a metric $\delta$, then there is a naturally associated metric $g = F^{*}\delta$ on $\mathscr M$, called the pull-back metric, given in local coordinates by
\be\label{pullbackdef} g_{\mu\nu}(P) = \frac{\partial y^{a}}{\partial x^{\mu}}\frac{\partial y^{b}}{\partial x^{\nu}}\delta_{ab}(F(P))\, .\ee
The fact that the metric tensor $g$ is positive-definite follows from the facts that $\delta$ is positive-definite and that $\partial y^{a}/\partial x^{\mu}$ does not have a zero eigenvalue. If we waive the condition that $F$ is an immersion, then the symmetric tensor $g$ defined by \eqref{pullbackdef} will still be positive, but it will be degenerate (non-invertible) at the points where $\d F$ is not invertible. Thus, it will not be a metric in the usual sense.

An immersion $F:\mathscr M \rightarrow\mathscr N$ is a local diffeomorphism between $\mathscr M$ and its image $F(\mathscr M)$. When an immersion is injective, one says that the immersion is an embedding. There is then a global identification between the spaces $\mathscr M$ and $F(\mathscr M)$. Of course, a generic immersion is not an embedding. 

The simplest case to consider is that of immersions $F:\Sone\rightarrow\mathbb R^{2}$ of the circle into the plane. The image $\im F=\gamma$ is a closed curve. For a general immersion $F$, $\gamma$ has self-intersections. The immersion is an embedding if and only if $\gamma$ has no self-intersection. Another case, which is very important for our purposes, is that of immersions $F:\disk\rightarrow\mathbb R^{2}$ from the disk to the plane. The discretized disks of Section \ref{dismetSec} are the discretized versions of the images $F(\disk)$ of  immersions $F$. The associated metrics are the pull-back metrics \eqref{pullbackdef}, where $\delta$ is the Euclidean metric. This idea yields an identification between the space of constant curvature metrics and appropriate quotients of the space of immersions, that we are going to work out below.

Consistently with the convention we have already used, we always denote by $\disk$ the disk manifold. This is the ``source'' space, on which the quantum gravity theory is defined. It can be described by the usual polar coordinates $(\rho,\theta)$, $0\leq\rho <1$, or complex coordinates $z = \rho e^{i\theta}$,
\be\label{diskdef} \disk = \bigl\{ z\in\mathbb C \bigm| \abs{z} < 1\bigr\}\, .\ee
The boundary is $\Sone=\partial\disk=\{z\in\mathbb C\, |\, \abs{z} =1\}$ and the closed disk is denoted by $\bar\disk = \disk\cup\partial\disk$.

\subsubsection{\label{zeroImmSec}The space of zero curvature metrics on the disk, I}

Let $\Etwo = (\mathbb R^{2},\delta)$ be the two dimensional Euclidean plane, where $\delta$ is the usual flat metric. In Cartesian coordinates $(x,y)$ and complex coordinates $w = x + i y$,
\be\label{deltaE} \delta = \d x^{2} + \d y^{2} = |\d w|^{2}\, .\ee
A flat two-dimensional disk $(\disk, g)$, where $g$ is a flat metric on $\disk$, is locally isometric to $\Etwo$. This is equivalent to saying that there exists an immersion
\be\label{FEucl} F:\disk\rightarrow\mathbb R^{2}\ee
such that
\be\label{pullBE} g=F^{*}\delta\, .\ee
Let us denote by $\imm(\disk,\mathbb R^{2})$ the space of immersions $F:\disk\rightarrow\mathbb R^{2}$. This is a bit imprecise, because we do not specify at this stage the boundary behaviour of the allowed immersions; the discussion of this important point is postponed to the next subsection.

We can now make precise the notion of ``distorted disk'' that we have already encountered above.
\begin{definition} A distorted disk is an immersion belonging to $\imm(\disk,\mathbb R^{2})$.
\end{definition}
On $\imm(\disk,\mathbb R^{2})$ act the groups $\Iso(\Etwo) = \text{ISO}(2) = \text{T}(2)\rtimes\text{O}(2)$ of isometries of the Euclidean plane, on the left; and the group $\diff_{+}(\disk)$ of orientiation-preserving diffeomorphisms of the disk, on the right.	We consider only orientation-preserving diffeomorphisms for convenience, but including complex conjugation $z\mapsto\bar z$ is of course straightforward. By definition, the group $\diff_{+}(\disk)$ includes boundary reparameterizations. At the infinitesimal level, it is generated by vector fields on $\disk$ that are tangent to the boundary at the boundary, as is standard for a manifold with boundaries. 

If we note $g_{F}=F^{*}\delta$ the metric associated with the immersion $F$ as in Eq.\ \eqref{pullBE} then, by definition of the isometry group, $g_{F} = g_{\psi\circ F}$ if $\psi\in\Iso(\Etwo)$. Moreover, if $\varphi\in\diff_{+}(\disk)$, then $g_{F}=\varphi^{*}g_{F\circ\varphi^{-1}}$, i.e.\ the metrics $g_{F}$ and $g_{F\circ\varphi^{-1}}$ are related by the usual action of the group of diffeomorphisms of $\disk$. They are thus identified. Overall, we have shown that the space of flat metrics on the disk is
\be\label{metzero} \met^{0}(\disk)  = \Iso(\Etwo)\backslash\imm(\disk,\mathbb R^{2})/\diff_{+}(\disk)\, .\ee

This can be further simplified by noting that, by acting with $\diff_{+}(\disk)$ together with the reflection $F\mapsto\bar F$, which is an element of $\Iso(\Etwo)$, we can always find a holomorphic representative $F$, with $F'$ never vanishing, for each equivalence class of immersions in the quotient space on the right-hand side of \eqref{metzero}. For a holomorphic immersion, the pull-back metric \eqref{pullBE} is
\be\label{pullmetholoflat} g_{F} = |F'(z)|^{2} |\d z|^{2}\, .\ee
This result is a special case of a more general theorem that ensures that any metric on $\disk$, not necessarily flat, can always be put in the so-called conformal gauge,
\be\label{conformalgauge} g = e^{2\Sigma}|\d z|^{2}\, ,\ee
for a conformal factor (``Liouville field'') $\Sigma$. The flatness condition is equivalent to the fact that $\Sigma$ is harmonic on the disk and is thus the real part of a holomorphic function $H$, $\Sigma = \re H$. If $F'=e^{H}$, 
\be\label{cgsigmaE} e^{\Sigma} = |F'(z)|\ee
and thus Eq.\ \eqref{conformalgauge} is equivalent to \eqref{pullmetholoflat}. Let us denote by $\himm(\disk)$ the set of holomorphic immersions $F:\disk\rightarrow\mathbb C$, i.e.\ the set of holomorphic functions on the disk such that $F'(z)\not =0$ for all $z\in\mathbb\disk$.\footnote{As for $\imm(\disk,\mathbb R^{2})$, we let unspecified the behaviour at the boundary. This is discussed in the next subsection.} On $\himm(\disk)$ act the group $\Iso_{+}(\Etwo) = \text{ISO}_{+}(2) = \text{T}(2)\rtimes\text{SO}(2)$ of orientation-preserving isometries of the Euclidean plane on the left and the group $\Aut(\disk)$ of holomorphic diffeomorphisms of the disk on the right. This group is the remnant of the full group of diffeomorphisms in conformal gauge \eqref{conformalgauge}. It is generated by the conformal Killing vectors on the disk and is isomorphic to $\PslR$. We thus have
\be\label{metzeroh} \met^{0}(\disk)  = \Iso_{+}(\Etwo)\backslash\himm(\disk)/\Aut(\disk) = 
\text{ISO}_{+}(2)\backslash\himm(\disk)/\PslR\, .\ee
The explicit actions of the groups on the disk and on the Euclidean plane are given by
\begin{align}\label{pslaction}& \PslR:\quad z' = e^{i\alpha}\frac{z-z_{0}}{1-\bar z_{0}z}\,\cvp\quad |z_{0}|<1\, ,\ (z,z')\in\bar\disk^{2}\, ,\\ 
\label{ISOact} & \text{ISO}_{+}(2):\quad w' = e^{i\alpha}(w-w_{0})\, ,\quad  w_{0}\in\mathbb C\, ,\ (w,w')\in\mathbb C^{2}=\Etwo\times\Etwo\,  .
\end{align}
Conformal gauge will be thoroughly studied in a companion paper \cite{ferraJTconfgauge}.

\noindent\emph{Remarks}:

\noindent i) Both the group $\diff_{+}(\disk)$ of diffeomorphisms of the disk and the group $\Iso(\Etwo)$ of isometries of the Euclidean plane must be considered as gauge symmetries. This statement is equivalent to the presentation \eqref{metzero} of the space of metrics as a quotient space. It is useful to understand that the origin of these two gauge symmetries is very different. On the one hand, the gauge symmetry $\diff_{+}(\disk)$ is consubstantial with our problem: it is the fundamental diffeomorphism gauge symmetry present in any gravity theory. On the other hand, the gauge symmetry $\Iso(\Etwo)$ appears only because we use immersions to describe the metrics. It is a redundancy that enters via the formula \eqref{pullBE}. The same can be said after the partial gauge fixing that yields \eqref{metzeroh}: the groups $\PslR$ and $\text{ISO}_{+}(2)$ are both gauge symmetries but $\PslR$ is always present in conformal gauge for quantum gravity on the disk whereas $\text{ISO}_{+}(2)$ appears because we parameterize the metrics in terms of holomorphic immersions.

\noindent ii) In the spirit of what we have done in Section \ref{dismetSec}, where configurations obtained by performing lattice rotations were counted as inequivalent, we could quotient only by the group of translations $\text{T}(2)$ instead of the full isometry group of $\Etwo$. This kind of modification is largely irrelevant. It simply introduces trivial factors of $2\pi$ in the definition of the partition function. The important point is to take into account the non-compact pieces of the gauge symmetries in order to work with well-defined path integrals.


\subsubsection{\label{zeroImmSecbis}The space of zero curvature metrics on the disk, II}

The discussion above does not address the problem of the boundary properties of the disk metrics. There are many mathematically natural possibilities and thus a large number of a priori interesting spaces of flat metrics on $\disk$. A precise definition of the Euclidean quantum JT gravity theory requires to single out one of these spaces and to define a probability measure on it.

A seemingly natural choice (which has been made implicitly in virtually all studies on JT gravity so far, see in particular the review in Section \ref{earlySec}) is to impose that the immersions $F$ extend continuously to the boundary $\partial\disk$. In conformal gauge, Eq.\ \eqref{conformalgauge}, this is equivalent to imposing that the boundary Liouville field $\sigma = \Sigma_{|\partial\disk}$ is continuous.\footnote{Note that the bulk Liouville field is always infinitely differentiable, because it must be a harmonic function.} With this choice, the boundary length
\be\label{bdlengthdef} \ell = \int_{0}^{2\pi}\!e^{\sigma}\d\theta\ee
is well-defined and one can consider the space $\met^{0}_{\ell}(\disk)$ of flat metrics on the disk with fixed boundary length $\ell$.

An important consequence of the analysis presented in our work, consistently with the discussion in \cite{ferrari}, is that this smoothness condition on the disk boundary is not well-founded from the point of view of the microscopic, UV-complete definition of the models.

To clearly appreciate the nature of the problem, it is useful to use a simple analogy with the standard functional integral formulation of quantum mechanics. In this case, one wishes to define a sum over ``paths.'' But what is exactly the space of paths one must consider? There are two physically motivated approaches to answer this question. The first is to construct a path integral that reproduces the results obtained from canonical quantization in the traditional operator formalism of quantum mechanics. The second is to start from a discretized model, in which the sum over path is a well-defined finite problem, and look for a critical point that allows to define a continuum limit. Both procedures yield the Wiener measure on the set of \emph{continuous} paths, describing Brownian motion. The typical Brownian paths turn out to be almost surely nowhere differentiable and have Hausdorff dimension two. If we had started from the seemingly natural but ad-hoc idea that the paths had to be differentiable, or rectifiable, we wouldn't have been able to obtain the correct theory.

We propose that exactly the same logic should be used in the case of JT gravity. One could try to set up a rigorous canonical quantization formalism in real time and then explore its consequences for the Euclidean path integral formulation. This is very interesting, but not in the spirit of the present paper. We favour the discretized formulation, as defined in Sections \ref{dismetSec} and \ref{Genfun1Sec}. The continuum limit should then yield precise predictions for the boundary properties of the typical flat disk metrics. Alternatively, one could start from a formal continuum approach, based on the formal ultralocal measure on the space of metrics, and try to make sense of it, see \cite{ferrari} and \cite{ferraJTconfgauge}. 

We conjecture that all these methods (canonical quantization, discrete formulation, continuous formulation) yield the same quantum gravity theory. The results announced in \cite{ferrari}, and confirmed  in \cite{Loopcalc} by studying a particular limit of JT gravity coupled to conformal matter, strongly suggest that the resulting disk metrics are highly singular on the disk boundary. In conformal gauge, the boundary Liouville field is predicted to be a random distribution, not a random function.\footnote{This is in sharp contrast with the traditional one-dimensional quantum mechanical variables, which are random continuous functions. This behaviour is closer to that of a quantum field in dimension two.} The resulting typical immersed boundary curves are fractal curves of Hausdorff dimension $d_{\text H}>1$.\footnote{For the pure JT gravity model, which corresponds to the discretized model discussed in Section \ref{dismetSec}, the Hausdorff dimension of the disk boundary is predicted to be two \cite{ferrari}. This is the same value as for Brownian paths, even though the disk boundary in JT is not Brownian; see the discussion in Sections \ref{bdviewSec} and \ref{continuumSec}.} In particular, the usual geometrical boundary length defined by \eqref{bdlengthdef} is infinite for such metrics, but it is possible to define a notion of renormalized length, or ``quantum'' length $\beta_{\text q}$.\footnote{We use the notation $\beta_{\text q}$ for the quantum length of the boundary, in order to avoid confusion with another parameter $\beta$, that we shall denote by $\beta_{\text S}$, which is traditionally used in the Schwarzian limit, see Eq.\ \eqref{Schlimit}.} The explicit definition of $\beta_{\text q}$ is given in Section \ref{continuumSec} in the discretized formulation and in \cite{ferrari,ferraJTconfgauge} in the continuum formulation. It is similar to the diffusion time in a Brownian process, which provides a notion of ``quantum length'' for a Brownian path. An important qualitative feature is that the quantum length does not have the dimension of a geometrical length. Instead, $\beta_{\text q}^{1/d_{\text H}}$ has the dimension of length. The fact that the renormalized length $\beta_{\text q}$ exists is a desirable feature for physics, since it allows to define canonical and microcanonical ensembles, see below.

A similar discussion applies to the spaces of constant positive and negative curvature metrics that we are going to discuss below. The relevant metrics and immersions are highly singular on the boundary, with the same UV properties as in the flat case.

\subsubsection{\label{metpmdisSec}The spaces of constant positive and negative curvature metrics on the disk}

The discussion of Section \ref{zeroImmSec} can be repeated, with straightforward modifications, to describe the spaces of constant positive and negative curvature metrics. The difference is that the Euclidean space, which plays the role of the canonical target space in the case of the flat metrics, is replaced by the round two-sphere $\Stwo$ for positive curvature and by the hyperbolic space $\Htwo$ for negative curvature. 

Let $\Stwo$ be the round two-sphere, with canonical constant positive curvature metric $R=+2$ given by
\be\label{deltaStwo} \delta^{+} =  \frac{4 |\d w|^{2}}{(1+|w^{2}|)^{2}}\ee
in the usual projective coordinates $w\in\bar{\mathbb C}$. Let $\Htwo$ be the hyperbolic space, with canonical constant negative curvature metric $R=-2$ given by
\be\label{deltaHtwo} \delta^{-} =  \frac{4 |\d w|^{2}}{(1-|w^{2}|)^{2}}\ee
in the usual Poincar\'e disk coordinates, $|w|<1$. Let us note that $\Htwo$ is topologically a disk, endowed with the metric $\delta^{-}$, so we can write $\Htwo = (\disk,\delta^{-})$. 

Let $\met^{\pm}(\disk)$ be the spaces of constant positive or negative curvature metrics, $R=\pm 2$, on $\disk$, with suitable boundary properties, as described in \ref{zeroImmSecbis}. All such metrics can be written as
\be\label{pullBS} g = F^{*}\delta^{\pm}\ee
for some immersion $F:\disk\rightarrow\Stwo$ or $F:\disk\rightarrow\Htwo$. Let $\imm(\disk,\Stwo)$ and $\imm(\disk,\Htwo)$ be the associated spaces of immersions of $\disk$ into $\Stwo$ or $\Htwo$. By following the same logic that yields Eq.\ \eqref{metzero} in the flat case, we obtain
\begin{align}\label{metp} & \met^{+}(\disk)  = \Iso(\Stwo)\backslash\imm(\disk,\Stwo)/\diff_{+}(\disk)\, ,\\\label{metm} & \met^{-}(\disk)  = \Iso(\Htwo)\backslash\imm(\disk,\Htwo)/\diff_{+}(\disk)\, .
\end{align}
Going to conformal gauge, Eq.\ \eqref{conformalgauge}, we can write the metrics in terms of holomorphic functions $F$,
\be\label{cgpmsigma} e^{\sigma} = \frac{2|F'(z)|}{1\pm|F(z)|^{2}}\quad \text{for}\quad R=\pm2\, .\ee
In the case of positive or negative curvature, the holomorphic functions are from $\disk$ to $\bar{\mathbb C}$ or from $\disk$ to $\disk$, respectively. This yields
\begin{align}\label{metph} & \met^{+}(\disk)  = \Iso_{+}(\Stwo)\backslash\himm(\disk,\bar{\mathbb C})/\Aut(\disk) = 
\text{SO}(3)\backslash\himm(\disk,\bar{\mathbb C})/\PslR\, ,\\\label{metmh} &
\met^{-}(\disk)  = \Iso_{+}(\Htwo)\backslash\himm(\disk,\disk)/\Aut(\disk) = 
\PslR\backslash\himm(\disk,\disk)/\PslR\, .
\end{align}
The automorphism group $\PslR$ acting on the right is as in \eqref{pslaction}, whereas the isometry groups act on the left as
\begin{align}
\label{SOact} & \text{SO}(3):\quad w' = e^{i\alpha}\frac{w-w_{0}}{1+\bar w_{0}w}\, \cvp\quad  w_{0}\in\bar{\mathbb C}\, ,\ (w,w')\in\bar{\mathbb C}^{2} = \Stwo\times\Stwo\, ,\\\label{pslisoact}& \PslR:\quad w' = e^{i\alpha}\frac{w-w_{0}}{1-\bar w_{0}w}\,\cvp\quad |w_{0}|<1\, ,\ (w,w')\in\disk^{2}=\Htwo\times\Htwo\, . 
\end{align}
Note that, following remark i) at the end of subsection \ref{zeroImmSec}, the two $\PslR$ groups, acting on the right and on the left in \eqref{metmh}, are of totally different origins, even though they coincide as abstract groups; this is an important point to keep in mind, see \cite{ferraJTconfgauge} for further discussion.

\subsubsection{\label{stab1Sec}Quantum JT gravity on the disk}

The disk partition function of quantum JT gravity in zero, positive or negative curvature, in the canonical ensemble, is expressed as a path integral
\be\label{qJT1} W^{0,\pm}(\beta_{\text q},\La) = \int_{\met_{\beta_{\text q}}^{0,\pm}(\disk)}\d\mu(g)\,  e^{-\frac{\La}{16\pi}A[g]}\ee
over the spaces $\smash{\met_{\beta_{\text q}}^{0,\pm}(\disk)}$ of constant curvature metrics on the disk with a fixed quantum length $\beta_{\text q}$.\footnote{Throughout the paper, we use the letter $W$ to denote quantum gravity partition functions in fixed topology, instead of $Z$, to emphasize the fact that they are continuum limits of generating functions which are themselves interpreted as the sum over connected Feynman graphs in suitable matrix models.} These spaces are subspaces of the spaces $\met^{0,\pm}(\disk)$ defined above, for which we impose the constraint of fixed quantum boundary length, see \cite{ferrari,ferraJTconfgauge}.  In \eqref{qJT1}, $A[g]$ denotes the area of the disk for the metric $g$ and $\La$ is the cosmological constant. The expression \eqref{qJT1} is of course formal. The challenge is to define measures $\smash{\d\mu(g)}$ over the spaces $\smash{\met_{\beta_{\text q}}^{0,\pm}(\disk)}$. As we have already mentioned, an approach for constructing these measures is to start from the discretized formulation and to take an appropriate continuum limit of the generating functions defined in \ref{Genfun1Sec}, see Section \ref{continuumSec}. A direct continuum approach is also proposed in \cite{ferrari,ferraJTconfgauge}. 

\noindent\emph{Remark on stability}

The partition functions \eqref{qJT1} will always be well-defined for positive or zero bulk cosmological constants, $\La\geq 0$. The situation for $\La<0$ is much more subtle.

At the classical level, we usually consider metrics that extend smoothly to the boundary. In zero and negative curvature, the isoperimetric inequalities then imply that the area is bounded above for fixed boundary length. In this sense, the models are thus classically stable. However, in positive curvature, there is no isoperimetric inequality and the model is classically unstable for $\La<0$, see Section \ref{poscurvlastSec}.

At the quantum level, the situation is much more subtle. Even in zero or negative curvature, the isoperimetric inequalities do not apply, since the boundary is a fractal curve of infinite length. There are thus configurations of infinite area that have a finite quantum length. As a consequence, the classical action $\frac{\La}{16\pi}A$ is not bounded from below when $\La<0$. But this does not imply that the models are quantum mechanically ill-defined, since the entropy of the configurations must be taken into account.

We conjecture the following behaviour:

i) The negative curvature theory is well-defined for all values, positive or negative, of the  cosmological constant $\La$.

ii) The zero curvature theory is well-defined only for a range $\La>\La_{\text c}$ of cosmological constants, where the critical value $\Lambda_{\text c}$ is strictly negative.

iii) The positive curvature theory is well-defined only for $\La\geq 0$.

This conjectured behaviour will be further discussed in Section \ref{nonzeroRSec}.

\subsection{\label{earlySec}Discussion of the existing literature}

Now that the stage is set, at least at the conceptual level, it is useful to review some of the basic ideas that underly the vast existing literature on JT gravity. This subsection is written mainly for non-expert in the field of SYK and JT gravity. Emphasis is put on important differences between standard approaches and our own work, in order to put into perspective the new point of view that we are introducing. 

By far the most popular framework to study the negative curvature model is based on the so-called ``reparameterization ansatz'' and the Schwarzian action. This description is usually assumed to become correct in a particular near-hyperbolic limit
\be\label{Schlimit}\ell\rightarrow\infty\, ,\quad \La\rightarrow -\infty\, ,\quad \frac{2 \ell}{|\La|} = \beta_{\text S}\quad \text{fixed}.\footnote{The factor of two in the definition of $\beta_{\text S}$ is inserted to match with the standard definition in the literature. Note that the cosmological constant $\La$ is related to the boundary value $\Phi_{\text b}$ of the dilaton field in the usual formulation of JT gravity by $|\La| = -\La = 2\Phi_{\text b}$.}\ee
As we shall argue, the Schwarzian theory is an effective description of JT gravity, that emerges in the limit $\La\rightarrow -\infty$ and is valid on distance scales much larger than the curvature length scale. The very definition of the limit in Eq.\ \eqref{Schlimit} is ``effective,'' since it refers to a smooth boundary length parameter $\ell$ that does not exist at the microscopic level. The relation between the macroscopic length $\ell$ and the microscopic parameters will be discussed in Section \ref{nonzeroRSec}.

An interesting attempt at a microscopic definition of JT gravity was proposed by Kitaev and Suh in \cite{KitaevSuh}. The idea is to assume that the disk boundaries are random Brownian loops. The resulting path integral is then related to the quantum mechanics of a charged particle in hyperbolic space. As is clear from the discussion in Section \ref{dismetSec}, the vast majority of the configurations introduced in this way are not associated with disk metrics, because a typical random loop does not bound a distorted disk. Therefore, the Kitaev-Suh model is not a model of random metrics.

Other proposed microscopic descriptions are based on using self-avoiding loops \cite{StanfordSAP,VerlindeSAP}. This captures only a very small subset of the allowed configurations, actually a subset of measure zero. These approaches are thus in the same spirit as the minisuperspace approximations to quantum gravity. They are unlikely to provide a good approximation to the microscopics of JT gravity, even qualitatively. For instance, the Hausdorff dimension of self-avoiding loops counted with the uniform measure, as in \cite{StanfordSAP}, is $4/3$, to be compared to the correct value 2 for pure JT. In the approach of \cite{VerlindeSAP}, the boundaries are smooth. A picture based on ``edge modes'' or broken diffeomorphisms, has also been studied; it is another minisuperpace approximation, closely related to the other approaches.

\subsubsection{\label{repaSec}Reparameterization ansatz}

The reparameterization ansatz is relevant only in the case of the negative curvature theory.

\paragraph{Brief review}

\begin{figure}
\centerline{\includegraphics[width=6in]{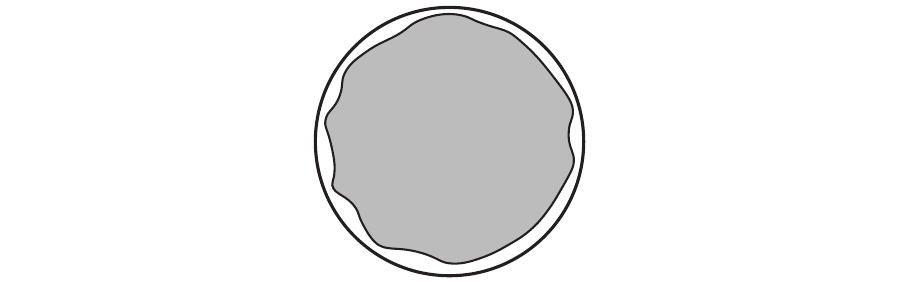}}
\caption{\label{repansatzfig} A reparameterization ansatz, or near-hyperbolic geometry, as defined in \cite{holomod3}. The open disk with the thick boundary is hyperbolic space $\Htwo = (\disk,\delta^{-})$ in the Poincar\'e disk representation, Eq.\ \eqref{deltaHtwo}. The geometry is defined by the cut-out region, in grey, which is itself determined by the shape of a gently wiggling boundary curve. In our language, the cut-out regions associated with reparameterization ansatzes are distorted disks associated with a very special class of embeddings of the source disk, on which the quantum gravity theory is defined, into hyperbolic space.}
\end{figure}

The authors of \cite{holomod3} considered the possibility of defining negative curvature geometries by ``cutting out pieces of hyperbolic spaces'' with particular shapes, as they phrase it. A typical example is depicted in Fig.\ \ref{repansatzfig}. The piece that is cut-out is defined by the shape of its boundary. The boundary curve is assumed to be smoooth, of length $\ell$, and is described as follows. We introduce the arc-length coordinate $s$ along the curve and polar coordinates $(r,\phi)$, $w=r e^{i\phi}$, such that the hyperbolic metric \eqref{deltaHtwo} reads
\be\label{deltaHtwo2} \delta^{-} = \frac{4(\d r^{2} + r^{2}\d\phi^{2})}{(1-r^{2})^{2}}\,\cdotp\ee
We parameterize
\be\label{repaansatzdef} r(s) = 1-\frac{1}{\ell} h(\vartheta)\, ,\quad \phi(s) = f(\vartheta)\, , \quad\text{where}\quad\vartheta = \frac{2\pi s}{\ell}\ee
is a $2\pi$-periodic variable. The length $\ell$ is considered to be large.\footnote{Note that, with the choice of metric \eqref{deltaHtwo2}, the units are chosen in such a way that the curvature length scale is one. ``Large'' thus means large compared to the curvature length scale.} The function $f$ is assumed to be smooth, strictly monotonic and such that $f(\vartheta + 2\pi) = f(\vartheta) + 2\pi$. These conditions define the group
\be\label{diffSonedef} \diffSp = \bigl\{f\in C^{1}(\mathbb R)\mid f(\vartheta + 2\pi) = f(\vartheta) + 2\pi\ \text{and}\ f'(\vartheta) >  0\bigr\}/2\pi\mathbb Z\ee
of orientation-preserving diffeomorphisms of the circle. Denoting with a dot the derivatives with respect to $s$, the function $h$ is related to $f$ by the constraint
\be\label{repansatzcons1} \dot r^{2}+r^{2}\dot\phi^{2} = \frac{1}{4}(1-r^{2})^{2}\, .\ee
One also assumes that the functions $f$ and $h$ are infinitely differentiable and their derivatives $f^{(k)}$ and $h^{(k)}$, $k\geq 1$, with respect to $\vartheta$, are of order one when $\ell$ is large. This implies that $\dot r = O(\ell^{-2})$, $\ddot r = O(\ell^{-3})$, $\dot\phi = O(\ell^{-1})$, $\ddot\phi = O(\ell^{-2})$, etc. Under these assumptions, Eq.\ \eqref{repansatzcons1} implies that $h$ can be expressed in terms of $f$ in a large $\ell$ expansion of the form
\be\label{hexpgen} h = 2\pi\sum_{n\geq 0}\Bigl(\frac{2\pi}{\ell}\Bigr)^{n}h_{n}\, ,\ee
where the coefficients $h_{n}$ are local expressions involving higher and higher derivatives of $f$. For instance, the first terms are computed to be
\be\label{hasfepsexp} h_{0}= f' ,\ h_{1}= -\frac{1}{2}f'^{2} ,\ h_{2}=\frac{f''^{2}}{2f'}\cvp\ h_{3}= \frac{1}{8}f'^{4}-\frac{1}{2}f''^{2},\ h_{4}=-\frac{1}{2}f'f''^{2}-\frac{5f''^{4}}{8f'^{3}}+\frac{f''^{2}f'''}{f'^{2}}\cdotp\ee
The resulting curve is a gently ``wiggling'' boundary. The associated large $\ell$ geometries are called ``nearly hyperbolic.'' Everything is fixed in terms of the diffeomorphism $f$, hence the name ``reparameterization ansatz.''

\paragraph{Reparameterization ans\"atze and distorted disks}

\begin{figure}
\centerline{\includegraphics[width=6in]{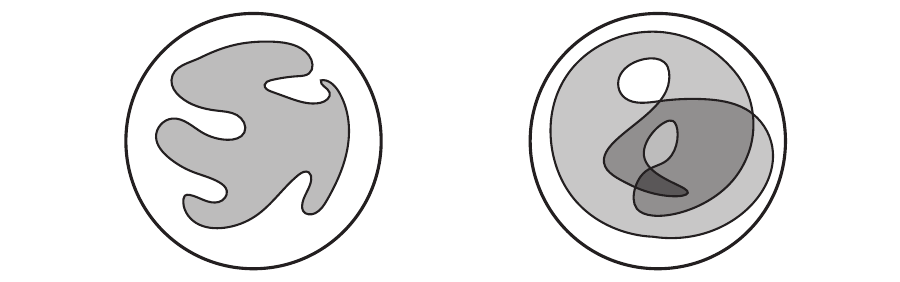}}
\caption{\label{notrepansatzfig} Left inset: an embedding that is not a reparameterization ansatz, because the function $\phi(s)$ is not monotonic. Right inset: an immersion that is not an embedding, because the associated distorted disk has overlaps. Note that both configurations are associated with constant curvature metrics on the disk and thus contribute to the JT quantum gravity path integral.}
\end{figure}

In the language of Section \ref{immersionSec}, the nearly-hyperbolic geometries correspond to particular immersions $F:\disk\rightarrow\Htwo$ and thus to particular disk metrics $g = F^{*}\delta^{-}$. The cut-out piece of hyperbolic space is the image $F(\disk)$ and thus a special instance of a distorted disk. The immersions, or, equivalently, the distorted disks, that are taken into account by the reparameterization ansatz, are \emph{very} atypical:

i) The reparameterization ansatz immersions are actually embeddings, which implies in particular that the distorted disks have no overlaps.

ii) Since $f$ is assumed to be strictly monotonic, the boundary curve cannot make zig-zags. 

An embedding that is not a reparameterization ansatz, as well as an immersion that is not an embedding, are depicted in Fig.\ \ref{notrepansatzfig}. Note that immersions that are not embeddings cannot be seen as ``cut-out'' pieces of hyperbolic space.

iii) The condition that the derivatives $f^{(k)}$ of the reparameterization are of order one means that we cut-off frequencies of order $\ell$ or higher. In other words, the distances under which the wiggling of the boundary curve occurs must be much larger than the curvature scale.

iv) The reparameterization ansatz describes geometries that have a smooth, rectifiable boundaries, whereas typical geometries will have fractal boundaries almost surely.

Replacing the full path integral \eqref{qJT1} by a path integral over the smooth wiggling boundaries described by the reparameterization ansatz is thus an approximation in which only a measure zero subset of the possible metrics are taken into account. We shall argue in Section \ref{nonzeroRSec} that this approximation is justified in the limit $\La\rightarrow -\infty$ as an effective, long-wavelength description, for which the curvature length scale plays the r\^ole of an ultraviolet cut-off.

\noindent\emph{Remarks}:

\noindent i) For any given cut-out region, $F$ exists and is unique modulo the right action of $\diff_{+}(\disk)$. In conformal gauge, this result is equivalent to the Riemann mapping theorem. 

\noindent ii) The condition $f(\vartheta+2\pi) = f(\vartheta) + 2\pi$ on the reparameterization $f$ is usually imposed in a rather ad-hoc way. One might wonder why other sectors with higher winding number, $f(\vartheta+2\pi) = f(\vartheta) + 2\pi n$, $n\in\mathbb N$, are not included (changing the sign of $n$ would simply amounts to changing the orientation). The answer is that it is only in the case $n=1$ that the immersion $F$ exists, and therefore it is only in this case that the gently wiggling curve is associated with a metric on the disk. This follows from the continuous version of the constraint \eqref{EulergenflatJTdisk}, see Section\ \ref{WhitneySec}.

\paragraph{The space $\wig(\disk)$ of wiggling boundaries}

Exactly as in our discussion in Section \eqref{metpmdisSec}, configurations that are related by the action of the $\PslR$ isometry group of $\Htwo$, Eq.\ \eqref{pslisoact}, must be identified. This action yields an action on $\diffSp$, via the relation $\phi(s) = f(\vartheta)$ in \eqref{repaansatzdef}, whose explicit form may be derived straightforwardly in the large $\ell$ expansion by using \eqref{pslisoact} and $w=r e^{i\phi}$. At leading order, the action matches with the left action of the standard $\PslR$ subgroup of $\diffSp$, which amounts to identifying circle diffeomorphisms $f$ and $\tilde f$ related by a relation of the form
\be\label{pslsubdiff} \tan \frac{\tilde f}{2} = \frac{a\tan\frac{f}{2} + b}{c\tan\frac{f}{2} + d}\,\cvp\quad ad-bc = 1\, .\ee
The parameters $(z_{0},\alpha)$ in \eqref{pslisoact} and $(a,b,c,d)$, constrained by $ad-bc = 1$, are related by
\be\label{parreldiskaut} e^{i\alpha} = \frac{b-c-i(a+d)}{c-b-i(a+d)}\,\cvp\quad z_{0}=-\frac{b+c+i(a-d)}{b-c-i(a+d)}\,\cdotp\ee
Moreover, diffeomorphisms that differ by a constant shift of the variable $\vartheta$ should also be identified, since there is no privileged point on the boundary.\footnote{In most of the previous literature, it is actually implicitly assumed that one has a marked point on the boundary. This is not natural from the quantum gravity point of view. This point is discussed further below.} The set of metrics associated with the reparameterization ansatz thus form the space
\be\label{Wigspacedef} \wig(\disk) = \PslR\backslash\diffSp/\Sone\, .\ee
The notation Wig refers to the ``wiggling'' boundary. 

Note that, unlike in \eqref{metm}, the group $\diff_{+}(\disk)$ does not appear in the quotient defining $\wig(\disk)$. This is consistent: the description in terms of the circle diffeomorphism $f$ does not refer directly to the source disk $\disk$ and thus the group $\diff_{+}(\disk)$ does not act on $f$.

In the above discussion, to derive \eqref{Wigspacedef} with an action of the isometry group given by \eqref{pslsubdiff}, we have implicitly restricted ourselves to isometry transformations that are independent of $\ell$. But, actually, the reparameterization ansatz and the large $\ell$ expansion are not consistent with the action of the full isometry group. Indeed, the hyperbolic space $\Htwo$ is homogeneous and $\PslR$ acts transitively. This implies, for example, that we can always find an isometry that maps any point on the wiggling boundary parameterized as in \eqref{repaansatzdef} to the center $r=0$ of the Poincar\'e disk. This new configuration is clearly not of the form \eqref{repaansatzdef} for $h= O(1)$. One can also find isometries that create turning points and are thus inconsistent with the condition $f'>0$. Since the full $\PslR$ is gauged, this is a priori problematic. However, the situation is not so different in spirit from the usual perturbative treatment of gauge theories, for which the usual gauge fixing conditions have Gribov ambiguities. The perturbative consistency is ensured by the fact that the Gribov copies are too far away in field space to be seen in perturbation theory. 

\paragraph{The Schwarzian theory}

The area $A[g]$ for a reparameterization ansatz can be computed straightforwardly. One starts  from the exact formula giving the extrinsic curvature of the boundary 
\be\label{kexact} k = \frac{2(1+r^{2})}{1-r^{2}}\dot\phi + \frac{r\ddot\phi}{\dot r}\,\cdotp\ee
As for $h$ in \eqref{hexpgen}, \eqref{hasfepsexp}, the extrinsic curvature has a local large $\ell$ expansion involving higher and higher derivatives of $f$,
\be\label{kexpgen} k = 1 + \sum_{n\geq 1}\Bigl(\frac{2\pi}{\ell}\Bigr)^{n}k_{n}\, .\ee
The first few terms can be easily computed,
\be\label{kexp}\begin{split} & k_{1}=0\, ,\quad k_{2}= \Sch[f] + \frac{1}{2}f'^{2}\, ,\quad k_{3}=0\, ,\\ & k_{4}= \frac{f''f^{(4)}}{f'^{2}}+\Bigl(\frac{f'''}{f'}\Bigr)^{2} - \frac{11f''^{2}f'''}{2f'^{3}}-\frac{1}{2}f'f'''+\frac{27}{8}\Bigl(\frac{f''}{f'}\Bigr)^{4}-\frac{1}{8}f'^{4}+\frac{1}{4}f''^{2}\, , \ \text{etc.}
\end{split}
\ee
The symbol $\Sch[f]$ denotes the Schwarzian derivative of $f$,
\be\label{Schderdef} \Sch[f](\vartheta) = \{ f,\vartheta\} = \frac{f'''}{f'}-\frac{3}{2}\Bigl(\frac{f''}{f'}\Bigr)^{2} = \Bigl(\frac{f''}{f'}\Bigr)'- \frac{1}{2}\Bigl(\frac{f''}{f'}\Bigr)^{2}\, .\ee
Using the Gauss-Bonnet formula \eqref{GaussBonnet} for $(h,b)=(0,1)$ and $R=-2$, one then gets 
\be\label{Arearepansatz} A[g] = \ell - 2\pi +\frac{2\pi}{\ell} \int_{0}^{2\pi}\Sch\bigl[\tan\frac{f}{2}\bigr]\,\d\vartheta + O(\ell^{-3})\, .\ee
The partition function within the reparameterization ansatz approximation is then
\be\label{Zrepansatz} W_{\text{Sch}}(\beta_{\text S}) = \int_{\wig(\disk)}\d\mu (f) \, e^{-S_{\text{Sch}}[f]}\, ,\ee
for a ``Schwarzian'' action
\be\label{SSch} S_{\text{Sch}}[f] = -\frac{1}{4\beta_{\text S}}\int_{0}^{2\pi}\Sch\bigl[\tan\frac{f}{2}\bigr]\,\d\vartheta = \frac{1}{8\beta_{\text S}}\int_{0}^{2\pi}\Bigl[\Bigl(\frac{f''}{f'}\Bigr)^{2} - f'^{2}\Bigr]\,\d\vartheta\, ,\ee
where the parameter $\beta_{\text S}$ was defined in \eqref{Schlimit}. The formal measure $\d\mu(f)$ is postulated to be the Fadeev-Popov measure obtained after gauge-fixing the $\PslR$ symmetry from the natural formal left-invariant measure on $\diffSp$. Modulo a multiplicative factor of the form $e^{\zeta_{0}-\beta \varepsilon_{0}}$, where $\zeta_{0}$ and $\varepsilon_{0}$ are arbitrary constants, the Schwarzian partition function can be computed exactly and turns out to be one-loop exact \cite{SWloca}, 
\be\label{ZRAexact} W_{\text{Sch}}(\beta_{\text S}) = e^{\zeta_{0}-\beta_{\text S} \varepsilon_{0}}\beta_{\text S}^{-3/2} e^{\frac{\pi}{4\beta_{\text S}}}\, .\ee
The constants $\varepsilon_{0}$ and $\zeta_{0}$ can be interpreted as renormalized ground state energy and zero temperature entropy, see below. They cannot be predicted in the Schwarzian theory because they correspond to a priori arbitrary finite local counterterms one may add to the action, associated with a boundary cosmological constant and the topological Einstein-Hilbert terms respectively.\footnote{In principle, counterterms should be discussed in the UV-complete microscopic description, not in the effective Schwarzian description. Consistency between the two descriptions imply an interesting constraint, see Section \ref{nonzeroRSec}.} This formalism, called the ``Schwarzian theory,'' is at the basis of essentially all the analytic calculations made in JT gravity, see e.g.\ \cite{RAcalc} and references therein.

The Schwarzian partition function \eqref{ZRAexact} can be expressed in terms of a density
\be\label{Schrho}\rho_{\text{Sch}}(E) = \frac{2}{\pi}e^{\zeta_{0}}\sinh\sqrt{\pi (E-\varepsilon_{0})}\ee
as
\be\label{SchZdensity} W_{\text{Sch}}(\beta_{\text S}) = \int_{\varepsilon_{0}}^{\infty}\rho_{\text{Sch}}(E) e^{-\beta_{\text S} E}\,\d E\, .\ee
This may suggest a boundary quantum mechanical interpretation for which $\beta_{\text S}$ is an inverse renormalized temperature and $\rho_{\text{Sch}}$ the density of state for a boundary Hamiltonian $H$, such that 
\be\label{ZQMinter} W_{\text{Sch}}(\beta_{\text S}) \overset{?}{=} \tr e^{-\beta_{\text S} H}\, .\ee
We put a question mark on top of the equal sign because this formula cannot be strictly correct as such. Indeed, the partition function for a Hamiltonian with a continuous spectrum is always infinite\footnote{This is a consequence of the fact that the stationary states are not normalizable in the case of a continuous spectrum. A similar remark was also made in \cite{SWloca}.} and thus \eqref{ZQMinter} does not make sense. The state-of-the-art interpretation of the Schwarzian partition function, together with its generalizations, when multiple boundaries are present, is that it gives an ensemble average for a suitable random matrix model, the random matrix being the Hamiltonian itself \cite{SSS,SSSW}. Formula \eqref{ZQMinter} is thus replaced by
\be\label{ZEAinter} W_{\text{Sch}}(\beta_{\text S}) = \bigl\langle\tr e^{-\beta_{\text S} H}\bigr\rangle_{\text{disk}}\ee
where only the disk contribution in the genus expansion of the matrix model is kept on the right-hand side.\footnote{Higher genus contributions correspond to higher genus topologies on the JT gravity side.}

\paragraph{Discussion}

i) The Schwarzian partition function $W_{\text{Sch}}$ should be the limit of the ``microscopic'' partition function $W^{-}(\beta_{\text q},\La)$ defined in \eqref{qJT1}, when $\La\rightarrow -\infty$, with a particular scaling given by \eqref{Schlimit}. Recall that $\ell$ is a macroscopic length parameter, which, in principle, can be expressed in terms of the microscopic parameters $\beta_{\text q}$ and $\La$, see Section \ref{nonzeroRSec}.

However, it turns out that the partition function \eqref{ZRAexact}, with the coefficient of proportionality $\smash{\beta_{\text S}^{-3/2}}$, corresponds to a counting for which one marks a point on the disk boundary. The marking yields an additional factor of $\beta_{\text S}$ with respect to the standard quantum gravity conventions, that we are using in the present paper.\footnote{Marking a point on the boundary is not natural in quantum gravity, because all the markings are associated with the same metric.} Another way to state this is that the Schwarzian partition function corresponds to an integral over $\PslR\backslash\diffSp$ instead of $\wig(\disk)=\PslR\backslash\diffSp/\Sone$. The fact that the proportionality factor $\smash{\beta_{\text S}^{-3/2}}$ corresponds to a counting with a marked point on the boundary is also supported by the calculations presented in \cite{Loopcalc} and is consistent with what one finds in simplified models \cite{StanfordSAP,KitaevSuh,RDFollowup1}.


ii) The difficulties in interpreting the Schwarzian theory in a full-fledged quantum mechanical framework do not prevent the Schwarzian theory from capturing correctly some important features of genuine quantum mechanical systems in some limit. For example, the free energy predicted in the Schwarzian framework is of the form
\be\label{SchF} \beta_{\text S} F = -\ln W_{\text{Sch}} = \beta_{\text S} \varepsilon_{0} - \zeta_{0} - \frac{\pi}{4\beta_{\text S}} + \frac{3}{2}\ln\beta_{\text S}\, .\ee
In the SYK model (or in similar quantum mechanical tensor or vector-matrix models), the free energy has the general form \cite{Kitaev, KitaevSuh}
\be\label{FSYK} \beta F_{\text{SYK}} = \beta E_{0} - N\bigl(s_{0} + \frac{a_{1}}{\beta J} + \frac{a_{2}}{(\beta J)^{2}} + o(1/(\beta J)^{2})\bigr) + \frac{3}{2}\ln(c\beta J) +o(N^{0})\, .\ee
The parameters are the inverse temperature $\beta$, a dimensionful coupling constant $J$ that plays the role of a UV cut-off for the Schwarzian description and the number of degrees of freedom $N$; $a_{1}$, $a_{2}$, $c$ are dimensionless numerical constants.\footnote{In the SYK model, one often writes $a_{1}=2\pi^{2}\alpha_{\text S}$.} $E_{0}$ and $Ns_{0}$ represent the ground state energy and the zero temperature entropy. The expansion \eqref{FSYK} is valid at large $N$ and large dimensionless coupling $\beta J$. Interestingly, the ratio $\beta J/N$ may be of order one. The leading terms at large $N$ and $\beta J$ in \eqref{FSYK}, even when $\beta J/N$ is of order one, can be identified with \eqref{SchF} if we set
\be\label{JTvsSYKrel} \beta_{\text S} = \frac{2\ell}{|\La|} = \frac{\pi\beta J}{4 N a_{1}}\,\cdotp\ee

The Schwarzian description is also consistent with the thermodynamic properties of near-extremal black holes; in particular, it reproduces the linear growth of the heat capacity at low temperature. The emergence of a continuous spectrum is also in line with the idea that a large $N$ limit is taken. The near-conformal properties are associated with the large cut-off $\beta J\gg 1$.

iii) In spite of its great interest, the Schwarzian description is far from capturing the full complexity of the SYK model or of similar quantum mechanical models. It is an infrared effective description, that does not correctly account for the UV properties that are crucial to ensure the quantum consistency of the models. Simple conjectured extensions of the Schwarzian result, as reviewed below in Section \ref{JTptSec}, are also not able to reproduce the details of the large $\beta J$ expansion of the quantum mechanical models, that involves fractional powers of $1/(\beta J)$. This is of course well-known. What may not have been fully realized or emphasized until now is that the Schwarzian description has the same status with respect to JT gravity as it has with respect to SYK. It is an effective description of JT gravity that falls short of capturing essential features of the microscopic description, like, for instance, the fractal nature of the boundary. As a consequence, the one-loop exactness of the partition function \eqref{ZRAexact} does not have any fundamental meaning.

iv) It is interesting to pinpoint the various levels at which we make an approximation when we use the Schwarzian description.

\noindent$\bullet$ The description takes into account only a subset of measure zero in the space of constant negative curvature disk metrics. For instance, configurations like the ones depicted in Fig.\ \ref{notrepansatzfig} are not taken into account. Moreover, the boundary is assumed to be smooth, whereas it is a fractal at the microscopic level.

\noindent$\bullet$ The Schwarzian description is based on the use of an \emph{approximation} of the area term $\frac{\La}{16\pi} A$ in the action, valid on distance scales much larger than the curvature scale. It is easy to construct ``reparameterization ansatz'' configurations for which the area is finite and yet not given by the Schwarzian action. For instance, take
\be\label{fbreak1} f(\vartheta) = \vartheta +\frac{\alpha}{n^{2}}\cos (n\vartheta)\ee
where the frequency $n$ is of order $\ell$. This ansatz is such that $f'$ and $f''$ are of order one at large $\ell$, because the amplitude is chosen to go as $1/n^{2}$, but $f''' = O(\ell)$, $f^{(4)}= O(\ell^{2})$, etc. In this case, the local expansions \eqref{hexpgen} and \eqref{kexpgen}, and thus the formula \eqref{Arearepansatz} for the area, are not valid. However, the small amplitude in $1/n^{2}$ ensures that the condition
\be\label{areaconRA} A[g] - \ell + 2\pi = O(1/\ell)\, ,\ee
which yields the finiteness of the action after the usual counterterms have been subtracted, is satisfied. Note that it is also easy to construct finite action configurations that have many microscopic zig-zags or overlaps. 

%

\noindent$\bullet$ Last, but not least, the correct UV structure of the theory is not governed by the area term but by a completely different induced boundary gravitational action \`a la Liouville, as explained in \cite{ferrari,ferraJTconfgauge}. The induced gravitational action is essential to ensure the consistency of the theory on all scales.

\subsubsection{\label{edgeSec}Edge modes}

\paragraph{Basic idea}

The so-called \emph{edge modes} appear when gauge theories are formulated on manifolds with boundaries in such a way that the boundary conditions break part of the gauge symmetry on the boundary. The broken generators then become genuine dynamical degrees of freedom. A typical example is BF gauge theory on a disk with gauge group $G$. Naively, because the BF theory gauge connection is flat and the disk is contractible, there is no degree of freedom. But the standard boundary conditions break the gauge symmetry on the boundary. Edge modes are then described by maps $h:\partial\disk\rightarrow G$ and the model is equivalent to  the dynamics of a particle moving on $G$. A similar phenomenon occurs in the gauge theoretic description of JT gravity \cite{GTFer}. 

Gravitational edge modes are similarly associated with the partial breaking of the group of diffeomorphisms. They play a central role in three dimensional asymptotically hyperbolic or AdS gravity, where the asymptotic boundary conditions break part of the full group of diffeomorphisms \cite{edgegrav3d}. It is natural to study whether a similar description is possible in two dimensions as well \cite{junggi1}. Since imposing the constraint of constant  curvature eliminates the local metric degrees of freedom, it is legitimate to seek a formulation in terms of edge modes. The wiggling boundary picture seems consistent with this idea.

However, there are major differences between the two and the three dimensional cases and these differences seem to be at the origin of some confusion.\footnote{We would like to thank Junggi Yoon for a stimulating discussion on these topics. This discussion has motivated us to write the present subsection, the aim of which is to clarify certain subtle and potentially confusing aspects in a pedagogical way.} An important difference is that in two dimensions, there is a unique complete asymptotically hyperbolic geometry in two dimensions, which is hyperbolic space itself, $\Htwo=(\disk,\delta^{-})$. We thus work with finite cut-off geometries (finite quantum boundary length $\beta_{\text q}$ or finite smooth length $\ell$ if one deals with smooth boundaries). But in this case, the group of diffeomorphisms $\diff_{+}(\disk)$ is \emph{not} broken: it remains an \emph{exact} gauge symmetry of the model, for any finite value of the microscopic parameters, and this unbroken gauge symmetry includes the boundary reparameterizations. \emph{There is thus no standard edge mode in the fundamental gravitational description of the model.}

\paragraph{A construction in JT gravity and the space of embeddings}

Let us consider the group $\diff_{+}(\Htwo)$ of diffeomorphisms of hyperbolic space. This is of course a group of diffeomorphisms of the disk $\disk$, but, unlike for $\diff_{+}(\disk)$, the diffeomorphisms of $\diff_{+}(\Htwo)$ do not need to extend smoothly to $\partial\disk$, which is at infinite distance in $\Htwo$. The group $\diff_{+}(\Htwo)$ acts of the space of immersions $\imm(\disk,\Htwo)$ by left multiplication, $F\mapsto\psi\circ F$. Using $g=F^{*}\delta^{-}$, we get an action on the space of metrics $\met^{-}(\disk)$. The space $\met^{-}(\disk)$ thus decomposes into a disjoint union of orbits under $\diff_{+}(\Htwo)$.

One of the orbits corresponds exactly to the space of embeddings, $\emb(\disk,\Htwo)\subsetneq\imm(\disk,\Htwo)$, which is associated to the subspace of metrics
\be\label{SALmetdef} \salmet^{-}(\disk) = \Iso(\Htwo)\backslash\emb(\disk,\Htwo)/\diff_{+}(\disk)\subsetneq\met^{-}(\disk)\, .\ee
The notation $\salmet^{-}(\disk)$ is motivated by the fact that embeddings are in one-to-one correspondence with the boundary curves that are self-avoiding loops. Note that all the embeddings are in the same orbit because there is always a diffeomorphism that maps the interior of two distinct self-avoiding loops. 

To obtain an explicit description of this orbit, we can start picking a particular element of the orbit, i.e.\ a given embedding $F_{0}$, associated with a background metric $g_{0}=F_{0}^{*}\delta^{-}$. For instance, we can use a round disk, $F_{0}(z) = R z$ for some $R$, $0<R<1$. We denote by $\disk_{0}=\im F_{0}\subset\Htwo$ the associated distorted disk and by $\diff_{+}(\disk_{0})$ the subgroup of $\diff_{+}(\Htwo)$ that maps $\disk_{0}$ onto $\disk_{0}$ and $\partial\disk_{0}$ onto $\partial\disk_{0}$. Any $\psi\in \diff_{+}(\disk_{0})$ is such that $\psi_{|\bar\disk_{0}} = F_{0}\circ\varphi\circ F_{0}^{-1}$ for some $\varphi\in\diff_{+}(\disk)$. For $\psi\in\diff_{+}(\disk_{0})$, the embeddings $F_{0}$ and $\psi\circ F_{0}$ thus correspond to the same metric, since $\psi\circ F_{0} = F_{0}\circ\varphi$ for $\varphi\in\diff_{+}(\disk)$. The group $\diff_{+}(\disk_{0})$ is thus a genuine gauge symmetry, that realizes the group of diffeomorphisms of the source disk embedded in $\Htwo$ via $F_{0}$. In other words, $\diff_{+}(\disk_{0})$ is the isotropy group of the embedding $F_{0}$ under the action of $\diff_{+}(\Htwo)$ in the quotient \eqref{SALmetdef}. Overall, this implies that
\be\label{SALmetquot} \salmet^{-}(\disk) = \Iso(\Htwo)\backslash\diff_{+}(\Htwo)/\diff_{+}(\disk_{0})\, .\ee
This is an ``edge-mode-like'' description: one might say that the group $\diff_{+}(\Htwo)$ is broken and that the ``edge modes,'' that are the non-trivial diffeomorphisms with respect to the quotient \eqref{SALmetquot}, are the degrees of freedom of the theory. These ``edge modes'' are the diffeomorphisms that do not let the boundary $\partial\disk_{0}$ invariant. They are sometimes called ``radial diffeomorphisms.''

Let us emphasize that the ``edge mode'' terminology in the present context is dangerous, because it hides some crucial differences with the gauge theoretic models for which a fundamental edge mode description is possible. The differences are twofold.

i) First, the group $\diff_{+}(\Htwo)$ is not the group $\diff_{+}(\disk)$ of diffeomorphisms of the source manifold $\disk$. The group $\diff_{+}(\disk)$ is the fundamental gauge symmetry of the gravitational theory. The group $\diff_{+}(\Htwo)$ appears when one describes the metrics via immersions in $\Htwo$, as the group of diffeomorphisms of the target space $\Htwo$.

ii) Second, the space of metrics \eqref{SALmetquot} described in this formalism is only a small subset of the full set of metrics, associated with embeddings, whereas the general case involves immersions. 

Let us note that the reparameterization ansatz, which is itself a special case of embeddings, can be obtained from \eqref{SALmetquot} by considering infinitesimal elements of $\diff_{+}(\Htwo)$.

\paragraph{Some clarifications on the nature of the degrees of freedom in JT gravity}

The formulation in terms of the wiggling boundary and/or in terms of edge modes may lead to considerable confusion as to the nature of the degrees of freedom of the model and in the way these degrees of freedom have to be treated. One might think that the metric $\delta^{-}$ on the target space $\Htwo$, the wiggling boundary, and the metric $g$ on the source disk $\disk$, are three independent sets of variables, that must be varied independently in an action principle, or that must be integrated over in a path integral. This is incorrect. \emph{The only degrees of freedom are the constant curvature metrics $g$ on the source disk $\disk$, modulo the action of $\diff_{+}(\disk)$, as described in Sections \ref{dissurfrevSec} to \ref{immersionSec}.} The wiggling boundary is nothing more than a convenient trick to describe a subset of the metrics on $\disk$, using $g=F^{*}\delta^{-}$ and a set of special embeddings $F$. Varying the wiggling boundary amounts to varying $F$ and thus $g$. Note that $\delta^{-}$ is fixed. Thus, the wiggling boundary does not correspond to new degrees of freedom in the model. As for the metric $\delta^{-}$ on $\Htwo$, it is unique. Varying it or integrating over it does not make any sense.

\subsubsection{Approaches to the finite cut-off models}

The first proposal for a microscopic definition of JT quantum gravity, for the case of negative curvature, was made by Kitaev and Suh in \cite{KitaevSuh}. The proposal is based on two assumptions:

i) The degrees of freedom of the theory are in one-to-one correspondence with the shape of a boundary curve immersed in $\Htwo$.

ii) The boundary curve can be an arbitrary closed random walk. 

These assumptions lead to an equivalence between quantum JT gravity in negative curvature and the quantum theory of a particle in $\Htwo$. As argued in \cite{KitaevSuh}, the properties of the model so obtained are consistent with the results expected in the Schwarzian limit \eqref{Schlimit}.

However, if we assume, as we do, that JT gravity is a theory of random constant curvature metrics, the asumptions made in \cite{KitaevSuh} are inadequate. Indeed, on the basis of what we have already explained, it is clear that the vast majority of closed random walks do not bound a distorted disk and are therefore not associated with a metric. Simple examples of forbidden closed walks were depicted in Fig.\ \ref{forbidFig}. The rather subtle conditions under which a closed curve is associated with a metric will be explored in detail in Section \ref{bdviewSec}. So assumption ii) is ruled out by the metric axiom. Moreover, and this is one of the biggest surprise to emerge from the analysis of Section \ref{bdviewSec}, the very idea that the degrees of freedom are in one-to-one correspondence with the shape of the boundary is actually erroneous! Assumption i) must therefore also be called into question.

The conclusion is that the model proposed by Kitaev and Suh, even though perfectly well-defined, is not a theory of quantum gravity, if one is willing to insist on the fact that a theory of quantum gravity in two dimensions must be a diffeomorphism invariant theory of random metrics.

This being said, the Kitaev-Suh model is a very valuable tool to explore, at the qualitative level, some important physical features of JT gravity, like, for instance, the emergence of the Schwarzian description on long distance scales for large negative cosmological constant. A notable fact, that suggests that the Kitaev-Suh model is a sensible approximation to JT, is that the Hausdorff dimension of the Brownian paths is 2, matching the expected value for pure JT \cite{ferrari}. It can also be generalized to the flat and positive curvature cases, yielding interesting qualitative clues on the possible behaviour of JT gravity in these cases as well. Further discussion of the model will be given in Section \ref{randompathrev} and in \cite{RDFollowup1}.

Stanford and Yang have proposed in \cite{StanfordSAP} a very different approach. They do keep the assumption i) by Kitaev and Suh, that the degrees of freedom are in one-to-one correspondence with the shape of the boundary curve, but, instead of using arbitrary closed random walks, they propose that the boundary curves should be self-avoiding loops. As explained above, each self-avoiding loop corresponds to an embedding of the disk and is thus associated with a unique metric on the disk, belonging to the space $\salmet^{-}(\disk)$ defined in \eqref{SALmetdef}. However, this takes into account only a measure zero fraction of the full set of constant negative curvature metrics. In particular, the Hausdorff dimension of the a typical SAL boundary is 4/3 and not 2 as in JT. It is unlikely that the SAL model could provide a satisfactory approximation to JT gravity, but we found some qualitative aspects of the model and of the discussion in \cite{StanfordSAP} very instructive. For these reasons, we shall provide more information on the self-avoiding loop model in Section \ref{randompathrev}. 

Iliesu et al., in \cite{VerlindeSAP}, use the same basic assumptions as in \cite{StanfordSAP}, but instead of starting from the usual microscopic definition of random self-avoiding loops, which is based on the natural uniform measure in the discretized lattice formulation, and which yields in particular fractal boundaries, they develop the model directly in the continuum, using various assumptions on the path integral integration measure, etc., which are in line with the usual asumptions made in the Schwarzian description. In particular, they deal with smooth rectifiable boundaries. Of course, the results obtained in this way are very different from the results presented in \cite{StanfordSAP}. Instead of an attempt at a microscopic definition of the model, the work in \cite{VerlindeSAP}, and its interesting generalization in \cite{Griguolo}, may be better interpreted as being a proposal for a possible extension of the effective Schwarzian framework, trying to push its range of validity beyond the leading order in the near-hyperbolic limit \ref{Schlimit}, but by remaining within the framework of an effective long-wavelength description.

In the next subsection, we are going to show that the results in \cite{VerlindeSAP} yield a non-trivial conjecture on the form of the semi-classical expansion of JT gravity, that does not seem to have been noted previously. As we shall argue, this prediction appears to be hard to reconcile with the expected microscopic properties of the theory.

\subsubsection{\label{JTptSec}A conjecture to all orders in the large $|\La|$ expansion}

According to \cite{VerlindeSAP}, the finite cut-off disk partition function in negative curvature is given by\footnote{To compare with the formulas in \cite{VerlindeSAP}, note that, in our conventions, the cosmological constant is related to the boundary value $\phi_{\text b}$ of the dilaton in the usual formulation of JT gravity by $|\La| = -\La = 2\phi_{\text b}$.}
\be\label{IKTVZdensity} W_{\text{IKTV}}^{-}(\beta_{\text S}) \equiv \int_{\varepsilon_{0}}^{\infty}\rho_{\text{IKTV}}(E)\, e^{-\beta_{\text S} E}\,\d E\ee
for a density
\be\label{IKTVrho}\rho_{\text{IKTV}}(E) = \frac{2}{\pi}e^{\zeta_{0}}\Bigl(1-\frac{32\pi}{\La^{2}} \bigl(E-\varepsilon_{0}\bigr)\Bigr)\sinh\sqrt{\pi \bigl(E-\varepsilon_{0}\bigr)\Bigl(1-\frac{16\pi}{\La^{2}}\bigl(E-\varepsilon_{0}\bigr)\Bigr)}\ee
generalizing the Schwarzian density \eqref{Schrho}. The symbol $\equiv$ in Eq.\ \eqref{IKTVZdensity} means that the formula must be understood perturbatively, as explained  below. Indeed, if taken seriously as a strict equality, Eq.\ \eqref{IKTVZdensity} would predict an imaginary part in the partition function. Instead, one must imagine that there is a cut-off which is $o(|\La|^{2})$ in the integral over $E$. One can also give an expression in terms of a modified Bessel function,
\be\label{IKTVZexact} W_{\text{IKTV}}^{-}(\beta_{\text S}) \equiv e^{\zeta_{0}-\beta_{\text S}\varepsilon_{0}} \frac{\La^{2}}{8\ell}\frac{e^{-\frac{|\La|\ell}{16\pi}}}{1+(\frac{2\pi}{\ell})^{2}} I_{2}\Biggl[\frac{|\La|\ell}{16\pi}\sqrt{1+\bigl(\frac{2\pi}{\ell}\bigr)^{2}}\Biggr]\ee
or an alternative integral representation using the Schwarzian density \eqref{Schrho} 
\be\label{IKTVint2} W_{\text{IKTV}}^{-}(\beta_{\text S})\equiv
\int_{\varepsilon_{0}}^{\infty}\rho_{\text{Sch}}(E)\, e^{-\frac{\ell |\La|}{16\pi} \bigl[1-\sqrt{1-64\pi (E-\varepsilon_{0})/\La^{2}}\bigr]}\, \d E\, ,\ee
which is suggested by an approach via the one-dimensional $T\bar T$ deformation \cite{ttbar}.

The formulas \eqref{IKTVZdensity}-\eqref{IKTVrho}, \eqref{IKTVZexact} and \eqref{IKTVint2} all yield the same near-hyperbolic expansion, corresponding to the limit \eqref{Schlimit}, to all orders in $1/\ell$ at fixed $\beta_{\text S}$, see below. However, they contain imaginary parts and differ by terms that are exponentially small, of order $\smash{\exp(-\frac{|\La|\ell}{8\pi})}$. These non-perturbative pieces in the formulas are unphysical. The reference \cite{VerlindeSAP} discusses possible ways to constrain the non-perturbative contributions, and the non-perturbative structure from the point of view of resurgence is also studied in \cite{Griguolo}.

\paragraph{The near-hyperbolic expansion to all orders}

The near-hyperbolic expansion may be obtained straightforwardly by expanding the density \eqref{IKTVrho}. One can also use the known large $z$ asymptotic expansion of the Bessel function,
\be\label{I2Besselexp} I_{2}(z) = \frac{e^{z}}{\sqrt{2\pi z}}\sum_{k\geq 0} \frac{\Gamma(k-3/2)\Gamma(k+5/2)}{2^{k}\pi k! z^{k}}\,\cdotp\ee
Note that the (unphysical) non-perturbative corrections to \eqref{I2Besselexp} are proportional to $ie^{-z}$. One gets 
\be\label{ZIKTV1} \ln W^{-}_{\text{IKTV}} =\zeta_{0} - \beta_{\text S}\varepsilon_{0} + \frac{\pi}{4\beta_{\text S}} - \frac{3}{2}\ln\beta_{\text S} +\frac{\pi}{\beta_{\text S}} \sum_{k\geq 1}\Bigl(\frac{\pi}{\ell}\Bigr)^{2k} P_{k}(\beta_{\text S}/\pi)\ee
where $P_{k}$ is a polynomial of degree $k+1$,
\be\label{Pkexam}\begin{split} & P_{1}(x) = -\frac{1}{4} - 5 x - 15 x^2\, ,\quad P_{2}(x) = \frac{1}{2} + 10 x + 30 x^2 - 60 x^3\, ,\\ & P_{3}(x) = -\frac{5}{12}\bigl( 3 + 64 x + 216 x^2 - 576 x^3 + 432 x^4\bigr)\, ,\\& P_{4}(x) = \frac{7}{2} + 80 x + 300 x^2 - 960 x^3 + 1080 x^4 + 1440 x^5\, ,\\
& P_{5}(x) = -\frac{21}{2} - 256 x - 1050 x^2 + 3840 x^3 - 5400 x^4 - 11520 x^5 + 
 47520 x^6\, ,\quad\text{etc.}
\end{split}\ee
The first terms in \eqref{ZIKTV1} reproduce the Schwarzian partition function \eqref{ZRAexact}, which is the expected leading piece in the limit \eqref{Schlimit}. These leading terms get contributions only at tree-level and one-loop, as is well-known \cite{SWloca}. Interestingly, the above formulas also predict that the subleading corrections, at any given order $1/\ell^{2k}$, get contributions only up to $k+1$ loops, the loop counting parameter being here $\beta_{\text S}$. To summarize, the IKTV framework suggests two non-trivial properties:

i) The effective Schwarzian description can be extended to all orders in $1/\ell^{2}$.

ii) Within this effective description, a consistent loopwise expansion exists and there is a non-renormalization theorem, that generalizes the famous one-loop exactness at leading order, stating that the contributions at order $1/\ell^{2k}$ are $(k+1)$-loop exact, for all $k\geq 0$.

These properties are atypical for an effective, long-wavelength description. Usually, descriptions of this kind do not yield consistent loopwise expansions, since loop corrections probe arbitrarily short distance scales. And indeed, our discussion below will cast some doubts on the validity of these results, at least within our purely metric and two-dimensional point of view on JT gravity.

\paragraph{A semi-classical large $|\La|$, fixed $\ell$ expansion} 

The IKTV framework actually has even more powerful consequences than those we discussed above, and which do not seem to have been noted in the literature. By reordering the terms of the original large $\ell$, fixed $\beta_{\text S}=2\ell/|\La|$ expansion, it is indeed possible to get a new
\be\label{scVerlindeexp} \La\rightarrow -\infty\, ,\quad \ell\quad\text{fixed}\ee
expansion of the form
\be\label{ZIKTVLoop} \ln W^{-}_{\text{IKTV}} =\zeta_{0}-\beta_{\text S}\varepsilon_{0}+ |\La| f_{0}(\ell) + f_{1}(\ell,\La) + \sum_{L\geq 2}|\La|^{1-L}f_{L}(\ell)\, ,\ee
with
\begin{align} \label{treeJTneg} |\La|f_{0}(\ell) &= \frac{|\La|\ell}{16\pi}\Biggl[\sqrt{1+\Bigl(\frac{2\pi}{\ell}\Bigr)^{2}}-1\Biggr]\\
\label{oneloopJTZIKTV} f_{1}(\ell,\La) &= -\frac{3}{2}\ln\frac{2\ell}{|\La|} -\frac{5}{4}\ln\biggl[1+\Bigl(\frac{2\pi}{\ell}\Bigr)^{2}\biggr]\\
\label{highloopZIKTV} f_{L}(\ell) &= p_{L}\biggl[1+\Bigl(\frac{\ell}{2\pi}\Bigr)^{2}\biggr]^{\frac{1-L}{2}}\, ,\quad L\geq 2\, ,
\end{align}
with numerical constants $p_{L}$,
\be\label{pLlisthighloop}\begin{split} & p_{2} =-15\, ,\ p_{3} = -60\, ,\ p_{4}= -180\, ,\ p_{5} = 1440\, ,\ p_{6}=47520\, ,\ p_{7} = 921600\, ,\\ &
p_{8} = \frac{121953600}{7}\, , \ p_{9}= 373248000\, ,\ p_{10} = 9628761600\, ,\ \text{etc.}
\end{split}\ee
This expansion may be derived either from \eqref{IKTVint2} by using a saddle-point approximation or directly from \eqref{I2Besselexp}. Note that the usual expansion \eqref{ZIKTV1} can also be obtained from \eqref{ZIKTVLoop}, replacing $|\La|$ by $2\ell/\beta_{\text S}$, $\beta_{\text S}$ being fixed, using the fact that $|\La|^{1-L}f_{L}(\ell) \propto 1/\ell^{2L-2}$ for $L\geq 2$ and reordering the terms. The two expansions, Eq.\ \eqref{ZIKTV1} and Eq.\ \eqref{ZIKTVLoop}, thus contain exactly the same information, but presented in different ways.

The notable fact about the formulas \eqref{treeJTneg}, \eqref{oneloopJTZIKTV} and \eqref{highloopZIKTV} is that they show that the infinite series in $1/\ell^{2}$ contributing at a given order $|\La|^{1-L}$ according to IKTV have a finite radius of convergence and can thus be summed. In principle, we then obtain an information which is valid for finite values of $\ell$!

The term $\zeta_{0}-\beta_{\text S}\varepsilon_{0}$ in \eqref{ZIKTVLoop} corresponds to the arbitrary bulk Einstein-Hilbert and boundary cosmological constant counterterms, as already discussed in previous subsections. The leading order contribution \eqref{treeJTneg} matches with the renormalized classical action $\frac{|\La|}{16\pi}(A_{\text{cl}}-\ell + 2\pi)$. The contribution $-\ell + 2\pi$ can be absorbed in the counterterms $\zeta_{0}-\beta_{\text S}\varepsilon_{0}$ and $A_{\text{cl}}$ is the area of a round disk of boundary length $\ell$ embedded into hyperbolic space. This configuration corresponds to the maximal area metric and is thus the classical solution of JT gravity. This is all perfectly consistent with the idea that the expansion \eqref{ZIKTVLoop} could represent the standard quantum gravity semi-classical expansion, at fixed boundary length $\ell$, the semi-classical expansion parameter being $1/|\La|$, consistently with the fact that the classical action is proportional to $\La$; see also App.\ \ref{SemiclassApp}.

However, the extremely simple form of the result, Eqs.\ \eqref{oneloopJTZIKTV} and \eqref{highloopZIKTV}, appears to be hard to reconcile with some basic predictions of our microscopic definition of JT quantum gravity. Actually, the very existence of the semi-classical limit \eqref{scVerlindeexp} seems to be very difficult to make consistent with the microscopic formulation.

Indeed, as we have emphasized above, in particular in Section \ref{repaSec}, the smooth boundary length parameter $\ell$ is an effective, macroscopic parameter that makes sense only on length scales that are much larger than the curvature scale. Therefore, this parameter cannot enter, a priori, in formulas valid at strictly finite cut-off. Instead, one expects that loop corrections should generate an anomalous dimension for the length, revealing the fractal nature of the boundary. This is exactly what is found in the context of the semi-classical $c\rightarrow - \infty$ limit \`a la Zamolodchikov studied in \cite{Loopcalc}. Moreover, the UV structure of the theory, which is essential in making any loopwise expansion sensible, is not governed by the area term, but by an induced gravitational action that does not become semi-classical when $\La\rightarrow -\infty$ \cite{ferrari,ferraJTconfgauge}. Explicit evidence that this precludes the existence of a semi-classical expansion at $\La\rightarrow - \infty$ is given in \cite{Loopcalc} and general arguments in this direction will also be presented in Section \ref{continuumSec}.

Actually, a difficulty with the expansion \eqref{ZIKTVLoop} can be directly detected, by looking at the flat space limit of the theory. To discuss the flat space limit, it is convenient to introduce explicitly the curvature length scale, denoted by $L$, by writing the Ricci scalar $R=- 2/L^{2}$. The dimensionless parameters of the theory then become $\ell/L$ and $\La L^{2}$ and the flat space limit is $L\rightarrow\infty$ at $\ell$ and $\La$ fixed. In this limit, the tree-level term \eqref{treeJTneg} reduces to $\smash{\frac{|\La|\ell^{2}}{64\pi^{2}}}$, which is the expected result: the classical solution of flat JT gravity corresponds to a Euclidean disk of circumference $\ell$, which has area $\frac{\ell^{2}}{4\pi}$, and thus the classical action is $\smash{-\frac{|\La|}{16\pi}\times \frac{\ell^{2}}{4\pi} = -\frac{|\La|\ell^{2}}{64\pi^{2}}}$. However, at the one loop level, \eqref{oneloopJTZIKTV} yields a diverging result,
\be\label{Zzerooneloop} f_{1} = -\frac{3}{2}\ln\frac{2\ell}{|\La|L^{3}} -\frac{5}{4}\ln\biggl[1+\Bigl(\frac{2\pi L}{\ell}\Bigr)^{2}\biggr] =\frac{1}{2}\ln\frac{|\La|^{3}\ell^{2}L^{4}}{256\pi^{2}} + O\bigl(L^{-2}\bigr)\, .\ee
This seems to be a serious problem but should not come as a surprise. The macroscopic parameter $\ell$ does not make sense at all in flat space JT gravity. The flat space theory does not have a regime that can be described by a Schwarzian effective theory or by smooth boundary curves. It is impossible to obtain information about the flat space theory within the near-hyperbolic framework. The reliability of the would be one-loop result \eqref{oneloopJTZIKTV}, and, by extension, that of the would-be higher loop results \eqref{highloopZIKTV} or of the all-order near-hyperbolic expansion \eqref{ZIKTV1}, which indicate that the large $\ell$ series can be resummed to obtain results at finite $\ell$, must therefore be called into question.

%
%
%

%
\section{\label{bdviewSec}Curves that bound a distorted disk}

\subsection{\label{curvgeniSec}General introduction}

Models of (discretized) random paths are pervasive in theoretical physics and have also been studied extensively by mathematicians. The simplest and most important example is the random walk. The lattice version has a universal continuum limit called Brownian motion. The associated probability measure is the Wiener measure which is used, in particular, to define the path integral in Euclidean quantum mechanics. Other well-known examples of random paths are the Schramm-Loewner evolution processes, which are conjectured to describe the continuum limit of various models of physical interest: self-avoiding walks and polygons (SAW, SAP), percolation exploration paths, Ising model interfaces, etc. For some entry points in the literature on these subjects, see e.g.\ \cite{randpathlit} and the excellent textbook \cite{SAPtextbook}.

JT gravity is fundamentally a theory of random metrics. Yet, a description in terms of fluctuating boundaries seems very natural as well, at least when the base manifold is a disk. Actually, the fact that two equivalent dual formulations may exist, one in terms of random metrics and the other in terms of random paths, is conceptually non-trivial, since models of random geometries and models of random paths are usually associated with qualitatively different behaviours and different types of observables. In a sense, JT gravity seems to stand at a delicate equilibrium point between the two classes of models. This may be its distinctive and most interesting feature. 

Our goal in this Section, is to show that a formulation in terms of random closed paths is indeed possible and that this formulation is actually much more subtle than what has been envisioned so far. We are going to show that the relevant random path model is of a new kind, with unusual properties.

Focusing on the disk topology, for which there is only one boundary, the main features are as follows.

The set of allowed closed curves must obey a subtle constraint: \emph{only closed curves that bound a distorted disk are allowed}. The notion of distorted disk, both at the discretized level and in the continuum, has been introduced in Section \ref{diskviewSec}. Following a terminology found in the mathematical literature, we call the allowed curves \emph{self-overlapping polygons} (SOP) in the discretized case and simply self-overlapping curves (SOC) in the continuum.

The set of SOPs is strictly between the set of self-avoiding polygons SAPs and the set of arbitrary closed loops. This means that SAPs are SOPs, but a generic SOP is not a SAP; and that SOPs are closed loops (obviously), but a generic closed loop is not a SOP. As for SAPs, the constraint defining a SOPs is non-Markovian. However, whereas it is immediate to decide if a closed polygon is self-avoiding or not, the case of SOPs is much more subtle.

Fortunately, the problem of deciding whether a given closed curve is self-overlap\-ping or not has been studied extensively in the mathematical literature, from many different points of view (analytic, geometric, algorithmic), starting with the works of Titus and Blank \cite{titusblank}, as reviewed by Po\'enaru in \cite{Bourbaki}. The work of Blank completely solves the problem, using a nice method of ``cuts'' that will be illustrated below on examples. Motivated by possible applications in robotics and the design of integrated circuits, Shor and Van Wyk proposed an efficient algorithm which allows to determine if a curve is self-overlapping \cite{Shoralgo}. For a discretized curve of length $2n$, their algorithm runs in $O(n^{3}\ln n)$ time. A more recent algorithm was proposed in \cite{Mukherjeealgo}, with average running time $O(n^{2})$. There are also applications in the study of benzenoid hydrocarbons, see e.g.\ \cite{benzene} and references therein. Other references that we found very instructive include \cite{immersion}.

One of the most important message of \cite{ferrari} and the present paper is that the correct formulation of JT quantum gravity in terms of a fluctuating boundary is a \emph{theory of random self-overlapping polygons.} This provides an entirely new application of these objects and also emphasizes the importance of studying \emph{random} SOP, a subject on which virtually nothing is known.

The most surprising aspect of the model is that, in general, a given SOP does not determine a unique distorted disk; equivalently, it does not determine a unique metric. Since this is so important, and understandably very confusing, let us say it again: there exists cases where \emph{distinct} distorted disks, for which the associated disk metrics are physically different (which means that they are not equivalent modulo the action of disk diffeomorphisms), have exactly the \emph{same} boundary curve. This implies that the correct measure on the set of SOP is not the uniform measure: a non-trivial multiplicity index must be attached to each SOP, which counts the number of distinct metrics associated with it. This multiplicity can be computed, in principle, by using the Shor-Van Wyk or the Mukherjee algorithms \cite{Shoralgo,Mukherjeealgo}. At a conceptual level, this phenomenon implies that \emph{the degrees of freedom of JT quantum gravity are not in one-to-one correspondence with the shape of the boundary}.

\subsection{\label{bdcodeSec}SOPs, boundary code and symmetry factors} 

Consistently with Section \ref{diskviewSec}, we use a square lattice in our discussion; a similar treatment applies to triangular or hexagonal lattices. 

\subsubsection{Basic definitions}

A closed path, or polygon, on the square lattice, of length $\nu$, is a cyclic sequence of lattice points $(P_{0},P_{1},\ldots,P_{\nu})$, with $P_{k}=P_{k+\nu}$, such that $P_{k+1}$ is obtained from $P_{k}$ by walking along a single edge of the lattice (i.e.\ $P_{k}$ and $P_{k+1}$ are nearest neighbours), with the constraint that $P_{k+1}\not = P_{k-1}$ (we do not allow U-turns). Note that $\nu$ is necessarily even and thus we note $\nu=2n$. We have depicted in Fig.\ \ref{figdispath1} three typical polygons. 

\begin{figure}
\centerline{\includegraphics[width=6in]{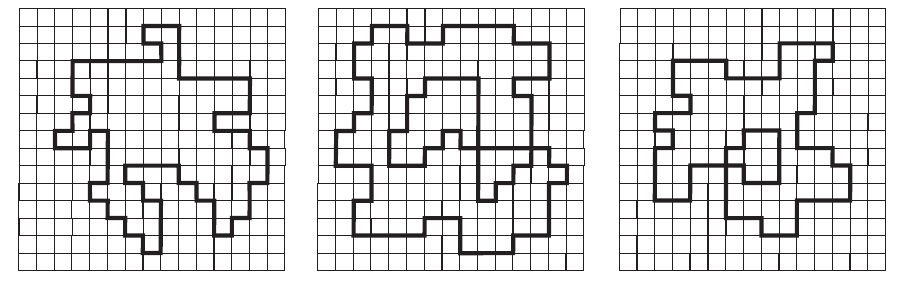}}
\caption{\label{figdispath1}Three instances of closed paths, or polygons, on the square lattice. The paths in the left and center insets are SOPs, but the path in the right inset is not, see the main text.}
\end{figure}

A path may be entirely coded by telling, at each step, in which direction we are going: Straight, Right, or Left.\footnote{For the hexagonal lattice, only Right and Left are available.} For instance, the ``path code'' for the simplest possible closed path, encircling a unique lattice tile, is $\text{LLLL}$. Of course, only the cyclic ordering of the letters in the code is important, since the starting point along the path may be chosen arbitrarily.\footnote{This is consistent with the fact, mentioned in Section \ref{diskviewSec}, that we do not have a marked point on the boundary.} The two possible orientations of a given path are considered to be equivalent, so two codes obtained by reversing the cyclic ordering and at the same time exchanging the L and R moves are identified. 

A self-overlapping polygon (SOP) is a closed path bounding a discretized flat disk (distorted disk) as defined in Section \ref{diskviewSec}. SOPs have a canonical ``anticlockwise'' orientation, which is such that the disk it bounds is on the left. The SOP code will always be chosen to be canonical. Note that the relation \eqref{EulergenflatJTdisk} is then automatically satisfied: the number of left turns minus the number of right turns is always equal to four. This rules out, for example, the polygon in the right inset of Fig.\ \ref{figdispath1}.

\subsubsection{\label{symfactSec}Counting and symmetry factors}

To define a theory of random SOPs, we need to define a measure on the set of SOPs. For now, let us limit ourselves to a few basic remarks. This discussion will be refined, in a subtle and surprising way, in Section \ref{MilnorSec}.

If we describe a SOP by its code, with the canonical orientation, two SOPs that coincide modulo lattice translations and rotations are automatically identified. Moreover, a reflection with respect to a vertical or horizontal axis acts by reversing the order of the letters in the path code. Because of this, we might consider the possibility of identifying a code and its reverse. Strictly speaking, this is the natural choice in the metric interpretation of the theory, but it turns out that it is not very convenient and in particular it does not match with the conventions used in Section \ref{Genfun1Sec}. Since we want to keep using the same conventions, we make the choice of not identifying a configuration with its image under a reflection. Moreover, we also count with an appropriate multiplicity SOPs that do not match under the action of lattice rotations. This amounts to counting any ``generic'' SOP path code with a factor of $4$. This choice will thus have a trivial effect in any well-defined continuum limit of the model.

Let us be a bit more precise. Let us pick a SOP of length $2n$. To each vertex $v$, we associate the ``rooted'' code $B_{v}$ at $v$, which by definition is the representative of the code for which the first letter corresponds to the vertex $v$. A SOP is \emph{generic} if the $2n$ rooted codes $B_{v}$, for any $v$, are all distinct. Generic SOPs are conventionally counted with a factor of 4.

More generally, pick a vertex $v$ and let $v'$ be the closest vertex following $v$ along the SOP such that $B_{v}=B_{v'}$. Assume that the distance between $v'$ and $v$ is $d$, which we note $v'=v+d$. Thus $B_{v}=B_{v+d}=B_{v+pd}$ for any integer $p$. The sequences of $d$ letters from $v+pd$ to $v+(p+1)d-1$ are the same for any $p$; let us note this sequence $\sigma$. Moreover, since $d$ is the smallest possible distance such that $B_{v}=B_{v+d}$, $d$ necessarily divides $2n$. We can thus write $2n = dr$ for some positive integer $r$. The code of the SOP is thus a repetition of $r$ times the sequence $\sigma$. If $v^{\sigma}_{\text L}$ and $v^{\sigma}_{\text R}$ represent the numbers of Left and Right vertices in the sequence $\sigma$, then, using \eqref{EulergenflatJT}, we get $v_{\text L}- v_{\text R} = 4 = r (v^{\sigma}_{\text L} - v^{\sigma}_{\text R})$. Thus $r$ divides 4. There are only three possibilities: either $r=4$, in which case the SOP has a $\mathbb Z_{4}$ symmetry corresponding to a 90 degree lattice rotation; or $r=2$ and the SOP has a $\mathbb Z_{2}$ symmetry corresponding to a 180 degree lattice rotation; or $r=1$ in which case the SOP is generic.

The precise counting rule is then that $\mathbb Z_{4}$-symmetric SOPs are counted once, $\mathbb Z_{2}$-symmetric SOPs are counted twice and generic SOPs and counted four times. 

An additional multiplicity factor, of a very different nature, will be introduced in Section \eqref{MilnorSec}.

\subsection{\label{SOPpropSec}General properties of self-overlapping curves}

We are now going to present criteria that allow to decide whether a closed curve is self-overlapping or not, by using information directly contained in the shape of the curve. This discussion allows to develop a better understanding of our problem and suggests some natural generalizations. Since the ideas apply equally well to discrete polygons and to continuous piecewise smooth curves, we deal with this latter case, which is more general. Note that in the continuum limit, which will be discussed in Section \ref{continuumSec}, these curves may become fractal. Strictly speaking, the discussion below applies before the continuum limit is taken.

Let us set up some notations and recall some of the results explained in Section \ref{immersionSec}. The source disk is as in Eq.\ \eqref{diskdef}, with coordinates 
\be\label{diskcoordapp} z = x^{1}+ i x^{2} = \rho e^{i\theta} \, .\ee
It is oriented by the volume form $\d x^{1}\wedge\d x^{2}$ (anti-clockwise orientation). The disk boundary $\partial\disk = \Sone$ is parameterized by $\theta\equiv\theta + 2\pi$. The target space $\mathbb R^{2}=\mathbb C$ is parameterized by Cartesian coordinates $(w^{1},w^{2})$ of by the complex coordinate $w = w^{1}+iw^{2}$.

We consider closed, continuous, piecewise smooth curves $\gamma:\Sone\rightarrow\mathbb R^{2}=\mathbb C$ such that $\gamma'\not = 0$, which is equivalent to saying that $\gamma$ is an immersion of $\Sone$ into $\mathbb C$.

The curve $\gamma$ is self-overlapping if and only if it can be extended into an immersion of $\disk$ into $\mathbb C$, that is to say, if and only if there exists an immersion $F:\disk\rightarrow\mathbb R^{2}$ such that $F_{|\partial\disk}=\gamma$. Without loss of generality, we always use orientation-preserving immersions, for consistency with the orientation conventions made up to now. We also know from Section \ref{immersionSec} that the immersion $F$ may be chosen to be holomorphic, in which case
\be\label{gamholoF} \gamma(\theta) = F(e^{i\theta})\, .\ee
We want to discuss necessary (and even, ideally, sufficient) conditions on $\gamma$ that ensure that the extension by an immersion $F$ exists. For the sake of simplicity, we assume that $\gamma$ is generic, in the sense that it has a finite number of self-intersection points and that the self-intersections are transverse and correspond to double points (the discussion can be straightforwardly adapted to the more general cases). 

Note that the most important points of the discussion in the present and the next subsection \ref{MilnorSec} are purely topological in nature and do not depend on the metric on the target space. They thus apply to the cases of zero, positive or negative curvature equally well.

\subsubsection{\label{WhitneySec}Whitney index}

Since the tangent vector $\gamma'$ never vanishes, it has a well-defined winding number
\be\label{Whitneyindex} \text W(\gamma) = \frac{1}{2i\pi}\int_{0}^{2\pi}\frac{\gamma''}{\gamma'}\, \d\theta\, .\ee
This is an integer, often called the Whitney index of the curve, which may also be seen as the winding number of the map from $\Sone$ to $\Sone$ defined by the normalized tangent vector $\gamma'/|\gamma'|$.\footnote{The normalized tangent vector $\gamma'/|\gamma'|$ is the unit tangent vector in the zero curvature case, when the target space is the Euclidean plane, but this is irrelevant for computing the winding number.} The Whitney index is an invariant of the curve modulo regular homotopies (homotopies $\gamma_{\la}$ that are immersions, $\gamma'_{\la}\not = 0$, for any homotopy parameter $\la$). Whitney theorem states that closed curves in the plane are classified by their Whitney index modulo regular homotopies.
\begin{proposition} Any self-overlapping curve with the anti-clockwise orientation has a  Whitney index $\text W(\gamma) = +1$.
\end{proposition}
\begin{proof} A simple proof uses the representation \eqref{gamholoF} in terms of a holomorphic function $F$. One gets in this way
\be\label{Whitp1} \text W(\gamma) = 1 + \frac{1}{2i\pi}\oint_{\partial\disk}\frac{F''}{F'}\, \d z\, .\ee
Since $F'\not = 0$ everywhere in the disk by the immersion property, $F''/F'$ is holomorphic in $\disk$ and the contour integral over $\partial\disk$ vanishes.

Another interesting proof uses the discretization. An immersion $F:\disk\rightarrow\mathbb C$ being given, we can always assume, even if we are dealing with the positive or negative curvature models, that the target space is endowed with the flat Euclidean metric. We can then use the discretization in terms of which the curve $\gamma$ becomes a polygon on the square lattice. It is clear that the total winding angle of the tangent vector is $\frac{\pi}{2}(v_{\text L}- v_{\text R})$; the Whitney index is this divided by $2\pi$ and is thus equal to one by \eqref{EulergenflatJTdisk}.
\end{proof}
\begin{corollary} Any self-overlapping curve has an even number of self-intersections.
\end{corollary}
\begin{proof} This comes from the fact that the Whitney index and the number of self-intersections have opposite parities. Indeed, the number of self-intersections is always preserved modulo 2 by a regular homotopy. Since Whitney theorem ensures that any self-overlapping curve, which must have $\text W(\gamma) = 1$, is regular-homotopic to a circle, and since a circle has no self-intersection, we conclude.
\end{proof} 

We immediately get an interesting generalization to arbitrary topologies. Imagine we have an immersion $F:\mathscr M_{g,b}\rightarrow \mathbb R^{2}$ from a surface of genus $g$ with $b$ circle boundaries into the plane. Let $\gamma_{1},\ldots,\gamma_{b}$ be the images of the boundary components of $\mathscr M_{h,b}$ by $F$. Then \eqref{EulergenflatJT} yields
\be\label{Whitneygen} \sum_{i=1}^{b}\text W(\gamma_{i}) = 2-2h-b\, .\ee

\subsubsection{\label{overlapSec}Number of overlaps}

\begin{figure}
\centerline{\includegraphics[width=6in]{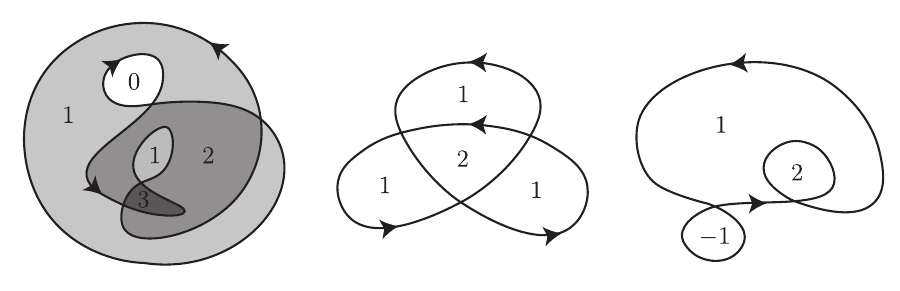}}
\caption{\label{overlapFig} Three examples of closed curves $\gamma$ for which we have indicated the winding numbers $\nu_{\gamma}$ in each bounded simply connected components of $\mathbb C\backslash\im\gamma$. The curve on the left is self-overlapping. It is a slightly more complicated version of the example in Fig.\ \ref{deformFig}. The curve in the center satisfies the constraints $\nu_{\gamma}(P)\geq 0$ everywhere but has an odd number of self-intersections (and Whitney index 2) and thus cannot be self-overlapping. It is actually an interior boundary, see Section \ref{BlankSec}. The curve on the right has Whitney index one, but there is a region for which $\nu_{\gamma}<0$ and thus it is not self-overlapping (nor an interior boundary).}
\end{figure}

Consider an arbitrary point $P$ in $\mathbb C\backslash\gamma$, with complex coordinate $w_{P}$. The winding number $\nu_{\gamma}(P)$ of $\gamma$ with respect to $P$ is defined in the usual way to be the winding number of the ray vector from $P$ to $\gamma$ when the curve is travelled all over once in the anti-clockwise direction,
\be\label{windingP} \nu_{\gamma}(P) = \frac{1}{2i\pi}\int_{0}^{2\pi}\frac{\gamma'}{\gamma - w_{P}}\, \d\theta\, .\ee
The set
\be\label{disunion} \mathbb C\backslash\im\gamma = \bigl(\cup_{i}\Omega_{i}\bigr)\cup \Omega_{\infty}\ee
is the disjoint union of simply connected open sets, which are all bounded except $\Omega_{\infty}$. The winding $\nu_{\gamma}(P)$ is integer-valued and constant in each components $\Omega_{i}$. Clearly, it is zero in the unbounded component.

If $\gamma$ is a self-overlapping curve, it bounds a distorted disk which may overlap with itself. The number of layers is a constant in each open set $\Omega_{i}$. We denote it by $\nu_{i}$. Clearly, this number is always positive and is zero in the unbounded component.

\begin{proposition}\label{windprop} Let $\gamma$ be self-overlapping. Then the number of layers $\nu_{i}$ in a bounded simply connected component $\Omega_{i}$ of $\mathbb C\backslash\im\gamma$ is equal to the winding number $\nu_{\gamma}(P)$ of $\gamma$ with respect to any $P\in\Omega_{i}$. In particular, $\nu_{\gamma}(P)\geq 0$ for any $P$.
\end{proposition}
\begin{proof} Using \eqref{gamholoF}, the integral \eqref{windingP} defining the winding number $\nu_{\gamma}(P)$ can be rewritten as
\be\label{nugamcont} \nu_{\gamma}(P) = \frac{1}{2i\pi}\oint_{\partial\disk}\frac{F'(z)}{F(z)-w_{P}}\,\d z\, .\ee
The contour integral picks a contribution 1 from each pole of the integrand, i.e.\ from each solution of $F(z)=w_{P}$ for $z\in\disk$. The number of solutions is the number of preimages of $P$ under $F$, which is precisely the number of layers at $P$. 
\end{proof}
This result is particularly interesting because the criterion $\nu_{\gamma}(P)\geq 0$ that any self-overlapping curve must satisfy is easy to check. It also yields a very simple way of counting the number of layers at any point without having to understand the details of the structure of the distorted disk bounded by the curve. This is actually quite non-trivial, as will become clear in the next subsection \ref{MilnorSec}.
\begin{corollary}\label{Areacorollary} Let $g$ be a constant curvature metric on $\disk$. Then the area $A[g]$ of $\disk$ with respect to $g$ depends only on the shape of the boundary of the associated distorted disk.
\end{corollary}
\begin{proof}
Writing as usual $\mathbb C\backslash\im\gamma = (\cup_{i}\Omega_{i})\cup\Omega_{\infty}$, we have
\be\label{areanuform} A[g] = \sum_{i}\nu_{i} A[\Omega_{i}]\, ,\ee
where $\nu_{i}$ is the number of layers of the distorted disk covering $\Omega_{i}$ and $A[\Omega_{i}]$ is the area of the component $\Omega_{i}$ computed with the relevant target space metric, which is given by \eqref{deltaE}, \eqref{deltaStwo} or \eqref{deltaHtwo} depending on whether we're considering the case of zero, positive or negative curvature.
\end{proof}
Three simple examples, with the winding numbers in each components explicitly indicated, are depicted in Fig.\ \ref{overlapFig}.

\noindent\emph{Remark 1}: a simple and convenient way to determine the winding around a point $P$ is to draw a generic oriented half-line from $P$ to infinity and to count the number of crossings of this half-line with the curve, with a plus sign if the curve crosses from right to left and a minus sign if the curve crosses from left to right.

\noindent\emph{Remark 2}: if $\gamma$ is an arbitrary closed curve, there is no obvious notion of ``area enclosed by $\gamma$.'' But we can always write $\mathbb C\backslash\im\gamma = (\cup_{i}\Omega_{i})\cup\Omega_{\infty}$ and define the winding numbers $\nu_{i}\in\mathbb Z$ with respect to the bounded simply connected components $\Omega_{i}$. This yields three natural notions of area,
\be\label{notionsofarea} A_{\text{L\'evy}}[\gamma] = \sum_{i}\nu_{i}A[\Omega_{i}] \, ,\quad A_{\text{wind.}}[\gamma] =  \sum_{i}|\nu_{i}| A[\Omega_{i}] \, ,\quad 
A_{\text{arith.}}[\gamma]= \sum_{i}A[\Omega_{i}]\ee
called the L\'evy (or algebraic, or stochastic), winding and arithmetic areas respectively. The algebraic area was studied long ago by Paul L\'evy in the context of Brownian motion \cite{Levy}. It plays an essential r\^ole in the definition of the Kitaev-Suh model \cite{KitaevSuh} and its generalizations, see Section \ref{randompathrev} and \cite{RDFollowup1}.

\subsubsection{\label{BlankSec}Blank cuts, Blank words, interior boundaries}

\begin{figure}
\centerline{\includegraphics[width=6in]{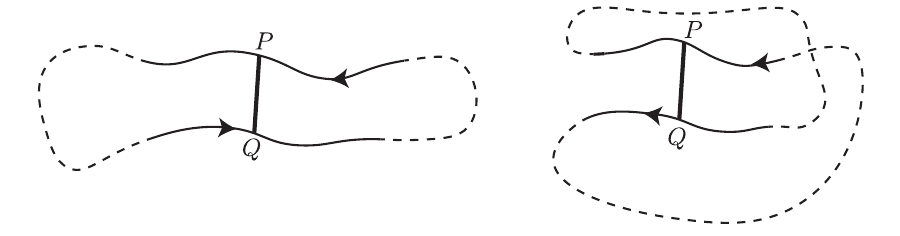}}
\caption{\label{cutFig} A Blank cut (left inset) and a non-Blank cut (right inset). The cut-path is in thick line.}
\end{figure}

A cut of an oriented closed curve $\gamma$ is a choice of two distinct points $P$ and $Q$ on $\gamma$ together with a cut-path from $P$ to $Q$ that crosses $\gamma$ transversally at $P$ and $Q$. To a cut of $\gamma$ is associated a decomposition of $\gamma$ into two closed oriented curves, one obtained by going from $P$ to $Q$ along $\gamma$ and then from $Q$ to $P$ along the cut-path, the other obtained by going from $Q$ to $P$ along $\gamma$ and then from $P$ to $Q$ along the cut-path. 

There are two types of cut. A \emph{Blank cut} is such that the curve crosses the cut-path from left to right at one of its endpoint and from right to left at the other endpoint. The two possibilities and the associated decomposition is shown in Fig.\ \ref{cutFig}. If $\gamma$ decomposes in $\gamma_{1}$ and $\gamma_{2}$, the Whitney indices are such that $\text W(\gamma_{1}) + \text W(\gamma_{2}) = \text W(\gamma) + 1$ for a Blank cut and $\text W(\gamma_{1}) + \text W(\gamma_{2}) = \text W(\gamma)$ otherwise.

It can be shown that a curve is self-overlapping if and only if a decomposition in terms of Blank cuts having the following properties exists. First, using a Blank cut, $\gamma=\gamma_{0}$ is decomposed into a simple (i.e.\ non self-intersecting) curve $\tilde\gamma_{0}$ and a new closed curve $\gamma_{1}$. If $\gamma_{1}$ is simple, the procedure stops. Otherwise, we use a new Blank cut to decompose $\gamma_{1}$ into a new simple curve $\tilde\gamma_{1}$ and $\gamma_{2}$. And we keep going in this way until, in the last step, the two curves produced by the decomposition are simple.

The intuition behind this approach is that the interiors of the simple curves $\tilde\gamma_{i}$ are pieces of the distorted disk that do not overlap, from which the full distorted disk can be  reconstructed by gluing along the cuts.

This procedure can be made systematic in the following way. In each bounded connected component $\Omega_{i}$ of $\mathbb C\backslash\im\gamma$, $1\leq i\leq m$, pick a point $P_{i}$ and draw an oriented half-line $L_{i}$ from $P_{i}$ to infinity in a generic way. Then build a word out of $m$ letters $a_{i}$ as follows. Go all over $\gamma$ starting from an arbitrary base point. Each time you cross a half-line $L_{i}$ from right to left, add the ``positive'' letter $a_{i}$ to the word. Each time you cross a half-line $L_{i}$ from left to right, add the ``negative'' letter $a_{i}^{-1}$ to the word. The result is called a Blank word. Since the base point on $\gamma$ is arbitrary, the Blank word is defined modulo cyclic permutations of the letters.

Note that the number of layers $\nu_{i}$ of the distorted disk on $\Omega_{i}$ is equal to the number of times the letter $a_{i}$ occurs minus the number of times $a_{i}^{-1}$ occurs in the Blank word.

A blank cut corresponds to a pair $(a_{i},a_{i}^{-1})$ in the word such that only positive letters appear between $a_{i}$ and $a_{i}^{-1}$ (or between $a_{i}^{-1}$ and $a_{i}$). A step in the Blank cut decomposition corresponds to ``reducing'' the word by removing $a_{i}$, $a_{i}^{-1}$ and the sequence of positive letters in between. The curve $\gamma$ is self-overlapping if and only if $\text W[\gamma]=1$ and the word can be reduced in such a way that at the end we are left with positive letters only. One can show that the conclusion does not depend on the set of generic points $P_{i}$ and half-lines $L_{i}$ we pick. The problem of recognizing a self-overlapping curve is thus reduced to a purely combinatorial analysis of the properties of the Blank word.

\begin{figure}
\centerline{\includegraphics[width=6in]{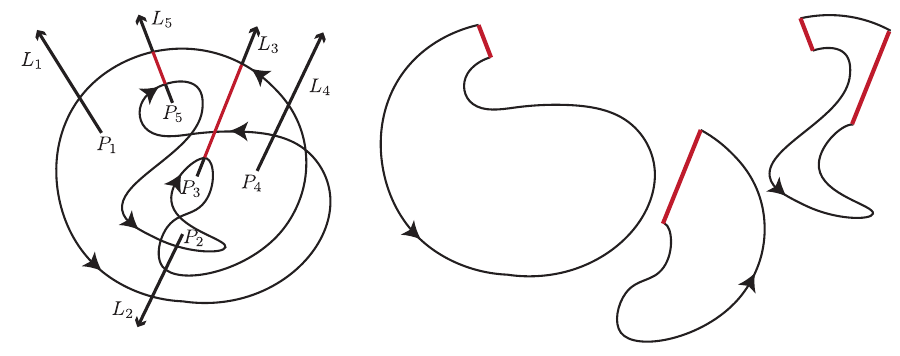}}
\caption{\label{BlankcutdecFig}Blank cut decomposition of a self-overlapping curve. The Blank word is $a_{5}a_{1}a_{2}a_{4}a_{3}a_{5}^{-1}a_{2}a_{3}^{-1}a_{2}a_{4}a_{3}$. It is reduced to $a_{2}$ by removing the sequences $a_{5}a_{1}a_{2}a_{4}a_{3}a_{5}^{-1}$ and then $a_{3}^{-1}a_{2}a_{4}a_{3}$. These sequences, together with the final word, represent the three simple curves depicted on the right. These simple curves bound pieces from which the distorted disk can be reconstructed, by gluing along the Blank cuts (outlined in red).}
\end{figure}

To illustrate the Blank word method, we have depicted in Fig.\ \ref{BlankcutdecFig} a Blank cut decomposition for the self-overlapping curve on the left of Fig.\ \ref{overlapFig}.

\begin{figure}
\centerline{\includegraphics[width=6in]{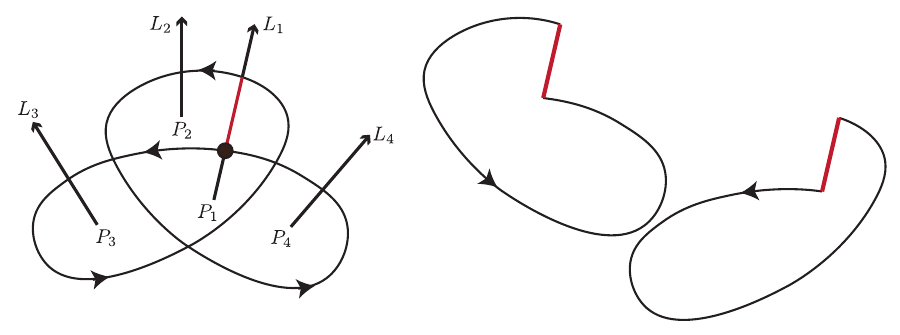}}
\caption{\label{interiorFig}An interior boundary with Blank word $a_{1}a_{2}a_{4}a_{1}a_{3}$. It can be reconstructed by gluing the interior of two simple curves, shown on the right, corresponding to the non-Blank cut $a_{1}a_{2}a_{4}a_{1}$. The associated map $F:\disk\rightarrow\mathbb C$ has a branch point, indicated by a black bullet.}
\end{figure}

An interesting generalization corresponds to the case when a decomposition in terms of simple curves is possible, but only if one allows both Blank and non-Blank cuts. This amounts to allow pairings of $a_{i}$ with itself (or of $a_{i}^{-1}$ with itself) in the Blank word. An example is depicted in Fig.\ \ref{interiorFig}. Curves with this property are called \emph{interior boundaries} in the mathematical literature. Note that the configurations depicted in the second and third insets of Fig.\ \ref{forbidFig} are interior boundaries, their boundary curve being of the type of the example depicted in the center of Fig.\ \ref{overlapFig}. Interior boundaries clearly satisfy the positivity constraint $\nu_{\gamma}(P)\geq 0$ on the winding numbers, but their Whitney index may be an arbitrary positive integer. One can show that interior boundaries of Whitney index $k$ may be splitted by removing recursively $k$ self-overlapping pieces corresponding to closed arcs of the curve starting and ending at a  self-intersection point. In particular, interior boundaries that have $\text W=1$ are self-overlapping.

The defining property of a self-overlapping curve $\gamma$ is that there exists an immersion $F:\disk\rightarrow\mathbb C$ such that $F_{|\partial\disk} = \gamma$. In the case of an interior boundary, one still has a map $F:\disk\rightarrow\mathbb C$ such that $F_{|\partial\disk} = \gamma$, but $F$ is only constrained to be open (the image of an open set is an open set), light (the preimage of a point is a collection of distinct points) and orientation-preserving. Modulo a disk diffeomorphism, one can always assume that $F$ is analytic. For a self-overlapping curve, $F'\not = 0$ everywhere, but for an interior boundary $F$ may have a finite number of branch points around which it behaves as $F\sim (z-z_{0})^{r}$ for $r\geq 2$.\footnote{See e.g.\ \cite{marx} for details.} Interior boundaries thus do not correspond to smooth metrics on the disk.

\subsection{  \label{MilnorSec}The Milnor curve and multiplicity}

\begin{figure}
\centerline{\includegraphics[width=6in]{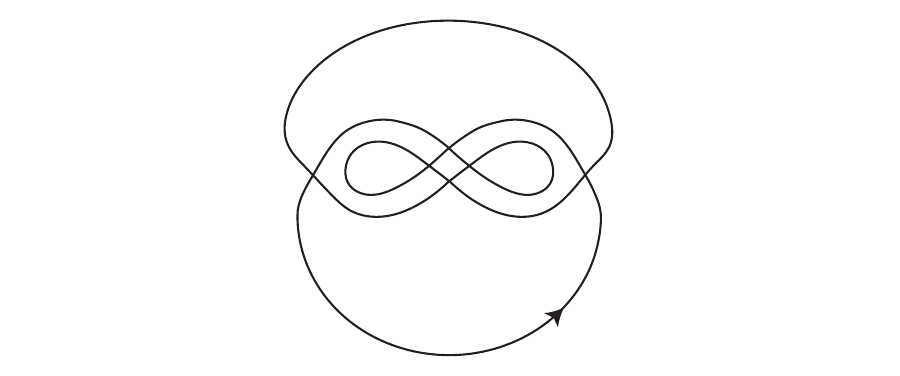}}
\caption{\label{figMilnor}The Milnor curve.}
\end{figure}

When the boundary curve is a self-avoiding loop, it separates the plane into two regions, the interior and the exterior. The associated distorted disk is then easily and uniquely identified with the interior region. When the boundary curve has self-intersections, it is no longer obvious, in general, to decide whether it bounds a distorted disk or not, and if it does, to reconstruct the distorted disk explicitly. A systematic procedure to solve this problem, using Blank cuts and words, has been outlined in the preceding subsection.

Let us now consider the closed curve depicted in Fig.\ \ref{figMilnor}. We call this curve the Milnor curve.\footnote{It seems to be acknowledged that this curve was first found by Milnor, even though his discovery was never published.} This is a typical example where a simple visual inspection is not enough to immediately decide whether it bounds a distorted disk or not. But we can straightforwardly compute a Blank word and look at its possible reductions. We let the reader do this very instructive exercice carefully. 

\begin{figure}
\centerline{\includegraphics[width=6in]{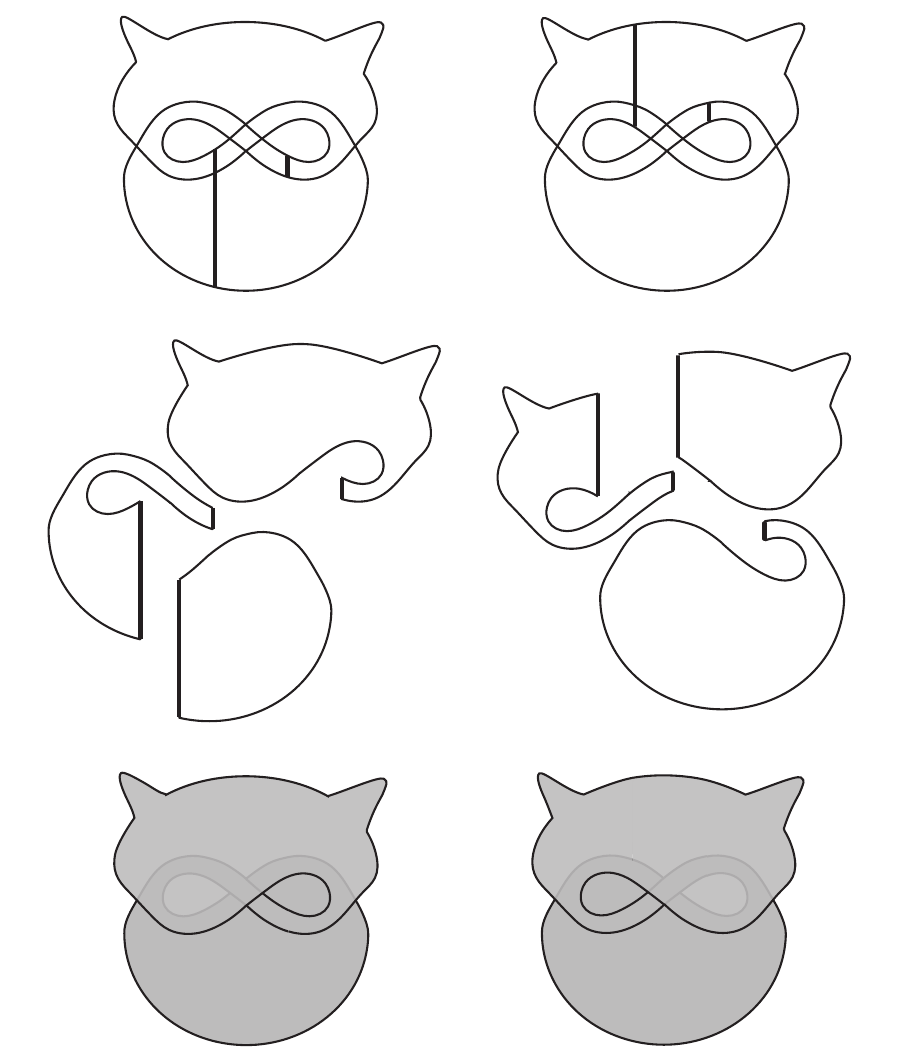}}
\caption{\label{figMilnor2}The Milnor curve bounds two distinct distorted disks, depicted on the left and on the right. Upper insets: the Milnor curves with two distinct sets of Blank cuts. Central insets: the resulting simple curves. Lower insets: reconstruction of the disks by gluing the simple pieces along the cuts.}
\end{figure}

The result is extremely interesting: one finds that the curve is self-overlapping, but \emph{it actually bounds two distinct distorted disks!}

The two disks correspond to two distinct reductions of a Blank word, using distinct set of cuts. This is depicted in Fig.\ \ref{figMilnor2}. It is important to realize that the pieces making up the two distorted disks are genuinely different and could not be made identical by cutting in different ways. Mathematically, this means that the two  immersions $F_{i}:\disk\rightarrow\mathbb C$, $i=1$ or 2, corresponding to the two distorted disks, are inequivalent in the sense that there is no diffeomorphism $\phi$ of the disk such that $F_{1} = F_{2}\circ\phi$. In particular, and this is the most relevant point for us, the two disk metrics $g_{1}=F_{1}^{*}\delta$ and $g_{2}=F_{2}^{*}\delta$ are physically distinguishable.

This fundamental example shows that the correspondence between self-overlapping curves and disk metrics is not one-to-one. In JT gravity language, this is equivalent to saying that \emph{the degrees of freedom of the theory (the space-time metrics) are not in one-to-one correspondence with the shape of the allowed (self-overlapping) boundaries}.

\begin{definition}\label{multDef}
The number of distorted disks bounded by a closed curve $\gamma$ is called the multiplicity of $\gamma$ and is denoted by $\mu_{\gamma}$.
\end{definition}
The integer $\mu_{\gamma}$ can always be determined systematically by counting the number of distinct reductions of a Blank word of the curve or, equivalently, by running the algorithm of Shor and Van Wyk.

Let us mention two general and useful qualitative features.

i) If $\mu_{\gamma}>1$, then one can show that there must exist a simply connected component of $\mathbb C\backslash\im\gamma$ for which the number of layers is $\nu\geq 3$. This result was conjectured in \cite{Shoralgo} and proven by Graver and Cargo \cite{immersion}; see also \cite{benzene}. In the case of the Milnor curve, Fig.\ \ref{figMilnor} or \ref{figMilnor2}, the small central connected component of $\mathbb C\backslash\im\gamma$, roughly square in shape, is covered three times by any of the two associated distorted disks.

ii) As explained in Section \ref{overlapSec}, the number of layers $\nu_{i}$ in each simply connected components of $\mathbb C\backslash\im\gamma$ is entirely determined by the shape of the boundary curve only. The sets of integers $\nu_{i}$ thus match for all the distinct distorted disks bounded by a given curve. Suppose that $\mu_{\gamma}\geq 2$ and denote by $g_{1},\ldots,g_{\mu_{\gamma}}$ the corresponding  metrics. It then follows from Corollary \ref{Areacorollary}, or equivalently from Eq.\ \eqref{areanuform}, that the areas of the disk calculated with these different metrics must all be equal,
\be\label{Areaid} A[g_{1}]=A[g_{2}]=\cdots = A[g_{\mu_{\gamma}}]\, .\ee
This is non-trivial because the metrics $g_{i}$ are not equivalent under the action of a diffeomorphism.

\noindent\emph{Remarks}

i) The Milnor curve can be generalized in a rather straightforward way to construct examples that have an arbitrary integer multiplicity, see Fig.\ 3 of \cite{Bourbaki} and the discussion in Section \ref{continuumSec}.

ii) In the case of interior boundaries, there will typically be a moduli space associated with a given curve, since the positions of the branching points depend on the particular choice of Blank cuts. An example of this is given by the configurations in the second and third insets of Fig.\ \ref{forbidFig}.

\subsection{\label{Genfun2Sec}Generating functions and generalizations}

\subsubsection{Generating functions}

We now have all the ingredients we need to define the JT gravity models from the point of view of the boundary curve. We focus here on the zero curvature case. The cases of negative and positive curvatures will be briefly discussed in Section \ref{posnegcurvesSec} and  \ref{nonzeroRSec}.

The generating (or partition) function is
\be\label{Zdefbdcurve} W(t,g) = \sum_{\gamma}\mu_{\gamma} t^{F(\gamma)}g^{v(\gamma)} = \sum_{n,p\geq 1} W_{2n,p}t^{p}g^{2n}\, .\ee
The sum is over all self-overlapping polygons, counted as explained in Section \ref{bdcodeSec}, with the additional multiplicity factor $\mu_{\gamma}$ defined in Section \ref{MilnorSec}. $F(\gamma)$ is the number of faces of the associated distorted disks, which depends only on $\gamma$, according to \eqref{Areaid}, and $v(\gamma)$ is the boundary length or equivalently the number of letters in the boundary code.

The definition \eqref{Zdefbdcurve} of the generating function has been carefully designed to be equivalent with the definition in Section \ref{Genfun1Sec}.

\begin{figure}
\centerline{\includegraphics[width=6in]{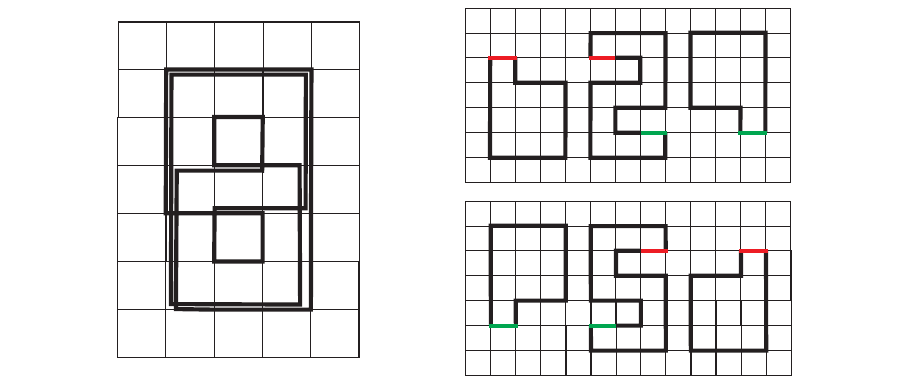}}
\caption{\label{figMilnormin}The smallest SOP with multiplicity two on a square lattice, with boundary code $\text{SSSSLSSLSSLSRRRRSLSSLSSLSSSSLSSLSSLSRRRRSLSSLSSL}$, $2n=48$ and $p=31$. This configuration is $\mathbb Z_{2}$-symmetric (invariant under a 180 degree rotation) and has multiplicity 2. It is thus counted $2\times 2 = 4$ times in the generating function. Similar Milnor configurations with no symmetry can be easily obtained, for instance by adding a face at the bottom, with new boundary code $\text{SSSSLSSLSSLSRRRRSLSSLSSLSSSSLSSLSSLSRRRRSLSSL}\textbf{RLLR}\text{L}$ (the added sequence is in boldface). Such a configuration must be counted 8 times in the generating function.}
\end{figure}

As already mentioned in \ref{Genfun1Sec}, at low orders, the generating function $W$ matches with the standard generating function for self-avoiding polygons. The counting starts to deviate at order $g^{18}t^{8}$, for which there exists a SOP that is not a SAP (see Fig.\ \ref{figlowdisk}). Moreover, the lowest order at which a SOP with higher multiplicity appears is $t^{31}g^{48}$.\footnote{We have not tried to prove this fact rigorously. For an analysis on the hexagonal lattice, see \cite{hexamin}.} This configuration is a Milnor-like polygon, depicted in Fig.\ \ref{figMilnormin}. 

The fact that higher multiplicities first appear at relatively high orders in the expansion of the generating function may convey the impression that higher multiplicity curves are very rare, but this is misleading. A simple argument given in Section \ref{continuumSec} shows that the number of SOPs of length $2n$ that have a higher multiplicity is exponentially large in $n$. We shall also see that, for a given SOP, the multiplicity index can also be exponentially large in $n$.\footnote{I would like to thank Peter Shor for pointing this out to me.}

The generating function \eqref{Zdefbdcurve} can be generalized to a more refined version in a natural way. Let $\mathscr W_{v_{\text L},v_{\text S},v_{\text R},p,\mu}$ be the number of SOPs such that the boundary code contains $v_{\text L}$, $v_{\text S}$ and $v_{\text R}$ letters L, S and R respectively, of total area $p$ and of multiplicity $\mu$. The generalized generating function is
\be\label{Zdefgeneralized} \mathscr W (t,g_{\text L},g_{\text S},g_{\text R},u) =  \sum_{\substack{p\geq 1,\, v_{\text L}\geq 4\\v_{\text S},v_{\text R},\mu\geq 0}} \mathscr W_{v_{\text L},v_{\text S},v_{\text R},p,\mu} t^{p}g_{\text L}^{v_{\text L}}g_{\text S}^{v_{\text S}}g_{\text R}^{v_{\text R}}u^{\mu}\, .\ee
Note that, due to the constraint \eqref{EulergenflatJT}, there is a trivial redundacy in the parameters.

\subsubsection{Open self-overlapping curves}

It is possible to define a natural notion of open self-overlapping curves, joining two distinct points $w_{1}$ and $w_{2}$ in $\mathbb C$. We consider a fixed oriented curve $\gamma_{21}$ from $w_{2}$ to $w_{1}$. An obvious possibility is to choose a straight line, which would correspond to ``chordal self-overlapping curves,'' but more generally $\gamma_{21}$ may be an arbitrary continuous curve from $w_{2}$ to $w_{1}$. We then say that an oriented open curve $\gamma$ from $w_{1}$ to $w_{2}$ is self-overlapping if and only if the closed curve $\gamma_{21}\cup\gamma$ is self-overlapping.

This notion is not directly relevant for the quantum JT gravity application that we study in this paper, but investigating random open self-overlapping curves may present some mathematical interest.

\subsubsection{Remarks on other topologies}

\begin{figure}
\centerline{\includegraphics[width=6in]{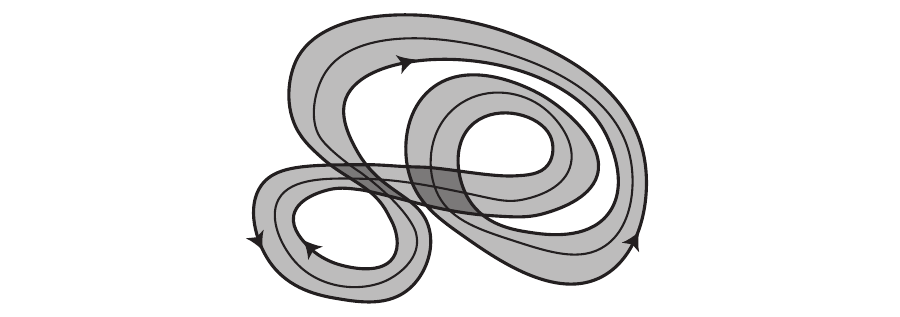}}
\caption{\label{cylFig}A distorted annulus obtained by the thickening of a non self-overlapping closed curve, which is depicted in thin blackline.}
\end{figure}

The concept of distorted disks explained in Section \ref{diskviewSec} can be straightforwardly generalized to the concept of distorted surfaces of arbitrary topologies, with genus $h$ and $b$ boundary components. In the continuum, distorted surfaces are described by immersions $F:\mathscr M_{h,b}\rightarrow\mathbb C$ that, by pullback, yield constant curvature metrics on $\mathscr M_{h,b}$.

For higher topologies, however, the point of view in terms of the boundary curves becomes extremely subtle, even more so than the disk case discussed above. The picture that emerges seems staggering. 

Even for the annulus, which seems to be the simplest case beyond the disk, the ideas developed for the disk are no longer applicable. For intance, one can readily see that boundary curves of arbitrary shapes are allowed, by ``thickening,'' see Fig.\ \ref{cylFig}. At a deeper level, it was shown in \cite{EppsteinMumford} that the problem of determining whether a curve bounds an arbitrary immersed surface or not is NP-complete, in sharp contrast with the case of the disk which is polynomial \cite{Shoralgo}.

The case of higher topologies will be briefly discussed again in Sections \ref{mmSec} and \ref{OutSec}.

\subsubsection{\label{posnegcurvesSec} The cases of non-zero curvature, light version}

The discretized formulation of flat JT gravity presented in Section \ref{diskviewSec} or above in Section \ref{bdviewSec} is based on the use of a regular tesselation of the flat Euclidean space, for instance the square lattice. This can be straightforwardly generalized to constant negative or positive curvatures, by using tessellations of hyperbolic space or the two-sphere, but these generalizations comes with interesting subtleties and limitations.

In negative curvature, there is an infinite number of regular tessellations $\{p,q\}$ made of regular $p$-gons and degree $q$ vertices. The edges and vertices associated with these tessellations form the so-called regular hyperbolic lattices. The only constraint that the integers $p$ and $q$ must satisfy is
\be\label{hyperbolictesscons} (p-2)(q-2)>4\, ,\ee
which follows from the formula for the curvature given in \eqref{Rpqform}. The hyperbolic lattices are commonly used to discretize/approximate hyperbolic space, see e.g.\ \cite{hyplatticeapp} for recent interesting applications. Random self-avoiding walks have also been studied on these lattices \cite{SAWhyper,SAWhyperball}. A natural approach to JT quantum gravity in negative curvature is thus to consider models of random self-overlapping polygons, counted with multiplicity, on the regular hyperbolic lattices. As will be explained in Section \ref{nonzeroRSec}, these models may have a very interesting physical interpretation, but they are not suitable to define JT gravity in a UV-complete way.

In positive curvature, regular tessellations correspond to the five Platonic solids. In particular, they have a finite number of tiles. Such discretizations are clearly inadequate to approximate JT gravity in any useful way. More general irregular tesselations of the two-sphere must be used to obtain an interesting model.

We refer to Section \ref{nonzeroRSec} for a more detailed discussion of the models in non-zero curvature.

\section{\label{continuumSec}On random self-overlapping curves}

The aim of this section is to provide a basic introduction to the study of random self-overlapping curves, a subject on which nothing seems to have ever been published. Our goal will be to outline some interesting expected properties of the model, emphasizing, in particular, the continuum limit. We focus on the flat theory but we also discuss, mainly at the qualitative level, the cases of negative and curvature. The microscopic properties are the same, whether the curvature in positive, zero or negative, but the macroscopic features will differ, even qualitatively, between the three cases.

We begin the section with two review subsections on some aspects of standard Liouville gravity and random curves models. The problem we are studying is subtly lying somewhere in between these two classes of models and having a synthetic presentation of some relevant features at hand is very useful (and very hard to find in the existing literature). Our goal is to put our analysis into an appropriate context, introduce some important concepts, refine our intuition and motivate some of the ideas and conjectures we will propose.  References include \cite{Liouvillesuccess2,FZZT,gravityreviews,Liourevref} on the Liouville side and \cite{Levy,Brownian1} and \cite{SAPcounting,randpathlit,SAPtextbook} for Brownian and self-avoiding curves respectively, as well as many additional references therein.

\subsection{\label{RevLSec}Review of Liouville gravity}

\subsubsection{\label{pureLiouSec}Pure Liouville gravity on the sphere}

\paragraph{Discretized formulation}

The standard two-dimensional gravity is a theory of random metrics on a two-dimensional surface of fixed topology.\footnote{One may also consider the possibility to sum over all the possible topologies, but this aspect is not discussed in the present paper.} It can be defined rigorously by using the ideas reviewed in Section \ref{dissurfrevSec}, starting from the set of discretized metrics obtained by quadrangulating the surface, and then taking an appropriate continuum limit.\footnote{The continuum limit is universal and does not depend on the type of polygons we use. We consider squares for simplicity.} The counting of the possible configurations was performed long ago by Tutte \cite{Tutte} by using combinatorial arguments. The counting can also be made elegantly by solving a quartic matrix model in the large $N$ limit \cite{BIPZ}.

The generating function
\be\label{WLioudef} \mathsf W^{(\text s)}(t) = \sum_{p\geq 0}\mathsf W^{(\text s)}_{p}t^{p}\ee
encodes the numbers $\mathsf W^{(\text s)}_{p}$ of quadrangulations of the sphere with $p$ squares.\footnote{The precise definition of $\mathsf W^{(\text s)}_{p}$ involves symmetry factors attached to each discretized surface, according to their automorphism group. This is similar to the discussion in Section \ref{symfactSec}, see also Section \ref{mmSec}. In the continuum limit, these symmetry factors play a trivial r\^ole.} In terms of a formal matrix integral over an $N\times N$ Hermitian matrix $X$, we have
\be\label{formalMM1} e^{N^{2}\mathsf W^{(\text s)}(t)} = \int\!\d X\, e^{-N(\frac{1}{2}\tr X^{2}-\frac{t}{4}\tr X^{4})}\, ,\ee
where it is understood that only planar diagrams are kept.\footnote{Note that we use here the ``dual'' point of view, as defined at the beginning of Section \ref{dissurfrevSec}, the matrix model vertex $\tr X^{4}$ being associated with a square of the sphere quadrangulation.} The series \eqref{WLioudef} must clearly diverge when $t=1$, since there is an infinite number of possible quadrangulations of the sphere. One finds
\be\label{WpspeTutte} \mathsf W_{p}^{(\text s)} = \frac{3^{p}(2p)!}{2 p(p+1)(p+2)p!^{2}}\,\cdotp\ee
When $p\rightarrow\infty$, which corresponds to large surfaces,
\be\label{WspasymnLiou} \mathsf W^{(\text s)}_{p} \underset{p\rightarrow\infty}{\sim} \frac{1}{2\sqrt{\pi}}\frac{12^{p}}{p^{7/2}}\, \cdotp\ee
This implies in particular that the series \eqref{WLioudef} has a strictly positive radius of convergence  $t_{*}$ such that
\be\label{tcritLiou1} \lim_{p\rightarrow\infty}\bigl(\mathsf W^{(\text s)}_{p}\bigr)^{1/p} = \frac{1}{t_{*}}=12\, .\ee

\paragraph{Continuum limit}

This $p\rightarrow\infty$ behaviour allows to define a continuum limit as follows. Consider the generating function at fixed area $W^{(\text s)}_{p}t^{p}$. Introducing the Planck length $\ell_{0}$, the area is 
\be\label{ArealPlanck} A=p \ell_{0}^{2}\, .\ee
The factor $t^{p}$ can be written as 
\be\label{tpcc} t^{p} = e^{p\ln t} = e^{-\frac{\La_{0}}{16\pi} A}\, ,\ee
where $\La_{0} = \frac{16\pi}{\ell_{0}^{2}} \ln\frac{1}{t}$ plays the r\^ole of a bare cosmological constant.\footnote{The factors of $16\pi$ are included to be consistent with the usual conventions for the cosmological constant used in the context of JT gravity, see e.g.\ Eq.\ \eqref{qJT1}.} A non-trivial continuum limit
\be\label{climit1} p\rightarrow\infty\, ,\quad \ell_{0}\rightarrow 0\, ,\quad A = p\ell_{0}^{2}\quad \text{fixed}\ee
is then obtained by adjusting the bare cosmological constant in such a way that
\be\label{climit2} t\rightarrow t_{*}^{-}\quad \text{and}\quad \La = \frac{16\pi}{\ell_{0}^{2}}\Bigl(1-\frac{t}{t_{*}}\Bigr)\geq 0\quad \text{is fixed.}\ee
The intuition behind this scaling is simple: when $t\rightarrow t_{*}$ from below, one can easily check that the average number of squares in the quadrangulations contributing to \eqref{WLioudef} diverges as $1/(t_{*}-t)$. The scaling thus ensures a finite macroscopic area. 

Note that the renormalized cosmological constant obtained is this way is always positive, because the limit has to be taken from below, $t<t_{*}$. There is no getting around this fact: the generating function \eqref{WLioudef} is ill-defined for $t>t_{*}$ because the number of quadrangulations of the sphere is infinite. As a consequence, one cannot define a continuum limit for which $t\rightarrow t_{*}$ from above. This is not surprising: for strictly negative $\La$, the classical action is not bounded from below and one expects the theory to be inconsistent.

Note also that the precise value of $t_{*}$ is irrelevant in the continuum limit. This is consistent with the fact that this value is not universal and depends for example on the type of polygons we use to discretize the surface. The only universal property of $t_{*}$ is that $0<t_{*}<1$. On the other hand, the exponent $7/2$ in \eqref{WspasymnLiou} is universal. The generating function at fixed area \eqref{WspasymnLiou} scales as\footnote{Here and below, the proportionality sign $\propto$ indicates that the generating functions in the continuum limit are fixed up to a trivial, parameter-independent, multiplicative constant, that may involve powers of the cut-off $\ell_{0}$.}
\be\label{Wsclsphere} \mathsf W^{(\text s)}(A) \propto A^{-7/2}e^{-\frac{\La}{16\pi} A}\, .\ee
This is a famous result, which can be generalized to random closed surfaces of arbitrary genus $h$,
\be\label{WgenclLiou} \mathsf W^{(h)}(A) \propto A^{-1 + (\gamma_{\text{str}} -2)(1-h)}e^{-\frac{\La}{16\pi} A}\, .\ee
The exponent $\gamma_{\text{str}}$, called the string susceptibility, is $\gamma_{\text{str}}=-1/2$ in the case of pure gravity and takes the more general KPZ form
\be\label{gammaKPZ} \gamma_{\text{str}} = \frac{1}{12}\Bigl(c-1-\sqrt{(c-1)(c-25)}\Bigr)\ee
when two-dimensional gravity is coupled to an arbitrary matter CFT of central charge $c < 1$. There is a so-called barrier at $c=1$, beyond which Eq.\ \eqref{gammaKPZ} predicts a non-sensical imaginary string susceptibility. At $c=1$, logarithmic corrections to the scaling law \eqref{WgenclLiou} are believed to be present \cite{barrierrefs}, with
\be\label{barrierscaling} W^(\text s)(A)\propto A^{-3}(\ln A)^{-2}e^{-\frac{\La}{16\pi}A}\, .\ee

\subsubsection{\label{LiouclimitSec}Pure Liouville gravity on the disk}

To come closer in spirit to the problem we are studying in the present paper, let us make one more step in the framework of Liouville gravity and review the case of the disk topology. On the disk, we can define in Liouville gravity the exact analogues of the JT gravity generating functions \eqref{Zdef1} and \eqref{Zfixed}. To avoid confusion and for the sake of consistency, we continue to use the notation $W$ for JT gravity and use instead $\mathsf W$ for Liouville gravity. The difference between $\mathsf W$ and $W$ is that, in the first case, we sum over all possible discretized disk metrics, whereas in the second case we restrict ourselves to the metrics of zero curvature.\footnote{Or more generally to constant curvature, see Section \ref{posnegcurvesSec} and \ref{nonzeroRSec}.} 

Let us mention right away that thinking that JT has to be simpler than Liouville because we have to take into account less metrics is na\"\i ve and probably misleading: problems with constraints are typically more difficult.\footnote{A famous illustration of this ``principle'' is the Ising model. The exact Onsager solution of the model on a fixed flat lattice is much more difficult than the solution on a random lattice \cite{Isingrandom}.}

\paragraph{Discretized formulation}

From direct combinatorial counting \`a la Tutte or from the solution of the matrix model, the Liouville disk generating function can be computed exactly, with only slightly more effort than in the case of the sphere. The number of quadrangulations of the disk of area $p$ and boundary length $2n$ is found to be
\be\label{WpndiskL} \mathsf W_{2n,p} = 3^{p}\frac{(2n-1)!(2p+n-1)!}{n!(n-1)!(p+n+1)!p!}\,\cdotp\ee
As a consistency check, note that $\mathsf W_{4,p-1}$ counts the number of quadrangulations of the sphere with $p-1$ unmarked and one marqued squares, the marked square corresponding to the boundary of length four. This implies that $\mathsf W_{4,p-1}=p \mathsf W_{p}^{(\text s)}$, where $\mathsf W_{p}^{(\text s)}$ is given in Eq.\ \eqref{WpspeTutte}. This is indeed satisfied.

The generating function at fixed boundary length is found to be
\be\label{W2ndL}\mathsf W_{2n}(t) = \sum_{p\geq 1}\mathsf W_{2n,p}t^{p} = \frac{(2n-1)!}{2^{2n-1}n!(n+2)!}a^{2n}(t)\bigl(n+1- \frac{1}{8}n a^{2}(t)\bigr) - \frac{(2n-1)!}{n!(n+1)!}\, \cvp\ee
where
\be\label{adefL} a^{2}(t) = \frac{2}{3t}\bigl(1-\sqrt{1-12 t}\bigr)\ee
is a solution of $3ta^{4}-4a^{2}+16 = 0$.\footnote{The interval $[-a,a]$ represents the support of the planar matrix model eigenvalue distribution.} 

From the point of view of the matrix model \eqref{formalMM1}, it is interesting to note that the generating function $\mathsf W_{2n}(t)$ is a simple expectation value,
\be\label{W2ndLMM} \mathsf W_{2n}(t) = \frac{1}{2n}\bigl\langle\frac{1}{N}\tr X^{2n}\bigr\rangle\, .\ee
Indeed, in the dual Feynman graph point of view, the insertion of $\tr X^{2n}$ corresponds to a marked face, which is removed to get the disk topology. The factor $1/(2n)$ takes into account the fact that all the edges of the boundary are equivalent (there is no marked edge on the boundary).\footnote{For the same reason, there is a factor $1/4$ in the interaction term in the matrix integral \eqref{formalMM1}.} One can also show straightforwardly that the full generating function
\be\label{ZfullLM} \mathsf W(t,g) = \sum_{n\geq 1}\mathsf W_{2n}(t)g^{2n}\ee
satisfies
\be\label{ZfulldZid} g\frac{\partial\mathsf W}{\partial g} = \sum_{n\geq 1}2n\mathsf W_{2n}(t)g^{2n} = \frac{1}{2g^{2}}\biggl[1-\frac{t}{g^{2}} + t\Bigl(\frac{1}{g^{2}}+\frac{a^{2}(t)}{2}-\frac{1}{t}\Bigr)\sqrt{1-a^{2}(t)g^{2}}\biggr]\, .\ee
These results show that $\mathsf W_{2n}(t)$ has a finite radius of convergence $t_{*}=1/12$, matching with the radius of convergence of the sphere generating function.\footnote{This fact and its generalizations are easily understood in the matrix model framework.} Similarly, one may check that $\tilde{\mathsf W}_{p}(g) = \sum_{n\geq 1}\mathsf W_{2n,p}g^{2n}$ has a radius of convergence $g_{*}=1/2$ and that $\sum_{n\geq 1}\mathsf W_{2n}(t)g^{2n}$ has a radius of convergence $g_{*}(t) = 1/a(t)$, etc. 

\paragraph{Continuum limit}

As in the case of the sphere, we can define a universal continuum limit. Let us work with the generating function at fixed boundary length $\mathsf W_{2n}(t)g^{2n}$. In terms of the Planck length $\ell_{0}$, the boundary length is
\be\label{bdlengthPlanck} \ell = 2n \ell_{0}\, .\ee
The factor $g^{2n}$ plays the r\^ole of a bare boundary cosmological constant,
\be\label{g2ncc} g^{2n} = e^{2n\ln g} = e^{-\la_{0}\ell}\, ,\ee
with $\la_{0} = (\ln\frac{1}{g})/\ell_{0}$. The continuum limit is
\be\label{climit3} n\rightarrow\infty\, ,\quad \ell_{0}\rightarrow 0\, ,\quad \ell = 2n\ell_{0}\ \text{fixed,}\ee
adjusting at the same time the bare bulk cosmological constant $t$ as in \eqref{climit2}, which is equivalent to setting
\be\label{diskLclf} \frac{t}{t_{*}} = 1-\frac{\ell^{2}\La}{64\pi}\frac{1}{n^{2}}\,\cdotp\ee
In this limit, $\mathsf W_{2n}(t)$ behaves as
\be\label{WenLiosc34} \mathsf W_{2n}(t) \underset{n\rightarrow\infty}{\sim} \frac{1}{\sqrt{\pi}}\frac{8^{n}}{n^{7/2}}e^{-\frac{\ell\sqrt{\La}}{8\sqrt{\pi}}}\biggl(1+\frac{\ell\sqrt{\La}}{8\sqrt{\pi}}\biggr)\, .\ee
The exponential growth $8^{n}$ can be cancelled, as usual, by adjusting the bare boundary cosmological constant $g$ to its critical value $g_{*}(t_{*}) = g_{*} = 1/a(t_{*}) = 1/(2\sqrt{2})$,
\be\label{climit4} g\rightarrow g_{*}^{-}\, ,\quad \la = \frac{1}{\ell_{0}}\Bigl(1-\frac{g}{g_{*}}\Bigr)\geq 0\quad \text{fixed.}\ee
We then get from \eqref{WenLiosc34} that $\mathsf W_{2n}(t)g^{2n}$ scales as
\be\label{WdiskLcont}\mathsf W(\ell,\La) \propto \ell^{-7/2}e^{-\la\ell-\frac{\sqrt{\La}}{8\sqrt{\pi}}\ell}\biggl(1 + \frac{\ell\sqrt{\La}}{8\sqrt{\pi}}\biggr)\, .\ee
We could absorb the $\La$-dependent contribution in the exponential by redefining $\la$, but we prefer to keep the formula in this way, with the parameter $\La$ directly coupling to the area, see below.

Note that, as we have already noticed in the case of the sphere topology, the theory makes sense only if $\La \geq 0$. For $\La<0$, even on a disk of fixed boundary length, the theory is unstable. As explained in Section \ref{stab1Sec} and further discussed in Section \ref{nonzeroRSec} and \cite{RDFollowup1}, this is similar to the expected behaviour of JT gravity in the positive curvature case, but in sharp contrast to the case of negative curvature, for which JT gravity is believed to exist for all positive and negative values of the cosmological constant. The behaviour in zero curvature is probably intermediate, the theory existing for all positive and at least for some range of negative values of the cosmological constant.

\paragraph{Bulk and boundary scalings}

An important and interesting feature of Liouville gravity is that the relative scaling of the bulk area and the boundary length is ``classical'' in the continuum limit, in the sense that the area scales as the square of the boundary length. For instance, working at fixed boundary length as above, the existence of the moments of the bulk area
\be\label{momareap} \langle A^{k}\rangle = \frac{\ell_{0}^{2}}{\mathsf W_{2n}(t)}\frac{\partial^{p}\mathsf W_{2n}(t)}{\partial\ln t^{p}} =\ell_{0}^{2k}\langle p^{k}\rangle \ee
in the continuum limit \eqref{climit3}, \eqref{diskLclf} implies that the expectation values $\langle p^{k}\rangle$ scale as $n^{2k}$ in this limit, for any integer $k\geq 1$. This can be checked explicitly from Eq.\ \eqref{W2ndL}. The expectation values $\langle A^{k}\rangle$ can actually be computed straightforwardly in the continuum from \eqref{WdiskLcont}, for instance
\be\label{areaexpLiou} \langle A\rangle = -\frac{16\pi}{\mathsf W}\frac{\partial\mathsf W}{\partial\La} = \frac{\ell^{2}}{8\bigl(1+\frac{\ell\sqrt{\La}}{8\sqrt{\pi}}\bigr)}\, \cdotp\ee
One could also take the point of view of the generating function for which both the area and the boundary length are fixed, Eq.\ \eqref{WpndiskL}. The above discussion implies that there exists a non-trivial continuum limit $\ell_{0}\rightarrow 0$ for which both the area and the boundary length are finite. In other words, we take $n\rightarrow\infty$, $p\rightarrow\infty$, with 
\be\label{arealengthc1} \frac{p}{(2n)^{2}} = \frac{A}{\ell^{2}}\quad\text{fixed.}\ee
Dealing with the usual exponential divergence by renormalizing the bare bulk and boundary cosmological constant using \eqref{climit2} and \eqref{climit4}, we get the continuum partition function at fixed $\ell$ and $A$,
\be\label{fixedAlLioucont} t^{p}g^{2n}\mathsf W_{2n,p}\underset{\substack{n,p\rightarrow\infty\\p=4n^{2}A/\ell^{2}}}{\propto}\mathsf W(A,\ell) = \frac{e^{-\frac{\La}{16\pi} A - \la \ell - \frac{\ell^{2}}{16 A}}}{\ell^{1/2}A^{5/2}}\, .\ee
By integration, 
\be\label{WlLWAlLiou} \mathsf W(\ell,\La) = \int_{0}^{\infty}\mathsf W(A,\ell)\,\d A\, ,\ee
we find again \eqref{WdiskLcont}.

In the continuum formulation of Liouville gravity, the bulk and boundary cosmological constant operators are represented by renormalized bulk $:\! e^{\gamma_{\text B}\varphi}\!:$ and boundary $:\! e^{\gamma_{\text b}\varphi}\!:$ exponentials of the Liouville field $\varphi$ for which the boundary constant $\gamma_{\text b}$ is equal to half the bulk constant $\gamma_{\text B}$, $\smash{\gamma_{\text b} = \frac{1}{2}\gamma_{\text B}}$.\footnote{This condition comes from the fact that the bulk operator $:\! e^{\gamma_{\text B}\varphi}\!:$ must be of conformal dimension 2 whereas the boundary operator $:\! e^{\gamma_{\text b}\varphi}\!:$ must be of conformal dimension 1.} This implies the classical scaling $A\propto \ell^{2}$.

As we have already emphasized in Section \ref{zeroImmSecbis} and in part of the discussion in Section \ref{earlySec}, the situation in JT gravity is completely different. The boundary cosmological operator is still represented by a renormalized boundary exponential $:\! e^{\gamma_{\text b}\varphi}\!:$ of the boundary Liouville field, with an associated renormalized quantum boundary length $\beta_{\text q} = \int :\! e^{\gamma_{\text b}\varphi}\!:\d\theta$, but the bulk operator is no longer renormalized, due to the constant curvature constraint. This yields a quantum-corrected scaling $A\propto \beta_{\text q}^{2\nu}$, for a non-trivial critical exponent $\nu = 1/d_{\text H}$ which is the inverse of the Hausdorff dimension of the boundary \cite{ferrari,ferraJTconfgauge}.

\paragraph{Generalization}

Based on the fact $\ell\sqrt{\La}$ is the only non-trivial dimensionless parameter in the model at fixed boundary length, one expects that the scaling limit of the generating function is of the general form
\be\label{WdiskLiougen} \mathsf W(\ell,\La) \propto \ell^{\alpha-3}e^{-\la\ell}f\biggl(\frac{\ell\sqrt{\La}}{8\sqrt{\pi}}\biggr)\, ,\ee
for an exponent $\alpha$ and a function $f$ which are a priori unknown. For pure gravity, Eq.\ \eqref{WdiskLcont} yields $\alpha = -1/2$ and $f(z)=e^{-z}(1+z)$. In spite of its very simple form, it is actually quite challenging to derive this result from a direct continuum approach, within Liouville theory \cite{Liouvillesuccess2,FZZT,gravityreviews,Liourevref}. Strong arguments \cite{FZZT} show that, in the more general case for which gravity is coupled to an arbitrary CFT of central charge $c<1$, $\alpha$ matches with the string susceptibility exponent $\gamma_{\text{str}}$ given by \eqref{gammaKPZ} and
\be\label{exactWLioudiskc} \mathsf W(\ell,\La) \propto \ell^{\gamma_{\text{str}}-3}e^{-\la\ell}\bigl(\ell\sqrt\La\bigr)^{1-\gamma_{\text{str}}} K_{1-\gamma_{\text{str}}}\Biggl[\frac{\ell\sqrt{\La/(16\pi)}}{\sqrt{\sin\frac{\pi}{1-\gamma_{\text{str}}}}}\Biggr]\, ,\ee
where $K$ is the modified Bessel function of the second kind.

For pure gravity, \eqref{exactWLioudiskc} is consistent with \eqref{WdiskLcont}, modulo a multiplicative constant in the definition of $\ell\sqrt{\La}$,
\be\label{idsqLsubtle} \bigl(\ell\sqrt{\La}\bigr)_{\text{sq.}} = \frac{2\sqrt{2}}{3^{1/4}}\bigl(\ell\sqrt{\La}\bigr)_{\text{cont.}}\, .\ee
The subscript ``sq.'' on the left-hand side refers to the quantity defined within the discretized approach using quadrangulations, Eq.\ \eqref{ArealPlanck} and \eqref{bdlengthPlanck}, whereas the subscript ``cont.'' on the right-hand side refers to the quantity defined directly within the continuum Liouville approach used in \cite{FZZT}. A finite multiplicative renormalization as in \eqref{idsqLsubtle} is harmless, but its origin might be confusing, so let's discuss it briefly. The problem comes from the fact that when one considers discretized paths, the length defined as in \eqref{bdlengthPlanck} will usually not match the length defined in the continuum, even when the discrete path has a well-defined length in the continuum limit. For instance, on a square lattice, the shortest path between two points of coordinates $(x_{i} = p_{i}\ell_{0}, y_{i}=q_{i}\ell_{0})$, $i=1,2$, has length $|x_{2}-x_{1}| + |y_{2}-y_{1}|$, whereas if we start directly in the continuum, the usual Euclidean distance would yield $\sqrt{(x_{2}-x_{1})^{2}+(y_{2}-y_{1})^{2}}$. Mathematically, the notion of distance obtained by taking the continuum limit of the square lattice and the usual Euclidean distance are said to be equivalent, in the sense that their ratio is bounded above and below by finite constants. In a random model, this subtlety has the effect of introducing finite multiplicative renormalization constants in the definition of the length, and this explains \eqref{idsqLsubtle}. The crucial point is that the length exists and is finite in the continuum limit.

Note that this problem does not arise for the definition of the area. Equation \eqref{idsqLsubtle} thus relates the length $\ell_{\text{sq.}}$ defined in terms of the continuum limit of quadrangulations in pure Liouville gravity to the length $\ell_{\text{cont.}}$ defined in the continuous version of the theory as
\be\label{idsqLsubtle2} \ell_{\text{sq.}} = \frac{2\sqrt{2}}{3^{1/4}}\ell_{\text{cont.}}\, .\ee
Bearing in mind the above discussion, it is reassuring to note that $2\sqrt{2}/3^{1/4}>1$.

\paragraph{Semi-classical limit}

The existence of a semi-classical regime in which the fluctuations of the geometry are tamed and the physics is dominated by a smooth classical geometry is a desirable feature of any interesting quantum gravity model. 

Pure Liouville quantum gravity does not have a semi-classical limit. On the disk, one may want to consider $\La\rightarrow +\infty$ at fixed $\ell$, but this limit is dominated by singular geometries for which the disk has a branched polymer structure. As for the limit $\La\rightarrow -\infty$, it does not make sense since the model is unstable for $\La <0$. Nevertheless, as noticed by Zalomodchikov in \cite{Zamolod82}, when Liouville gravity is coupled to conformal matter of central charge $c$, the limit $c\rightarrow -\infty$ does correspond to a semi-classical limit. Indeed, at large $|c|$, the Liouville action, that governs the gravitational dynamics, is multiplied by a constant proportional to $-c$. To obtain an interesting limit, one must also rescale the cosmological constant in such a way that the cosmological constant term in the action is proportional to $|c|$ as well. The semi-classical limit of Liouville gravity on the disk thus corresponds to 
\be\label{Lioucllimit} c\rightarrow - \infty\, ,\quad \La\rightarrow +\infty\, ,\quad \La/|c| = \frac{2}{3}\mu  \ \text{and}\  \ell\ \text{fixed.}\ee
The factor $2/3$ in the definition of $\mu$ is added because it allows to simplify some  formulas below. The semi-classical loopwise expansion of $\mathsf W(\ell,\La)$ takes the form
\be\label{WLioZamoexp} \ln \mathsf W(\ell,\La) =-\la\ell + (\gamma_{\text{str}}-3)\ln\frac{\ell}{2\pi} +\sum_{L\geq 0} |c|^{1-L}\mathsf f_{L}\biggl(\frac{\ell\sqrt{\mu}}{2\pi}\biggr)\, ,\ee
with, according to \eqref{gammaKPZ}, 
\be\label{gammaexpZ} \gamma_{\text{str}}-3 = -\frac{|c|}{6} - \frac{25}{6}+\frac{6}{|c|} + O\bigl(1/|c|^{2}\bigr)\, .\ee
The contributions $\mathsf f_{L}$ at any loop order can be straightforwardly obtained from the exact formula \eqref{exactWLioudiskc} and the asymptotic expansion of the Bessel function for large argument,
\be\label{KBesselexp} K_{\nu}(\nu z) \underset{\nu\rightarrow\infty}{=} \sqrt{\frac{\pi}{2\nu}}\frac{e^{-\nu h(z)}}{(1+z^{2})^{1/4}}\sum_{k\geq 0}(-1)^{k}U_{k}\Bigl(\frac{1}{\sqrt{1+z^{2}}}\Bigr)\nu^{-k}
\ee
where
\be\label{hdefBesselK} h(z) = \sqrt{1+z^{2}} + \ln\frac{z}{1+\sqrt{1+z^{2}}}\ee
and the $U_{k}$ are polynomials of degree $3k$ that can be computed recursively, for example
\be\label{UkBesselpol} U_{0}(x) = 1\, ,\quad U_{1}(x) = \frac{1}{24}\bigl(3x-5x^{3}\bigr)\, ,\quad \text{etc.}\ee
Normalizing $\mathsf W$ in such a way that
\be\label{WLioudisknorm} \mathsf W(\ell=2\pi,\mu=0,c) = e^{-\la\ell}\, ,\ee
we find, up to two loops,
\begin{align}\label{WLioZloop0} & \mathsf f_{0}(z) = \frac{1}{6}\ln\frac{1+\sqrt{1+z^{2}}}{2}-\frac{1}{6}\bigl(\sqrt{1+z^{2}}-1\bigr)\, ,\\
\label{WLioZloop1} & \mathsf f_{1}(z) = \frac{13}{6}\ln\frac{1+\sqrt{1+z^{2}}}{2}-\frac{13}{12}\bigl(\sqrt{1+z^{2}} -1\bigr)-\frac{1}{4}\ln\bigl(1+z^{2}\bigr)
\, ,\\\label{WLioZloop2} 
\begin{split} &
\mathsf  f_{2}(z) = -6 \ln\frac{1+\sqrt{1+z^{2}}}{2} + \frac{313 z^{4}+759 z^{2}+506}{48(1+z^{2})^{3/2}}-\frac{253}{24}\\ &\hskip 6cm +\frac{13}{4}\frac{z^{2}}{1+z^{2}}-\frac{\pi^{2}}{2}\bigl(\sqrt{1+z^{2}} -1\bigr)\, .
\end{split}
\end{align}

One can check explicitly that the leading terms in \eqref{WLioZamoexp}, corresponding to the contribution $-|c|/6$ in \eqref{gammaexpZ} and to $\mathsf f_{0}$ in \eqref{WLioZloop0}, precisely match with the expected result obtained from a classical disk geometry, see App.\ \ref{SemiclassApp}. In principle, the higher loop contributions can be obtained from the path integral following the same logic as in \cite{Zamolod82}. This is checked explicitly in \cite{Loopcalc}, where the analogue of the limit \eqref{Lioucllimit} is also studied for JT gravity.

\subsection{\label{randompathrev} Review of random paths models}

Our goal in this subsection is to summarize the standard framework in which models of random polygons, which yield theories of random closed paths in the continuum limit, are usually analysed. We thus consider a general model of random polygons on the square lattice. As we shall see, there are similarities, but also crucial differences, with the Liouville models of random metrics reviewed in the previous subsection. Understanding this physics in detail is important in order to elucidate which aspects, some more metric-like, others more path-like, are likely to be relevant for the new model of self-overlapping polygons (SOP) we are primarily interested in, see Section \ref{SOPJTSec}. Indeed, as already emphasized many times, the most characteristic property of this model is to have dual interpretations, either in terms of random polygons or in terms of random metrics.

Our discussion is illustrated on the examples of arbitrary random polygons (RP), self-avoiding polygons (SAP) and convex polygons (CP), which may all be seen as toy approximate models for flat JT gravity. The RP model is the flat version of the model proposed by Kitaev and Suh in negative curvature \cite{KitaevSuh}. The SAP model is the one studied in Stanford and Yang in \cite{StanfordSAP}. The CP model shares some properties with the one-loop exact reparameterization ansatz used in the near-hyperbolic limit of negative curvature JT gravity.

\subsubsection{\label{GenRPmodSec}Generalities}

Let us first note that, for an arbitrary polygon, there is no unique canonical notion of area, see e.g.\ Eq.\ \eqref{notionsofarea}. For this reason, the case of the RP model is a bit special, see Section \ref{RPexSec}. For the time being, we limit ourselves to models like CP, SAP or SOP for which the notion of area is unambiguous.

We note $W^{\text P}_{2n,p}$ the number of polygons of length $2n$ and area $p$ in the polygon model under consideration. The numbers $W^{\text P}_{2n,p}$ may include a multiplicity factor, as in the case of SOP, see Section \ref{Genfun2Sec}. We introduce, as usual, the generating functions at fixed length and area
\be\label{WPndef} W_{2n}^{\text P}(t) =\sum_{p\geq 1}W_{2n,p}^{\text P}t^{p}\, ,\quad\tilde W_{p}^{\text P}(g) = \sum_{n\geq 1}W_{2n,p}^{\text P}g^{2n}\ee
and the ``grand canonical'' generating function
\be\label{WgenPdef} W^{\text P}(t,g) = \sum_{n\geq 1}W_{2n}^{\text P}(t)g^{2n} = \sum_{p\geq 1}\tilde W_{p}^{\text P}(g)t^{p}\, .\ee
The parameter $t$, or more precisely its logarithm, plays the r\^ole of a ``bare bulk cosmological constant.'' Note that in the context of random polygon models, one may rather talk about ``pressure,'' with
\be\label{pressureccrel} P = -\frac{\La}{16\pi}\,\cdotp\ee
This terminology is more natural when one uses the polygons to model pressurized vesicles, cell membranes or polymer chains.

To develop our analysis, we assume that the following properties are satisfied.
\begin{hypothesis}\label{hyp1} There exists strictly positive constants $a$ and $b$ such that $W^{\text P}_{2n,p}=0$ unless $2b n \leq p\leq a (2n)^{2}$. 
\end{hypothesis}
In our models, this follows from the isoperimetric inequality, which ensures that for a given $n$ there is a maximum possible value $p=p_{\text{max}}(n) = [\frac{1}{4}n^{2}]$, and from Eq.\ \eqref{nprel}, which ensures that for a given $p$ there is a maximum value $n=n_{\text{max}}(p) = p+1$.

A simple important consequence is that $W^{\text P}_{2n}(t)$ is a polynomial in $t$. In particular, the number
\be\label{wnP} w_{2n}^{\text P} = W^{\text P}_{2n}(1) = \sum_{p\geq 1} W_{2n,p}^{\text P}\ee
of polygons of a given length is finite. \emph{These properties, as well as other crucial consequences described below, are not valid in Liouville gravity.}

An important physical feature is the typical ``size'' of the polygons. This can be measured by studying the area distribution, whose expected moments are given in terms of the derivatives of the generating function as
\be\label{areamomentdef} \langle p^{r}\rangle_{2n}(t) = \frac{1}{W_{2n}^{\text{P}}}\frac{\d^{r}W_{2n}^{\text{P}}}{\d(\ln t)^{r}}\,\cdotp\ee
Let
\be\label{w2nr} w_{2n;r}^{\text P} = \sum_{p\geq 1}p^{r} W^{\text P}_{2n,p}\, ,\ee
generalising \eqref{wnP}. We assume that, at large $n$,
\be\label{w2nras} w_{2n;r}^{\text P}\underset{n\rightarrow\infty}{\sim} a_{r} g_{*}^{-2n}(2n)^{2\nu (r-\vartheta)-1}\, ,\ee
for all $r\geq 0$, for certain strictly positive constants $a_{r}$ and $g_{*}$, and critical exponents $\nu$ and $\vartheta$. In particular,
\be\label{areaasymr} \langle p^{r}\rangle_{2n}(t=1) = \frac{w_{2n;r}^{\text P}}{w_{2n}^{\text P}}  \underset{n\rightarrow\infty}{\sim} \frac{a_{r}}{a_{0}}(2n)^{2\nu r}\, .\ee
The asymptotic scaling \eqref{w2nras} is found in many models, by analytical methods or numerically. We conjecture that this large $n$ scaling will be valid for SOP as well, modulo possible logarithmic corrections.

Observables other than the area can also be used to measure the size of polygons. A widely used quantity is the radius of gyration $R_{\gamma}$. For an arbitrary polygon $\gamma$ made of $\mathscr N$ points located at $w_{i}\in\mathbb C$, noting $w_{\gamma}= \frac{1}{\mathscr N}\sum_{i}w_{i}$ the ``center of mass,'' it is defined by
\be\label{gyrraddef} R_{\gamma}^{2} = \frac{1}{\mathscr N}\sum_{i=1}^{\mathscr N}|w_{i}-w_{\gamma}|^{2} = \frac{1}{\mathscr N^{2}}\sum_{i<j}|w_{i}-w_{j}|^{2}\, ,\ee
where the second equality can be checked by an elementary calculation. One could also consider the diameter $D_{\gamma}=\max_{i,j}|w_{i}-w_{j}|$ and several other quantities of the same type. It is natural to expect that the moments of these quantities will have scalings similar to \eqref{areaasymr}, e.g.\
\be\label{gyrscaling} \langle R^{2r}\rangle\underset{n\rightarrow\infty}{\propto} (2n)^{2\nu' r}\, .\ee
If the notion of size provided by the area and the radius of gyration are the same, then $\nu'=\nu$, which, as far as we know, happens in all known models for which both exponents have been studied.

From the above, one can easily derive a few important general properties that the generating functions must satisfy. 

i) The constant $g_{*}$ appearing in \eqref{w2nras} is the radius of convergence of the series 
$\sum_{n\geq 1}w_{2n}^{\text P}g^{2n}$. It is such that $0<g_{*}\leq 1$. The quantity $1/g_{*}$ is usually called the \emph{connective constant} of the model. Its precise value is not universal (it depends on the lattice) and is thus not of great interest in physics, but computing it is a challenging and interesting problem in mathematics, see e.g.\ \cite{duminilconnective} for the case of SAP.\footnote{For SAP on a square lattice, $1/g_{*}\simeq 2.638$.}

ii) For $t\leq 1$, the series $\sum_{n\geq 1}W^{\text P}_{2n}(t)g^{2n}$ has a strictly positive radius of convergence $g_{*}(t)$, which is a decreasing function of $t$, with $g_{*}(t=1) = g_{*}$. One shows that $g_{*}(t)\geq g_{*}t^{-b}$, where $b$ is the constant appearing in Hypothesis \ref{hyp1}. In our models, with $n_{\text{max}}(p) = p+1$, we actually have $g_{*}(t)\geq g_{*}/\sqrt{t}$. In particular, $\lim_{t\rightarrow 0}g_{*}(t) = +\infty$. When $t>1$, the radius of convergence is $g_{*}(t)=0$.

iii) The series $\sum_{p\geq 1}\tilde W_{p}^{\text P}(g_{*})t^{p}$, as well as the derivatives of this series with respect to $t$, to any order, has radius of convergence $t_{*}=1$ and, moreover, $\lim_{t\rightarrow 1^{-}}\partial_{t}^{r}W^{\text P}(t,g_{*}) = \partial_{t}^{r}W^{\text P}(1,g_{*})$ for any $r\geq 0$.

iv) The critical exponent $\nu$ is the inverse of the Hausdorff fractal dimension of the curves, $d_{\text H}=1/\nu$, and satisfies
\be\label{nuineq} \frac{1}{2}\leq\nu\leq 1\, .\ee
For instance, the values $\nu=1/2$, $\nu=3/4$ and $\nu=1$ corresponds to purely random loops, self-avoiding loops and smooth loops (realized for convex polygons), respectively. More generally, $\text{SLE}_{\kappa}$ curves have $\nu = \max \bigl((1+\kappa/8)^{-1},1/2\bigr)$. For JT quantum gravity coupled to a matter CFT of central charge $c\leq 0$, the continuum formulation developed in \cite{ferrari,ferraJTconfgauge} implies that
\be\label{critexponentsJT} \nu = \frac{1}{2}\Biggl[1+\sqrt{\frac{c}{c-24}}\Biggr]\, .\ee
This formula predicts that the exponent is a strictly decreasing function of $c$ that goes to one when $c\rightarrow -\infty$, consistently with the idea that this limit is semiclassical \cite{Zamolod82,Loopcalc}, and is equal to $1/2$ for the pure gravity theory, $c=0$, on which we focus in the present paper. So the SOP model has the same critical exponent $\nu$, or equivalently the same Hausdorff dimension, as the random walks.

\subsubsection{\label{phasessubSec}Phases}

\begin{figure}
\centerline{\includegraphics[width=6in]{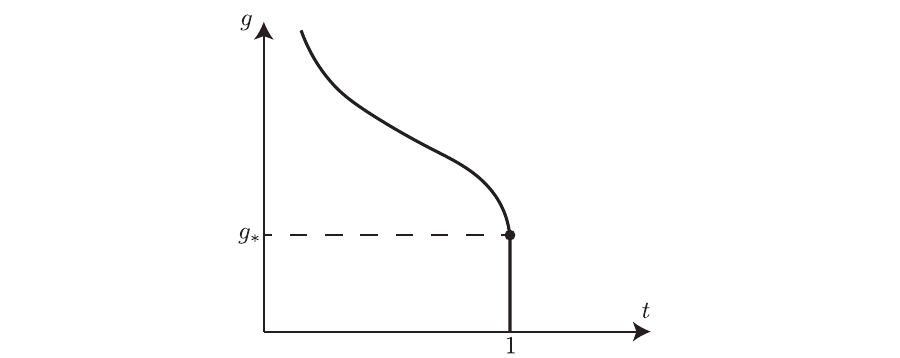}}
\caption{Phase diagram of a random closed paths model in the grand canonical ensemble (from \cite{SAPtextbook}). The representative curve of the function $g_{*}(t)$ is plotted. Above and on the right of the curve, the grand canonical generating function $W^{\text P}(t,g)$ does not exist. The thick dot represents the critical point $(t_{*}=1,g_{*})$ which plays a crucial r\^ole in defining the continuum limit of the model.\label{phasepathFig}}
\end{figure}

The phase diagram in the $(t,g)$-plane is depicted in Fig.\ \ref{phasepathFig}. If we work at fixed length $n$, we can distinguish three regimes.

i) The \emph{deflated regime} occurs when $t<1$ (positive cosmological constant, negative pressure) and corresponds to a large $n$ linear growth of the free energy $\ln W_{2n}^{\text P}(t)\sim -2n\ln g_{*}(t)$. In particular, the average area
\be\label{deflatedarea} \langle p\rangle_{2n}(t) = \frac{\d\ln W^{\text P}_{2n}}{\d\ln t} \underset{n\rightarrow\infty}{\sim} -2n\frac{tg_{*}'(t)}{g_{*}(t)}\ee
is proportional to $n$. This indicates that the typical polygon is tree-like or, equivalently, that the deflated regime is a branched polymer phase. 

ii) The \emph{critical point} occurs at $t=1$ (zero cosmological constant or pressure). It is characterized by the scaling laws \eqref{w2nras}, \eqref{areaasymr}. As we shall see below, the critical point is instrumental in defining the continuum limit and the continuum limit is entirely determined by the properties at the critical point.

iii) The \emph{inflated regime} occurs when $t>1$ (negative cosmological constant, positive pressure). Using Hypothesis \ref{hyp1}, one finds that, if $t>1$,
\be\label{inflatedineq} t^{a(2n)^{2}}\leq W_{2n}^{\text P}(t) \leq w^{\text P}_{2n}t^{a(2n)^{2}}\, .\ee
Using \eqref{w2nras} for $r=0$ (or the weaker condition $\ln w_{2n}^{\text P} = o(n^{2})$) we get
\be\label{inflfreeen} \lim_{n\rightarrow\infty}\frac{1}{(2n)^{2}}\ln W^{\text P}_{2n}(t) = a\ln t\, .\ee
In particular, recalling that $a=1/16$ for our models,
\be\label{areainflated}\langle p\rangle_{2n}(t) \underset{n\rightarrow\infty}{\sim} 4a n^{2} = \frac{n^{2}}{4}\ee
matches with the maximum area allowed by the isoperimetric inequality: the polygons are inflated ``balloons'' in this phase.

The very simple result \eqref{inflfreeen} is valid when the limit $n\rightarrow\infty$ is taken at fixed $t>1$. However, if $t-1$ is scaled appropriately with $n$ at large (but finite) $n$, a more interesting behaviour is generally expected, based, for example, on numerical simulations, see e.g.\ \cite{Leibler} for a study of the case of SAP. It is useful to distinguish two ranges of parameters in the inflated phase:

a) The region $0<n^{2\nu}(t-1)\ll n^{2\nu -1}$ corresponds to an inflated ``crumpled'' phase. If $\nu>1/2$, one can distinguish a weakly inflated regime, corresponding to $n^{2\nu}(t-1)$ of order one, and a strongly inflated regime, for $1\ll n^{2\nu}(t-1)\ll n^{2\nu -1}$.

b) For $\nu>1/2$ and $n^{2\nu}(t-1)\geq n^{2\nu -1}$ we have a fully inflated phase, characterized by $\langle p\rangle \propto n^{2}$.

The existence of the inflated crumpled phase plays an important r\^ole, see the next subsection on the continuum limit. The fact that the case $\nu=1/2$ is special also has important qualitative consequences, see \ref{RPexSec} for an example and the conjecture is Section \ref{areaSOPcomnjSec}.

\noindent\emph{Remark:} the transition at $t=t_{*}=1$ is reminiscent of a deconfining or of a Hagedorn phase transition, in the sense that the large $n$ scaling of the free energy changes, here from $\ln W_{2n}^{\text P}(t)\propto n$ for $t<1$ to $\ln W_{2n}^{\text P}(t)\propto n^{2}$ for $t>1$.

\subsubsection{\label{pathcontSec}Continuum limit}

A very important property of the random paths models is that their continuum limit can be defined starting from the critical point at $t=t_{*}=1$, i.e.\ for zero cosmological constant, and then possibly deform away from it, by turning on the renormalized cosmological constant.

\paragraph{Critical theory}

To define the continuum limit, we introduce as usual a cut-off $\ell_{0}$. The dimensionful area is then given by $A=p\ell_{0}^{2}$ and Eq.\ \eqref{areaasymr} yields
\be\label{areamomc1} \langle A^{r}\rangle = \frac{a_{r}}{a_{0}} \bigl(2n\ell_{0}^{1/\nu}\bigr)^{2\nu 
r}\, .\ee
The continuum limit is then defined by taking
\be\label{contpathdef} n\rightarrow\infty\, ,\quad \ell_{0}\rightarrow 0\, , \quad 2n \ell_{0}^{1/\nu} = \beta_{\text q}\quad \text{fixed.}\ee
The parameter $\beta_{\text q}$ is the ``quantum length'' of the polygons in the continuum limit, to which we have already referred several times in previous sections. Note that it is proportional to the microscopic length $n$, and yet its engineering length dimension is $1/\nu$. The scaling \eqref{contpathdef} yields well-defined moments for the area,
\be\label{areamoments} \langle A^{r}\rangle =  \frac{a_{r}}{a_{0}}\beta_{\text q}^{2\nu r}\, .\ee
Under mild assumptions, see the remarks at the end of this subsection, these moments determine uniquely the area distribution law at zero cosmological constant, noted $\rho^{\text P}_{\beta_{\text q}}$, which is such that
\be\label{Akrhoform} \langle A^{r}\rangle = \int_{0}^{\infty}A^{r}\rho^{\text P}_{\beta_{\text q}}(A)\d A\, .\ee

The continuum generating function at zero cosmological constant, $W^{\text P}(\beta_{\text q},\La=0)$, is determined by the asymptotics of $w_{2n}^{\text P}$, given by \eqref{w2nras} for $r=0$, together with the scaling \eqref{contpathdef}. Considering, as usual (compare with Eq.\ \eqref{climit4}), $g^{2n}w_{2n}^{\text P}$, for a bare parameter $g$ adjusted in such a way that
\be\label{clpath22} g\rightarrow g_{*}\, ,\quad \la = \frac{1}{\ell_{0}^{1/\nu}}\Bigl(1-\frac{g}{g_{*}}\Bigr)\ \text{fixed,}\ee
and discarding an overall factor $\ell_{0}^{2\vartheta+1/\nu}$ of the cut-off, we get
\be\label{WPgenzeroL} W^{\text P}(\beta_{\text q},\La=0) = a_{0}e^{-\beta_{\text q}\la}\beta_{\text q}^{-1-2\nu\vartheta}\, .\ee
The contribution $e^{-\beta_{\text q}\la}$ is a ``boundary cosmological counterterm'' that may be adjusted as will. It is important to emphasize here that the logarithm of this term is proportional to the \emph{quantum} length $\beta_{\text q}$. In particular, the boundary cosmological constant $\la$ has dimension of length to the power $-1/\nu$. This observation may be unremarkable from the point of view of a standard random path model, but it is non-trivial for the SOP model, which has a gravitational interpretation: it shows that the naive gravitational boundary length counterterm, which is ill-defined in the quantum theory because the boundary is fractal, should be replaced by a quantum boundary length counterterm.

\paragraph{Turning on $\La$} The generating function for non-zero cosmological constant is obtained by using the area distribution as
\be\label{Wconpathg} W^{\text P}(\beta_{\text q},\La) = W^{\text P}(\beta_{\text q},0)\int_{0}^{\infty}e^{-\frac{\La}{16\pi}A}\rho_{\beta_{\text q}}(A)\d A\, .\ee
This can be put in a form similar to \eqref{WdiskLiougen},
\be\label{Wconpathg2} W^{\text P}(\beta_{\text q},\La) = a_{0}e^{-\beta_{\text q}\la}\beta_{\text q}^{-1-2\nu\vartheta} f_{\text P}\Bigl(\frac{\La\beta_{\text q}^{2\nu}}{16\pi}\Bigr)\, ,\ee
for 
\be\label{fPath} f_{\text P}(z) = \frac{1}{a_{0}}\sum_{r\geq 0}\frac{(-1)^{r}a_{r}}{r!} z^{r}\, .\ee
The continuum generating function $W^{\text P}(\beta_{\text q},\La)$ may also be obtained directly from $g^{2n}W^{\text P}_{2n}(t)$ defined in \eqref{WPndef}, in a way similar to what is done in Liouville theory, by using the scalings \eqref{contpathdef}, \eqref{clpath22} and
\be\label{clpath23} t\rightarrow t_{*}=1\, ,\quad \La = \frac{16\pi}{\ell_{0}^{2}}\bigl(1-t\bigr)\quad \text{fixed,}\ee
which is exactly as in \eqref{diskLclf}).

The representation \eqref{Wconpathg} shows that the models are always well-defined for positive cosmological constant, $\La\geq 0$. On the other hand, the existence for negative cosmological constant is not a priori ensured. It depends on the asymptotic behaviour of the area probability distribution $\rho_{\beta_{\text q}}(A)$ for large $A$.

The discussion at the end of subsection \ref{phasessubSec}, which applies to the regime $\La<0$ for models defined on a flat lattice, shows that when $1/2<\nu<1$ the continuum limit selects the crumpled phase. It is then natural to conjecture that such models will exist for all values of $\La$, positive or negative. When $\nu=1/2$, the situation is less clear and one may expect an upper bound on the possible value of $-\La$; see, for example, the example in \ref{RPexSec} below. 

In the case of JT gravity, the asymptotic behaviour of the density $\rho_{\beta_{\text q}}$ and thus the existence of the models for negative $\La$ is expected to strongly depend on whether the curvature of space-time is positive, zero or negative, see Sections \ref{areaSOPcomnjSec}, \ref{nonzeroRSec} and \cite{RDFollowup1}.

\paragraph{Grand canonical ensemble and scaling form} 

It is also useful to discuss the grand canonical continuum generating function $W^{\text P}(\la,\La)$, which is the continuum limit of the generating function $W^{\text P}(t,g)$ defined in \eqref{WgenPdef}. One has\footnote{Note that the factor $e^{-\beta_{\text q}\la}$ is conventionally included in $W^{\text P}(\beta,\La)$, see Eq.\ \eqref{Wconpathg2}. Note also that the sum over $n$ should be replaced by an integral over $\beta$ with the measure $\smash{\d\beta/(2\ell_{0}^{1/\nu})}$, according to the definition \eqref{contpathdef}. We also discard any overall power of $\ell_{0}$, as usual.}
\be\label{WPathgrandcano} W^{\text P}(\la,\La) = \frac{1}{2}\int_{2\ell_{0}^{1/\nu}}^{\infty} W^{\text P}(\beta_{\text q},\La)\,\d\beta_{\text q}\, .\ee
We have reintroduced the cut-off $\ell_{0}$ in \eqref{WPathgrandcano}\footnote{The way we do it is natural because, from \eqref{contpathdef},  $2\ell_{0}^{1/\nu}$ is the minimal value of $\beta$ in the microscopic theory.} in order to regulate the divergence of the integral at $\beta_{\text q} = 0$ which, from \eqref{Wconpathg2}, occurs when $\nu\vartheta\geq 0$. From \eqref{Wconpathg2} and \eqref{fPath}, we get
\be\label{WPGCcont} W^{\text P}(\la,\La) = \frac{1}{2}\Bigl(\frac{\La}{16\pi}\Bigr)^{\vartheta}\mathcal F\biggl(\Bigl(\frac{16\pi}{\La}\Bigr)^{\frac{1}{2\nu}}\la\Bigr)\biggr) + \text{diverging piece,}\ee
where
\be\label{Fcaldef} \mathcal F (z) = \sum_{r\geq 0}\mathcal F_{r}z^{-2\nu(r-\vartheta)}\, ,\quad \mathcal F_{r} = \frac{(-1)^{r}}{r!}\Gamma\bigl(2\nu (r-\vartheta)\bigr)a_{r}\, .\ee
The diverging piece can be straightforwardly evaluated; if $\alpha$ is not a negative integer, it is a polynomial in $\la$ of degree $[2\nu\vartheta]$. Using \eqref{clpath22} and \eqref{clpath23}, one sees that Eq.\ \eqref{WPGCcont} follows from the scaling form,
\be\label{scalingdeflate} W^{\text P}(t,g) = W^{\text P}_{\text{reg.}}(t,g) + \frac{1}{2}(1-t)^{\vartheta}\mathcal F\Bigl(\frac{1-g/g_{*}}{(1-t)^{\frac{1}{2\nu}}}\Bigr)\ee
which is valid near the critical point $(t,g) = (1,g_{*})$, when $t\rightarrow 1^{-}$ and $g\rightarrow g_{*}^{-}$. The regular part is a non-universal piece that plays no role in the continuum limit. Using the expansion \eqref{Fcaldef} in \eqref{scalingdeflate}, we also get the critical behaviour
\be\label{partialrtWPath} \frac{1}{r!}\frac{\partial^{r}W^{\text P}}{\partial t^{r}}\bigl(t=1,g\bigr)\underset{g\rightarrow g_{*}^{-}}{\sim}\frac{\frac{1}{2}(-1)^{r}\mathcal F_{r}}{(1-g/g_{*})^{2\nu (r-\vartheta)}}\,\cdotp\ee
If $\vartheta$ is a positive integer and $r=\vartheta$, the above power-law behaviour is replaced by a logarithmic singularity.

\paragraph{Remarks}

i) Necessary and sufficient conditions for the moments of a random variable to fix uniquely the associated probability distribution are well-known, see e.g.\ \cite{momentref} and references therein. A simple sufficient condition is the existence of the theory at negative $\La$.

ii) In the case of the Liouville theory, it is straightforward to check from \eqref{areaexpLiou} that the moments $\langle A^{r}\rangle$ all diverge at the critical point $\La=0$, except for $r=1$. In particular, there is no convergent expansion \eqref{Wconpathg2} in Liouville gravity. Physically, this is related to the fact that Liouville gravity exists only for positive cosmological constant, as the exact results \eqref{areaexpLiou} or \eqref{exactWLioudiskc} clearly show. A similar qualitative behaviour is expected for the positive curvature JT gravity theory, see Section \ref{nonzeroRSec} and \cite{RDFollowup1}.

\subsubsection{\label{RPexSec} Example 1: purely random polygons} 

As a first example, let us discuss the case of the RP model: general, unconstrained, random polygons, with the uniform measure. This is the flat space version of the model studied by Kitaev and Suh in negative curvature \cite{KitaevSuh}.

For general polygons there is no canonical notion of area. We use L\'evy's algebraic area defined in Eq.\ \eqref{notionsofarea} to define the model. This notion matches the ordinary notion of area for CP, SAP and SOP. For arbitrary polygons, the algebraic area may be positive or negative and there is a symmetry $A\mapsto -A$.

The model was solved as soon as 1950 by Paul L\'evy himself \cite{Levy} and the results were rediscovered later by physicists using the link with the quantum mechanics of a particle in a uniform imaginary magnetic field \cite{Brownian1}.

On the square lattice, the numbers \eqref{wnP} are given by
\be\label{wnRPform} w_{2n}^{\text{RP}} = \frac{1}{2n}\frac{(2n)!^{2}}{n!^{4}}\underset{n\rightarrow\infty}{\sim}\frac{4^{2n}}{2\pi n^{2}}\,\cdotp\ee
The prefactor $1/(2n)$ comes from the fact that the polygons are unrooted, that is to say, there is no marked (privileged) point along the polygon. Comparing with \eqref{w2nras} and using the well-known fact that $\nu=1/2$ for random walks, we get $\vartheta=1$ and $a_{0}=2/\pi$. In the continuum, with $\beta_{\text q}$ given as in \eqref{contpathdef}, the model is described by the particle action
\be\label{SpartRP} S = \frac{1}{2}\int_{0}^{\beta_{\text q}}\dot{\mathbf r}^{2}\,\d\tau + \frac{\La}{32\pi}\oint\bigl( x \d y - y\d x\bigr)\, ,\ee 
where $\mathbf r = (x,y)$ is the particle position vector. The area distribution is found to be
\be\label{rhobetaRP}\rho_{\beta_{\text q}}^{\text{RP}}(A) = \frac{\pi}{2\beta_{\text q}}\frac{1}{\cosh^{2}\frac{\pi A}{\beta_{\text q}}}\ee
and the generating function at fixed $\beta_{\text q}$ is
\be\label{WRPf1} W^{\text{RP}}(\beta_{\text q},\La) =W^{\text{RP}}(\beta_{\text q},0)\int_{-\infty}^{+\infty}e^{-\frac{\La}{16\pi}A}\rho^{\text{RP}}_{\beta_{\text q}}(A)\d A = \frac{2e^{-\la\beta_{\text q}}}{\pi\beta_{\text q}^{2}}\frac{\frac{\La\beta_{\text q}}{32\pi}}{\sin\frac{\La\beta_{\text q}}{32\pi}}\,\cdotp\ee

It is important to realize that the peculiar symmetry $A\mapsto -A$, or equivalently $\La\mapsto - \La$, of the model, implies that we are always in the ``inflated crumpled phase,'' described at the end of Section \ref{phasessubSec}, for the critical exponent $\nu=1/2$. In particular, this model does not allow to probe the $\La>0$ regime of the models for which the area is positive. However, it is plausible that it can capture, at least at the qualitative level, some physical features of flat JT gravity in the negative cosmological constant regime. 

With this in mind, it is very interesting to highlight the following property, which illustrates nicely, on a toy model, the discussion we made in Section \ref{JTptSec}. 

Naively, the limit of large dimensionless parameter $|\La|\beta_{\text q}$ is a semiclassical limit, akin to the $\La\rightarrow - \infty$ limit at fixed $\ell$ discussed in JT gravity in Section \ref{JTptSec}. In this limit, the area term dominates in the action \eqref{SpartRP} and one may expect, in the same way as this might have been naively expected in JT, to have small fluctuations around a circle configuration, which maximizes the area.\footnote{On the square lattice, it is the square configuration that maximizes the area. In the continuum formulation based on the action \eqref{SpartRP}, it is the circle.} But this is \emph{not} what happens. The exact formula \eqref{WRPf1} for the generating function shows that, actually, there is a maximal value
\be\label{LamaxRP} \bigl(|\La|\beta_{\text q}\bigr)_{\text{max}} = 32\pi^{2}\, .\ee
There is thus no large $|\La|\beta_{\text q}$ limit in this model! Instead, the loop ``explodes'' at the critical value \eqref{LamaxRP}, for which the average area
\be\label{meanareaRP} \langle A\rangle = -16\pi\frac{\partial\ln W^{\text{RP}}}{\partial\La} =\frac{\beta_{\text q}}{2}\Biggl[\frac{1}{\tan\frac{\La\beta_{\text q}}{32\pi}} - \frac{32\pi}{\La\beta_{\text q}}\Biggr]\ee
diverges.

The divergence of the expected area in the model at fixed quantum boundary length $\beta_{\text q}$ is a nice illustration of the fact that the isoperimetric inequality gives no constraint in the continuum limit, because the closed polygon becomes a fractal in this limit, of infinite geometric length. Let us also mention that if, instead of taking the strict continuum limit, we simply consider $n$ large but finite, the sharp bound \eqref{LamaxRP} becomes a crossover scale between the crumpled regime for which $\langle A\rangle\propto n^{2\nu}\ell_{0}^{2}=n\ell_{0}^{2}$, which characterizes the continuum theory, to a fully inflated ``smooth'' regime with $\langle A\rangle\propto n^{2}\ell_{0}^{2}$. As we shall argue in Section \ref{nonzeroRSec}, this crossover is similar to the transition from the full-fledged microscopic description to the effective Schwarzian description at large negative $\La$ in the negative curvature JT theory; see also \cite{RDFollowup1}.

As a last comment, let us mention that it is also possible to compute the average area $A_{\nu}$ of the sectors that are enclosed by the loop and that have a winding number $\nu$, as defined in Eq.\ \eqref{windingP}, see Comtet et al.\ in \cite{Brownian1} and \cite{trujillo}. The result is
\be\label{AnRPf}\langle A_{0}\rangle = \frac{\pi\beta}{30}\, \cvp\quad\langle A_{\nu}\rangle = \frac{\beta}{2\pi \nu^{2}}\quad \text{if $\nu\not = 0$.}\ee
In particular, the average arithmetic area defined in Eq.\ \eqref{notionsofarea} is
\be\label{arithareaav}\langle A_{\text{arith.}}\rangle = \sum_{\nu=-\infty}^{+\infty}\langle A_{\nu}\rangle = \frac{\pi\beta}{5}\ee
whereas the average winding area
\be\label{windareaav}\langle A_{\text{wind.}}\rangle = \sum_{\nu=-\infty}^{+\infty}|\nu|\langle A_{\nu}\rangle = +\infty\ee
diverges. This is an interesting piece of information, to put into perspective with the model of self-overlapping polygons, for which the winding area coincides with the actual area of the associated distorted disk.

\subsubsection{\label{SAPclimitSec} Example 2: random self-avoiding polygons} 

The model of self-avoiding polygons (polygons that do not self-intersect), with the uniform measure, was proposed by Stanford and Yang in \cite{StanfordSAP} as a possible microscopic definition of JT. We have explained in previous Sections that such a model takes into account only a measure zero subset of the allowed configurations in JT and is thus expected be a rather mediocre minisuperspace approximation to the full theory. Yet, as for the simple RP model, it is illuminating to describe its main features.

Note that relatively little is actually proven rigorously for this model, in spite of decades of intense study. In particular, the numbers $W_{2n,p}^{\text{SAP}}$ counting SAPs of length $2n$ bounding a distorted disk of area $p$ are not known for arbitrary values of $n$ and $p$. Explicit formulas that would allow to determine these numbers recursively are not known either. But general properties can be derived and well-motivated conjectures exist about many important quantities \cite{SAPcounting,randpathlit,SAPtextbook}. In particular, the critical exponents $\nu$ and $\theta$ are
\be\label{nuthetaSAP} \nu = \frac{3}{4}\, \cvp\quad\vartheta = 1\quad\text{for SAP.}\ee
The (non universal) connective constant that governs the exponential growth in \eqref{w2nras} is $1/g_{*}\simeq 2.638$ on the square lattice.\footnote{See \cite{duminilconnective} for an exact calculation on the hexagonal lattice.} The parameter $a_{0}$ that enters in the asymptotic law \eqref{w2nras} for $r=0$ is not know exactly but can be determined numerically,
\be\label{azeroSAP} a_{0} \simeq 0.562301\, .\ee
On the square lattice, the numbers \eqref{wnP} are thus given by
\be\label{wnSAPform} w_{2n}^{\text{SAP}}\underset{n\rightarrow +\infty}{\sim}a_{0}g_{*}^{-2n}(2n)^{-5/2}\simeq 0.562\times\frac{2.638^{2n}}{(2n)^{5/2}}\,\cdotp\ee

The most impressive result available for the model is a conjecture for the exact solution in the continuum limit \cite{conjscaSAP}. This conjecture is supported by analogies with the simpler, exactly solvable, model of staircase polygons, and has been confirmed to a very high degree of precision by numerical analysis. To present the result, we introduce the logarithmic derivative of the Airy function
\be\label{Airylogder} \varphi(x) = -\frac{\Ai'(x)}{\Ai(x)}\,\cdotp\ee
This function has an asymptotic large $x$ expansion of the form
\be\label{Airylogasy} \varphi(x) = \sum_{r\geq 0}\varphi_{r}x^{\frac{1-3r}{2}}\, .\ee
Using the differential equation $\Ai''(x) = x\Ai(x)$, which implies that $\varphi$ satisfies a Ricatti equation, one finds that the coefficients $\varphi_{r}$ can be computed recursively from
\be\label{varphirrecurs} \varphi_{0}=1\, ,\ \varphi_{1} = \frac{1}{4}\, ,\varphi_{2}=-\frac{5}{32}\, ,\ \varphi_{r+1} = -\frac{3r}{4}\varphi_{r} - \frac{1}{2}\sum_{r'=2}^{r-1}\varphi_{r'}\varphi_{r-r'+1}\ \text{for $r\geq 2$.}\ee
The conjecture states that the derivative of the scaling function defined in \eqref{Fcaldef} is given by
\be\label{SAPconjf1} \mathcal F'(z) = (2a_{0})^{2/3}\varphi\bigl((2a_{0})^{2/3}\pi z\bigr)\, .\ee
The parameter $a_{0}$ is given in Eq.\ \eqref{azeroSAP}. The coefficients $a_{r}$ can be straightforwardly deduced from \eqref{SAPconjf1} by using \eqref{Fcaldef},
\be\label{arcoeffSAP} a_{r} = 2\sqrt{\pi}a_{0}\frac{(-1)^{r+1}r!}{\Gamma\bigl(\frac{3}{2}(r-\frac{1}{3})\bigr)}\frac{\varphi_{r}}{(2\pi^{3/2}a_{0})^{r}}\,\cdotp\ee
Equations \eqref{Wconpathg2} and \eqref{fPath} then yield the exact SAP generating function,
\be\label{WSAPexact} W^{\text{SAP}}(\beta_{\text q},\La) = -2\sqrt{\pi}a_{0}\frac{e^{-\beta_{\text q}\la}}{\beta_{\text q}^{5/2}}\sum_{r=0}^{\infty}\frac{\varphi_{r}}{\Gamma\bigl(\frac{3}{2}(r-\frac{1}{3})\bigr)}\biggl(\frac{\La\beta_{\text q}^{3/2}}{32\pi^{5/2}a_{0}}\biggr)^{r}\, .\ee
The series on the right-hand side of the above equation has an infinite radius of convergence, see Appendix \ref{SAPasymp}. The continuum SAP generating function is thus an entire function, well-defined for any value, positive or negative, of the cosmological constant $\La$.

It is interesting to study the limit for which the dimensionless parameter $\La\beta_{\text q}^{3/2}\rightarrow -\infty$. In contrast to the case of random polygons discussed in the previous subsection, this limit is well defined, since the generating function \eqref{WSAPexact} is an entire function. We show in Appendix \ref{SAPasymp} that\footnote{This result does not seem to have been derived before.}
\be\label{WSAPasym1} W^{\text{SAP}}(\beta_{\text q},\La) \underset{\La\beta_{\text q}^{3/2}\rightarrow -\infty}{=}
\frac{e^{-\beta_{\text q}\la}}{\beta_{\text q}^{5/2}}\,\frac{\La^{2}\beta_{\text q}^{3}}{2^{10}\pi^{5}a_{0}}\,
e^{\frac{\La^{2}\beta_{\text q}^{3}}{3\cdot 2^{12}\pi^{5}a_{0}^{2}}}\,\Bigl(1 + O\bigl(1/(\La^{2}\beta_{\text q}^{3})\bigr)\Bigr)\, .\ee
In particular, the average area is 
\be\label{SAPareaasymp} \langle A\rangle\underset{\La\beta_{\text q}^{3/2}\rightarrow -\infty}{=}
\frac{|\La|\beta_{\text q}^{3/2}}{3\cdot 2^{7}\pi^{4}a_{0}^{2}}\,\beta_{\text q}^{3/2}\biggl[1+\frac{3\cdot 2^{12}\pi^{5}a_{0}^{2}}{\La^{2}\beta_{\text q}^{3}} + O\bigl(1/(\La^{4}\beta^{6})\bigr)\biggr]\, .\ee
We see that there is a well-defined $\La\rightarrow -\infty$ expansion, but is it not a standard semi-classical expansion around a dominant configuration. In particular, the expected area diverges when $\La\rightarrow-\infty$. As in the case of ordinary random polygons reviewed in the previous subsection, this is possible because the microscopic length $\ell = 2n\ell_{0} =\beta/\ell_{0}^{1/3}$ diverges in the continuum limit and consequently the isoperimetric inequality implies no constraint. We can also note that the term governing the exponential growth in \eqref{WSAPasym1} is proportional to $\La^{2}$, not to $\La$, as might have been suggested by the fact that the cosmological constant term is proportional to $\La$, and the corrections are governed by the expansion parameter $1/(\La^{2}\beta^{3})$, not $1/(\La\beta^{3/2})$. 

\subsubsection{\label{RCexSec} Example 3: random convex polygons} 

As a last example, let us present the model of convex polygons. Convex polygons form a subset of SAP polygons for which the intersections of the associated distorted disk with horizontal rows and vertical columns of aligned occupied lattice tiles are all connected. For instance, the polygon on the left of Fig.\ \ref{disdiskbasicFig} is convex, but the polygon in the center of this figure is not. The model is thus an approximation-to-the-approximation from the point of view of JT, but it has the advantage of being exactly solvable \cite{CPsol,SAPtextbook}. Morevoer, it displays qualitative features reminiscent of those expected in the Schwarzian limit of negative curvarture JT.

The universal critical exponents $\nu$ and $\vartheta$ are
\be\label{nuthetaCP} \nu = 1\, ,\quad \vartheta = -1\ee
whereas the non-universal connective constant and coefficient $a_{0}$ are, on the square lattice, $1/g_{*}=2$ and $a_{0} = \frac{1}{256}$. The number of convex polygons thus grows as
\be\label{wnCPform} w_{2n}^{\text{CP}}\underset{n\rightarrow +\infty}{\sim} \frac{n 2^{2n}}{128}\, \cdotp\ee
The moments of the area at zero cosmological constant are given by
\be\label{areamomCP} \langle A^{r}\rangle = \frac{r!^{2}}{2^{2r}(2r+1)!}\beta_{\text q}^{2r}\, .\ee
Equivalently, the exact area distribution law is
\be\label{rhoCP}\rho^{\text{CP}}_{\beta_{\text q}}(A) = \frac{8}{\beta_{\text q}^{2}}\frac{1}{\sqrt{1-16 A/\beta_{\text q}^{2}}}\,\cdotp\ee
One can check that this very simple distribution matches the distribution obtained in the continuum limit of random rectangles, which is elementary to solve.\footnote{Of course, the asymptotic law \eqref{wnCPform} is not valid for rectangles, which simply have $w_{2n}=n-1$. But, in the continuum limit, the area moments \eqref{areamomCP} are the same for both rectangles and convex polygons. This kind of unexpected simplification, or universality, in the limit distributions, is at the basis of the conjecture for the solution of the SAP model, discussed in the previous subsection.}

The upper bound on the area, $A_{\text{max}} = (\beta_{\text q}/4)^{2}$, corresponds to the area of a square of perimeter $\beta_{\text q}$. The existence of the bound is a consequence of the isoperimetric inequality on the square lattice. Unlike the examples of RP and SAP reviewed in the previous subsections, this inequality applies to convex polygons in the continuum because the value of the critical exponent $\nu=1$ implies that the polygon length $\ell = 2n\ell_{0} = \beta_{\text q}$ has a finite limit in the continuum; the ``quantum length'' is simply the usual geometric ``classical'' length in this example. 

By using Eq.\ \eqref{Wconpathg}--\eqref{fPath}, the fixed length continuum generating function is found to be
\be\label{WCPex} W^{\text{CP}}(\beta_{\text q},\La) = \frac{1}{256}e^{-\beta_{\text q}\la}\beta_{\text q} e^{-\frac{\La\beta_{\text q}^{2}}{256\pi}}\int_{0}^{1}e^{\frac{\La\beta_{\text q}^{2}}{256\pi}x^{2}}\d x\, .\ee
It is then straightforward to study the $\La\beta_{\text q}^{2}\rightarrow - \infty$ limit. One finds
\be\label{WCPasymp} W^{\text{CP}}(\beta_{\text q},\La) \underset{\La\beta_{\text q}^{2}\rightarrow - \infty}{=} 
e^{-\beta_{\text q}\la}\frac{\pi}{16|\La|^{1/2}}e^{\frac{|\La|\beta_{\text q}^{2}}{256\pi}}\Bigl[ 1 + O\bigl(e^{-\frac{|\La|\beta_{\text q}^{2}}{256\pi}}\bigr)\Bigr]\, .\ee
In particular, the average area is
\be\label{CPareaasymp} \langle A\rangle \underset{\La\beta_{\text q}^{2}\rightarrow - \infty}{=} 
\frac{\beta_{\text q}^{2}}{16}\biggl[1-\frac{128\pi}{|\La|\beta_{\text q}^{2}}+ O\bigl(e^{-\frac{|\La|\beta_{\text q}^{2}}{256\pi}}\bigr)\biggr]\, .\ee
These results show that the $\La\beta_{\text q}^{2}\rightarrow - \infty$ limit in the convex polygon model is a standard ``semi-classical'' limit. It is dominated by the classical configuration maximising the area, which is the square of perimeter $\beta_{\text q}$ and area $ A_{\text{max}}=(\beta_{\text q}/4)^{2}$ on the square lattice. The leading exponential in \eqref{WCPasymp} is obtained by evaluating the ``classical action'' $S_{\text{cl}}=\frac{\La}{16\pi} A$ on this configuration. Corrections to the leading ``classical'' term are governed by $1/(|\La|\beta_{\text q}^{2})$, as expected for a standard semi-classical expansion. We also observe that the perturbative series stops at one loop: there is no term of order $1/(|\La|\beta_{\text q}^{2})^{L}$ for $L\geq 2$ in the asymptotic expansion \eqref{WCPasymp} of the generating function. 

This behaviour is clearly reminiscent of the effective Schwarzian description of the negative curvature JT model, including the perturbative one-loop exactness. This may have been expected, since convex loops are similar to the smoothly wiggling boundaries described by the reparameterization ansatz. Note, however, that even for this very simple model of convex loops, the model has non-perturbative corrections of order $e^{-\frac{|\La|}{16\pi}A_{\text{max}}}$.

\subsubsection{\label{RPSAPCPcompSec} Model comparison and summary}
\begin{figure}
\centerline{\includegraphics[width=5.25in]{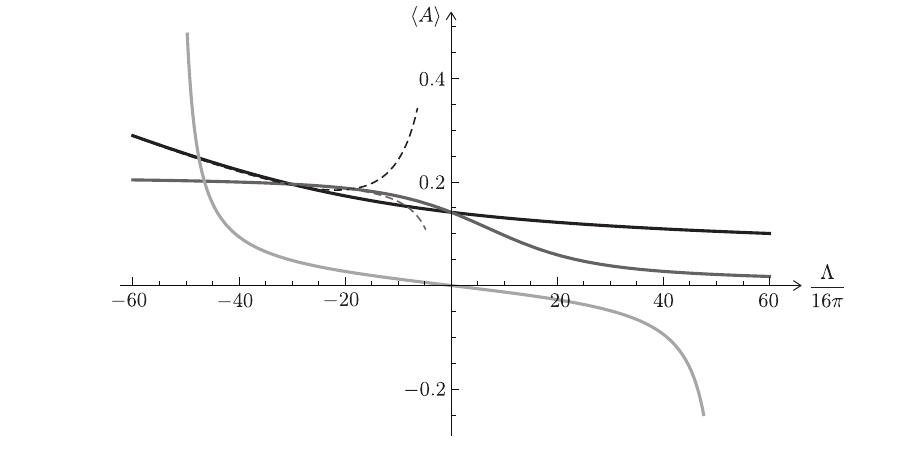}}
\caption{Area expectation values in the SAP (black line), CP (grey line) and R (light grey line) models, as a function of the cosmological constant. Dashed lines correspond to the asymptotic forms \eqref{SAPareaasymp} and \eqref{CPareaasymp}. The parameter $\beta_{\text q}$ for SAP is chosen to be $\beta_{\text{q,\,SAP}}=1$. The parameter $\beta_{\text q}=\beta_{\text{q,\,CP}}=\frac{6}{\pi a_{0}}\simeq 1.843$ for CP is chosen in such a way that the expected area for $\La=0$ matches with SAP. The parameter $\smash{\beta_{\text q} = \beta_{\text{q,R}}=\frac{\sqrt{10-3\pi}}{2\pi^{3/2}a_{0}} \simeq 0.1211}$ for R is chosen in such a way that the quadradic fluctuations $\Delta A$ for $\La=0$ matches with SAP.\label{areaFig}}
\end{figure}

We have reviewed three models of closed random paths. In Fig.\ \ref{areaFig}, the expected area as a function of the cosmological constant is represented for each model, to illustrate some similarities and differences between the three models. 

The R model of purely random Brownian paths is likely to be the best toy model for pure JT quantum gravity when the cosmological constant $\La$ is negative (it does not allow to probe the positive cosmological constant regime). Even though it is not a model of random metrics, the ``boundary'' (``boundary'' in quotation marks, because a purely random loop does not bound a disk in general) has the same Hausdorff dimension as for pure JT. The negative curvature version of this model was studied by Kitaev and Suh in \cite{KitaevSuh}. In zero curvature, there is a maximum allowed value for $|\La|$ beyond which the loop ``explodes''. The model can also be studied in positive curvature, yielding interesting insights for positive curvature JT as well, see Section \ref{nonzeroRSec} and \cite{RDFollowup1}.

The SAP model of self-avoiding loops has the advantage to describe random metrics on the disk. It was proposed by Stanford and Yang in \cite{StanfordSAP} as a microscopic definition of JT gravity. It takes into account only a measure zero subset of the space of constant curvature metrics on the disk. It predicts that the boundary is a non-trivial fractal, but with the wrong Hausdorff dimension. The model does not have a $\La\rightarrow -\infty$ semi-classical expansion and the area for a given $\beta_{\text q}$ is not bounded above, features shared by the SOP model describing JT gravity.

The CP model of convex loops is the simplest of all three models we have reviewed. It has the advantage of being exactly solvable. It describes random metrics on the disk, but an even smaller subset than that taken into account in the SAP model. The boundaries in the continuum limit of the convex loop model are smooth. It is a rather terrible model for flat space JT but shares some basic properties with the Schwarzian description of negative curvature JT in the limit $\La\rightarrow -\infty$.

\subsection{\label{SOPJTSec}Random self-overlapping polygons}

We now focus on the model of random self-overlapping curves, which provides the correct microscopic definition of JT quantum gravity on the disk. We shall use the intuition developped in the previous subsections. Our presentation, and the conjectures we propose, should be seen as a springboard for further research in the future. In particular, we provide only very few proofs, and only of elementary facts.

There are actually many random SOP models one may wish to consider. The model most relevant to us is the one corresponding to the definition \eqref{Zdefbdcurve}, which counts SOP with the multiplicity factor $\mu_{\gamma}$. This model provides the microscopic definition of pure JT gravity. When we talk about \emph{the} SOP model, we refer to this one. One could also consider models associated with the coupling of JT gravity with a bulk QFT, which amounts to counting SOPs with a much more complicated measure. For instance, precise predictions for the case of a coupling to an arbitrary CFT were made in \cite{ferrari}. Another obvious model to contemplate is the model with the uniform measure, for which we count each SOP once without taking into account the multiplicity factor $\mu_{\gamma}$. The generating function, noted with a hat, is then
\be\label{ZdefSOPprime} \hat W (t,g) = \sum_{\gamma\in\text{SOP}}t^{F(\gamma)}g^{v(\gamma)} = \sum_{n,p\geq 1}\hat W_{2n,p}t^{p}g^{2n}\, .\ee
We call this model \SOPu. It is not very natural from the metric point of view, since it amounts to identifying metrics that are not diffeomorphism-equivalent, but, as a model of random paths, \SOPu is interesting in its own right and some properties are easier to prove for \SOPu than for SOP. We discuss SOP and \SOPu in parallel below.

\subsubsection{Basic results and conjectures}

A crucial basic property we would like to understand is the exponential growth of the numbers $w_{2n}$ and $\hat w_{2n}$ of self-overlapping polygons counted with and without multiplicity, respectively. This amounts to proving the existence of the connective constant $1/g_{*}$ for these models. A simple proof, adapting well-known ideas, can be made in the case of the \SOPu model.

\begin{theorem}\label{connectiveTh} Let $\hat w_{2n}$ be the number of self-overlapping polygons of length $2n$ in the \SOPu model (for which polygons are counted with the uniform measure). Then the connective constant
\be\label{conncstSOPhat} \lim_{n\rightarrow\infty}\bigl(\hat w_{2n}\bigr)^{\frac{1}{2n}} = 1/\hat g_{*}\ee
exists and is finite.
\end{theorem}

Note that the connective constants of the various models we have studied automatically satisfy
\be\label{connectcstineq} 1/g_{*}^{\text R}\geq  1/\hat g_{*}\geq 1/g_{*}^{\text{SAP}}\geq 1/g_{*}^{\text{CP}}\, ,\ee
simply because there are more arbitrary polygons than SOP, more SOP than SAP, and more SAP than CP. These inequalities are true on any lattice, in spite of the fact that the precise values of the connective constants depend on the lattice.\footnote{We recall that on the square lattice $1/g_{*}^{\text R}=4$, $1/g_{*}^{\text{SAP}}\simeq 2.638$ and $1/g_{*}^{\text{CP}} =2$. It is natural to conjecture that \eqref{connectcstineq} is valid with strict inequalities.}

In the case of SOP, we conjecture that a similar result is true, but we shall not try to provide a proof.

\begin{conjecture}\label{connectiveConj} Let $w_{2n}$ be the number of self-overlapping polygons of length $2n$ in the SOP model, i.e.\ counted with the multiplicity factor $\mu_{\gamma}$. Then the connective constant
\be\label{conncstSOP} \lim_{n\rightarrow\infty}\bigl(w_{2n}\bigr)^{\frac{1}{2n}} = 1/g_{*}\ee
exists and is finite.
\end{conjecture}
Note that the inequality
\be\label{connectcstineq2} 1/g_{*}\geq  1/\hat g_{*}\ee
is then trivially satisfied.

The proof of Th.\ \ref{connectiveTh} is based on the super-multiplicativity of the sequence $(\hat w_{2n})_{n\geq 2}$. More generally, one has the following. 

\begin{proposition} If $\mathcal W_{2n,p,\mu}$ denotes the number of SOP of boundary length $2n$, area $p$ and multiplicity $\mu$, then
\be\label{superadd3} \mathcal W_{2(n_{1}+n_{2}),p_{1}+p_{2}+1,\mu_{1}\mu_{2}}\geq \mathcal W_{2n_{1},p_{1},\mu_{1}}\mathcal W_{2n_{2},p_{2},\mu_{2}}\, .\ee
Similarly, the sequences $(\hat W_{2n,p})_{n\geq 2,p\geq 1}$ and $(\hat w_{2n})_{n\geq 2}$ satisfy
\be\label{superadd12}\hat W_{2(n_{1}+n_{2}),p_{1}+p_{2}+1}\geq \hat W_{2n_{1},p_{1}}\hat W_{2n_{2},p_{2}}\, ,\quad \hat w_{2(n_{1}+n_{2})}\geq \hat w_{2n_{1}}\hat w_{2n_{2}}\, .\ee
\end{proposition}
\begin{proof} We adapt a well-known gluing argument used to prove the super-multiplicativity property in the case of SAP.

\begin{figure}
\centerline{\includegraphics[width=6in]{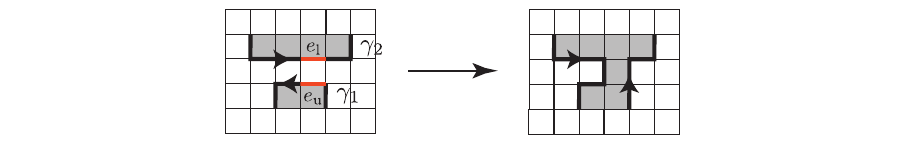}}
\caption{\label{glueFig}Gluing of two SOP $\gamma_{1}$ and $\gamma_{2}$ along an upper boundary edge $e_{\text u}$ of $\gamma_{1}$ and a lower boundary edge $e_{\text l}$ of $\gamma_{2}$ (outlined in orange) to form a new SOP $\gamma_{1}\sqcup\gamma_{2}$. Some faces of the distorted disks bounded by $\gamma_{1}$, $\gamma_{2}$ and $\gamma_{1}\sqcup\gamma_{2}$ are in grey. This construction is used in the proof of the superadditivity inequalities \eqref{superadd12}.}
\end{figure}

We introduce a Cartesian coordinate system $(x,y)\in\mathbb Z^{2}$ on the square lattice. For any self-overlapping polygon on the lattice, we call ``upper boundary piece'' any segment of the polygon that contains horizontal edges having the maximum possible value of $y$ and we call similarly ``lower boundary piece'' any segment of the polygon that contains horizontal edges having the minimum possible value of $y$.\footnote{Note that distinct lower or upper boundary pieces may be located on top of each other, as in the example of Fig.\ \ref{figMilnormin}.} The code for a upper or a lower boundary piece is necessarily of the form $\text{LS}\cdots\text{SL}$, with $q\geq 0$ Straight vertices appearing in between the two left turns. We call a boundary edge upper (or lower) if it a horizontal edge belonging to a upper (lower) boundary piece.\footnote{Note that we could similarly introduce the notions of rightmost and leftmost boundary pieces and rightmost and leftmost edges.} Note that any polygon has at least one upper and one lower edge.

We then define a notion of gluing two self-overlapping polygons $\gamma_{1}$ and $\gamma_{2}$ along an upper boundary edge $e_{\text u}$ of $\gamma_{1}$ and a lower boundary edge $e_{\text l}$ of $\gamma_{2}$ as depicted in Fig.\ \ref{glueFig}. The new polygon obtained after the gluing is noted $\gamma=\gamma_{1}(e_{\text u})\sqcup\gamma_{2}(e_{\text l})$. We also use the  ambiguous notation $\gamma=\gamma_{1}\sqcup\gamma_{2}$ when the precise choice of upper and lower boundary edges is irrelevant.

Crucially, if $\gamma_{1}$ and $\gamma_{2}$ are self-overlapping, the new polygon $\gamma=\gamma_{1}\sqcup\gamma_{2}$ is self-overlapping. This is obvious from the definition given in Section \ref{diskviewSec}: if $\gamma_{1}$ and $\gamma_{2}$ bound distorted disks, then $\gamma$ also bounds a distorted disk. Moreover, if we note $2n_{\gamma}$, $p_{\gamma}$ and $\mu_{\gamma}$ the boundary length, area and multiplicity of a SOP $\gamma$, then clearly
\be\label{glueid1}  n_{\gamma_{1}\sqcup\gamma_{2}} = n_{\gamma_{1}}+n_{\gamma_{2}}\, ,\quad p_{\gamma_{1}\sqcup\gamma_{2}} = p_{\gamma_{1}}+p_{\gamma_{2}}+1\, ,\quad \mu_{\gamma_{1}\sqcup\gamma_{2}} = \mu_{\gamma_{1}}\mu_{\gamma_{2}}\, .\ee
If we denote by $\text{SOP}_{2n,p,\mu}$ the set of SOP of length $2n$, area $p$ and multiplicity $\mu$, we can use the gluing to construct a map
\be\label{mapsuperadd1} \sqcup:\text{SOP}_{2n_{1},p_{1},\mu_{1}}\times \text{SOP}_{2n_{2},p_{2},\mu_{2}}\rightarrow \text{SOP}_{2(n_{1}+n_{2}),p_{1}+p_{2}+1,\mu_{1}\mu_{2}}\, .\ee
There are of course many possible maps, depending on the detailed choice of upper and lower boundary edges used in the gluing of each pair of elements of $\text{SOP}_{2n_{1},p_{1},\mu_{1}}\times \text{SOP}_{2n_{2},p_{2},\mu_{2}}$. We work with a given choice. This map is injective. This is immediately checked by noting that in $\gamma_{1}\sqcup\gamma_{2}$, the distorted disk associated with $\gamma_{1}$ is entirely below the distorted disk associated with $\gamma_{2}$. The equality $\gamma_{1}\sqcup\gamma_{2} =   \gamma'_{1}\sqcup\gamma'_{2}$, together with the fact that the lengths of $\gamma_{i}$ and $\gamma_{i}'$ are given by $2n_{i}$ and are thus the same, then implies that $\gamma_{i}=\gamma_{i}'$. The existence of the injective map $\sqcup$ implies the inequality \eqref{superadd3}. Moreover, denoting by $\text{SOP}_{2n,p}$ and $\text{SOP}_{2n}$ the set of SOP of length $2n$ and area $p$, or the set of SOP of boundary length $2n$, respectively, the gluing restricts to injective maps
\begin{align}\label{mapsuperadd2} & \sqcup:\text{SOP}_{2n_{1},p_{1}}\times \text{SOP}_{2n_{2},p_{2}}\rightarrow \text{SOP}_{2(n_{1}+n_{2}),p_{1}+p_{2}+1}\, ,\\
\label{mapsuperadd3} & \sqcup:\text{SOP}_{2n_{1}}\times \text{SOP}_{2n_{2}}\rightarrow \text{SOP}_{2(n_{1}+n_{2})}\, .
\end{align}
The existence of these maps imply the inequalities \eqref{superadd12}. 

By Fekete's lemma, the super-multiplicativity property of the sequence $(\hat w_{2n})$ implies that the limit in \eqref{conncstSOPhat}. A priori, the limit may be infinite. However, we have the trivial bound $\hat w_{2n}\leq 4\times 3^{2n}$ by the total number of possible boundary codes times the factor of four due to lattice rotations. Theorem \ref{connectiveTh} is thus proved. 
\end{proof}

The SOP model is more subtle to analyse, because the measure is not uniform and, moreover, we do not know the behaviour of the multiplicity factor. It is not difficult to construct SOP for which the multiplicity grows exponentially with the boundary length.\footnote{I would like to thank Peter Shor for pointing out this possibility.} For example, by gluing $r$ times the smallest Milnor SOP depicted in Fig.\ \ref{figMilnormin}, we get a SOP of length $2n = 48 r$ that has multiplicity $2^{r}=2^{n/24}$. It is natural to conjecture that the multiplicity cannot grow faster than an exponential of the length.

\begin{conjecture}\label{multconj} The multiplicity of a self-overlapping polygon of length $2n$, that is to say, the number of distinct distorted disks it bounds, is bounded above by $A B^{2n}$ for some strictly positive constants $A$ and $B$.
\end{conjecture}

This would imply in particular that $w_{2n}\leq A B^{2n}\hat w_{2n}\leq 4A (3B)^{2n}$, i.e.\ $\ln w_{2n}\leq O(n)$. Such a bound is natural for random path models but rather non-trivial  for a random geometry model. From the metric point of view, we can make the following reasoning. The total number of quadrangulations of the disk, with boundary length $2n$ and $p$ faces, is $\mathsf W_{2n,p}$, given in Eq.\ \eqref{WpndiskL}. It is easy to check that $\mathsf W_{2n,p+1}>\mathsf W_{2n,p}$. Moreover, quadrangulations corresponding to flat disks necessarily have $n-1\leq p \leq[n^{2}/4]$, see e.g.\ the discussion after Hyp.\ \ref{hyp1} in Section \ref{GenRPmodSec}. Thus
\be\label{roughineq} w_{2n}\leq\sum_{p=n-1}^{[n^{2}/4]}\mathsf W_{2n,p}\leq \bigl([n^{2}/4]-n+2\bigr)\mathsf W_{2n,[n^{2}/4]}\, .\ee
The large $n$ asymptotics of $\mathsf W_{2n,[n^{2}/4]}$ then yields the bound $\ln w_{2n}\leq O(n^{2})$. This is much weaker than the bound $\ln w_{2n}\leq O(n)$ that follows from the Conjecture \ref{multconj}.

More generally, we conjecture that the SOP model fits in the general framework discussed in Section \ref{randompathrev}. In particular,
\begin{conjecture} The critical exponents $\nu$ and $\vartheta$, which govern the way the continuum limit is defined, Eq.\ \eqref{contpathdef}, and the partition function at zero cosmological constant, Eq.\ \eqref{WPgenzeroL}, are given by
\be\label{nuthetaSOP} \nu = \frac{1}{2}\,\cvp\quad\vartheta = 2\, .\ee
\end{conjecture}
This conjecture is motivated by the continuum approach to JT gravity introduced in \cite{ferrari}. Note that logarithmic corrections to the asymptotics given in Eq.\ \eqref{w2nras} cannot be excluded and might be suggested by the fact that pure JT gravity seems to share similarities with Liouville theory coupled to $c=1$ matter \cite{ferrari}. For instance, 
\begin{conjecture} The large $n$ asymnptotic behaviour of the number $w_{2n}$ of SOP, counted with multiplicity, is of the form
\be\label{wSOPconjasy} w_{2n} \underset{n\rightarrow\infty}{\sim} a_{0} g_{*}^{-2n}(2n)^{-3}(\ln n)^{\varphi}\ee
for a connective constant $1/g_{*}>1/g_{*}^{\text{SAP}}$ and a new critical exponent $\varphi$ that may be non-zero.
\end{conjecture}

\subsubsection{Observables}

Beyond the standard boundary length, area\footnote{Note that for SOP, the notions of L\'evy and winding areas defined in \eqref{notionsofarea} coincide and match with the area of the disk computed with the associated metric.} and radius of gyration discussed before, there are several interesting observables to study. A non-exhaustive list is as follows (many variants are possible).

i) The arithmetic area $A_{\text{arith.}}$, defined in \eqref{notionsofarea}, which provides an intuitive measure of the size of the SOP drawn on the Euclidean plane. This is expected to yield a notion of size which is equivalent to the radius of gyration, with a scaling exponent $\nu'$, see Eq.\ \eqref{gyrscaling}. It is interesting to recall that, in the case of the Brownian path model, the arithmetic area and the winding area behaves very differently, see Eqs.\ \eqref{arithareaav} and \eqref{windareaav}. In the case of the SOP model it is natural to conjecture that $\nu'=\nu$.

ii) The number $\iota_{\gamma}$ of bounded simply connected components of $\mathbb R^{2}\backslash\gamma$ or, equivalently, the number of self-intersections $\iota_{\gamma}-1$ of the boundary curve.\footnote{In the continuum limit, a generic curve will self-intersect at double points only. Such a curve drawn on the plane thus yields a four-regular planar graph, whose number of vertices, which equals the number of self-intersections, must be the number of faces, given by $\iota_{\gamma}+1$, minus 2.} The average area of the bounded simply connected components is $a=\frac{1}{\iota_{\gamma}}A_{\text{arith.}}$.

iii) The average number of overlaps $\nu_{\gamma} = \frac{1}{\iota_{\gamma}}\sum_{i}\nu_{i}$, where $\nu_{i}$ is the number of overlaps in the $\smash{i^{\text{th}}}$ bounded simply connected component of $\smash{\mathbb R^{2}\backslash\gamma}$, as defined in Section \eqref{overlapSec}.

iv) The hull of the SOP is the closed curve that bounds the region occupied by the bounded simply connected components of $\mathbb R^{2}\backslash\im\gamma$. In the case of Brownian loops, a famous conjecture by Mandelbrot \cite{Mandelbrot,Mandelbrotproof} states the the hull is the continuum limit of the SAP model, with Hausdorff dimension $4/3$. It is natural to conjecture that the hull of ther SOP model will also be a $\text{SLE}_{\kappa}$ curve, for a certain value of $\kappa$.

iv) The multiplicity $\mu_{\gamma}$.

It is interesting to study all these observables in the SOP and in the \SOPu models as well as in the natural generalizations of SOP associated with the coupling to a CFT of arbitrary central charge $c$.

The most thought-provoking observable of all is undoubtedly the multiplicity, which is a very subtle and surprising feature of the quantum gravity model, that has no counterpart in other random path models. Are the SOP and \SOPu models equivalent in the continuum limit, with the same critical exponents, partition functions, etc.? This would be the case if the space of SOP having non-trivial multiplicities were of measure zero in the continuum limit. Or, on the contrary, does multiplicity play a crucial role in the quantum gravity models, the set of configurations having a non-trivial multiplicity being of measure strictly greater than zero? Is it actually essential to ensure, for instance, the consistency with quantum mechanics? Aesthetically speaking, the possibility that an effect as subtle could play a fundamental role in quantum gravity is, of course, extremely attractive. 

\subsubsection{\label{areaSOPcomnjSec}The area distribution law}

The area distribution law $\rho_{\beta_{\text q}}^{\text{SOP}}$ of the SOP model in the continuum limit governs the $\La$-dependence of the partition function, Eq.\ \eqref{Wconpathg}. Its large area asymptotics provides an important qualitative information on the theory and governs, in particular, the behaviour of the theory when $\La$ is negative and large. We conjecture the following.
\begin{conjecture}\label{areazeroconj1}
The area distribution of the SOP model, i.e.\ pure JT gravity in zero curvature, decays exponentially at large area, with 
\be\label{rhodecayzero} \ln\rho^{\text{SOP}}_{\beta_{\text q}}(A)\underset{A\rightarrow +\infty}{\sim} -\kappa A/\beta_{\text q}\, ,\ee
where $\kappa$ is a strictly positive numerical constant.
\end{conjecture}
This conjecture implies that JT gravity in zero curvature is defined only for $\La>\La_{\text c}$, for a strictly negative critical cosmological constant $\La_{\text c} = -16\pi\kappa/\beta_{\text q}$.

The motivation for this conjecture comes from the general discussion in Section \ref{phasessubSec}. At the end of this subsection, it is explained that, when $\nu=1/2$, which is the expected value for SOP, the so-called inflated crumpled phase a priori exist only for a finite range of the parameter $n^{2\nu}(t-1)$. In the continuum limit, this suggests that the theory exists only for a finite range of negative cosmological constant. This is also observed in the purely random polygon model, which has $\kappa = 2\pi$, see Eq.\ \eqref{rhobetaRP}.

When JT gravity is coupled to $c<0$ conformal matter, the exponent $\nu$ was predicted in \cite{ferrari} to be as in Eq.\ \eqref{critexponentsJT}, which implies that $1/2<\nu<1$, the limit $\nu=1$ being reached in the semi-classical regime $c\rightarrow -\infty$ studied in \cite{Loopcalc}. In these cases, based on the discussion of Section \ref{phasessubSec}, we expect the theory in flat space to exist for arbitrary values of the cosmological constant. We  make the following conjecture.
\begin{conjecture}\label{areazeroconj2} The area distribution of JT gravity in zero curvature coupled to $c\leq 0$ conformal matter, which has a critical exponent $\nu$ given by Eq.\ \eqref{critexponentsJT}, has the asymptotic behaviour
\be\label{rhodecayzeroc} \ln\rho^{0}_{\beta_{\text q}}(A)\underset{A\rightarrow +\infty}{\sim} - \kappa_{\nu} \bigl(A/\beta_{\text q}^{2\nu}\bigr)^{\chi_{\nu}}\, ,\ee
where $\kappa_{\nu}$ is a strictly positive numerical constant and $\chi_{\nu}$ a critical exponent, which is a strictly increasing function of $\nu$ satisfying $\lim_{\nu\rightarrow \frac{1}{2}^{+}}\chi_{\nu}=1$ and $\lim_{\nu\rightarrow 1^{-}}\chi_{\nu} = +\infty$.
\end{conjecture}
The fact that $\chi_{\nu}>1$ for $\nu>1/2$ ensures the existence of the theory for all $\La$. The conjectured limit when $\nu\rightarrow\frac{1}{2}^{+}$ is consistent with Eq.\ \eqref{rhodecayzero}. The conjectured limit when $\nu\rightarrow 1^{-}$ is consistent with the semi-classical analysis in \cite{Loopcalc}. Actually, a straightforward saddle-point analysis, using Eq.\ \eqref{Wconpathg}, shows that Eq.\ \eqref{rhodecayzeroc} is equivalent to the following asymptotic behaviour of the JT gravity partition function,
\be\label{Wzeroasympconj} \ln W^{0}(\beta_{\text q},\La) \underset{\La\beta_{\text q}^{2\nu}\rightarrow -\infty}{\sim} \frac{\chi_{\nu}-1}{\chi_{\nu}} (\chi_{\nu}\kappa_{\nu})^{-\frac{1}{\chi_{\nu}-1}}
\biggl(\frac{|\La|\beta_{\text q}^{2\nu}}{16\pi}\biggr)^{\frac{\chi_{\nu}}{\chi_{\nu}-1}}\, ,\ee
or, equivalently, of the expected area,
\be\label{zeroexpareaconj} \langle A\rangle \underset{\La\beta_{\text q}^{2\nu}\rightarrow -\infty}{\sim} 
\biggl(\frac{|\La|\beta_{\text q}^{2\nu}}{16\pi\chi_{\nu}\kappa_{\nu}}\biggr)^{\frac{1}{\chi_{\nu}-1}} \beta_{\text q}^{2\nu}\, .\ee
Comparing with the $c\rightarrow -\infty$ results in \cite{Loopcalc}, we find that
\be\label{chizeroasympsc} \chi_{\nu}\underset{c\rightarrow -\infty}{\sim} \frac{|c|}{12}\underset{\nu\rightarrow 1^{-}}{\sim}\frac{1}{2(1-\nu)}\,\cdotp\ee
It is tempting to conjecture that this asymptotic law as a function of $\nu$ is actually exact.
\begin{conjecture} The exponent $\chi_{\nu}$ governing the asymptotic behaviour of the area distribution, Eq.\ \eqref{rhodecayzeroc}, is given by 
\be\label{chinyconjf} \chi_{\nu} = \frac{1}{2(1-\nu)}\,\cdotp\ee
\end{conjecture}
This is consistent with the limit $\smash{\lim_{\nu\rightarrow\frac{1}{2}^{+}}\chi_{\nu} = 1}$. Moreover, the formula predicts that $\chi_{3/4}=2$ and thus $\langle A\rangle \propto |\La|$ in the limit $\La\rightarrow -\infty$ in this case. This is the behaviour found for the SAP model, see Eq.\ \eqref{SAPareaasymp}. It is natural to expect that flat space JT for $\nu=3/4$ (which corresponds to $c=-8$) and the SAP model will have similar properties in the large geometry, $\La\rightarrow - \infty$, limit.

\section{\label{mmSec}The matrix model formulation}

\subsection{\label{MMintSec}Introduction}

Matrix model techniques are extremely useful in the case of Liouville gravity and it is thus natural to look for a matrix model formulation of JT gravity. We construct such a matrix model in this section, focusing on flat JT. The possible relevance of the construction to the cases of positive and negative curvature will be discussed in Section \ref{nonzeroRSec}.

To avoid confusion, let us immediately emphasize that the matrix model are discussing below is conceptually distinct from the matrix model proposed by Saad, Shenker and Stanford in \cite{SSS}. Our matrix model is designed to compute the generating function $W(t,g)$ or, more generally, the generalized generating function $\mathscr W$ defined in \eqref{Zdefgeneralized}. It is thus a reformulation of the combinatorial problem of counting self-overlapping curves, taking into account the multiplicity index, or more generally discretized flat metrics on an arbitrary surface with boundaries. This is along the lines of the matrix model formulation of Liouville gravity, but it is very different from \cite{SSS}. The matrix model in \cite{SSS} is supposed to describe a random ensemble of quantum Hamiltonian to which JT gravity is dual in a certain limit. 

The basic idea underlying our construction, the so-called dually weighted graphs, dates back to a remarkable work of Di Francesco and Itzykson \cite{difrancescoitzykson}. Our contribution is to explain how to correctly introduce boundaries in this picture, which amounts to constructing the JT gravity loop operators. 

Models of dually weighted graphs are much more difficult to study than the usual Hermitian one matrix model or other similar models, but they have been solved in some interesting cases \cite{KSW}. The model relevant to JT gravity cannot be solved straightforwardly with the existing technology, but one may hope that a solution will eventually be found by appropriately extending the powerful formalism developed in \cite{KSW}. 

As a preparation for what follows, let us recall that the disk generating function at fixed boundary length in Liouville gravity is obtained by considering the expectation value of the so-called macroscopic loop operator $\tr X^{2n}$ in the matrix integral, Eq.\ \eqref{W2ndLMM}. By summing, the full generating function \eqref{ZfullLM} is thus 
\be\label{Wbgene11} \mathsf W(t,g) = \frac{1}{N}\Bigl\langle \ln\frac{1}{\det (\mathbb I - gX)}\Bigr\rangle =\frac{1}{N}
\ln\Bigl\langle \frac{1}{\det (\mathbb I - gX)}\Bigr\rangle\quad \text{when}\ N\rightarrow\infty\, .\ee
The second equality is valid because we consider only the leading large $N$ approximation, which is given by a saddle-point for the matrix $X$. It is then natural to introduce a pair of complex vector variables $(\tilde Q_{i},Q^{i})_{1\leq i\leq N}$ and consider the free energy $\mathsf F$ defined by
\be\label{Fgenwithvect} e^{\mathsf F(t,g)} = \frac{\int\d X\d\tilde Q\d Q\, e^{-N\tr (\frac{1}{2}X^{2}-\frac{t}{4}X^{4})-N\tilde Q(\mathbb I - g X) Q}}{\int\d X \,e^{-N\tr(\frac{1}{2}X^{2}-\frac{t}{4}X^{4})}}\, \cdotp\ee
Using the identity
\be\label{detvectorid} \frac{1}{\det (\mathbb I - gX)} = \int\!\d\tilde Q\d Q\, e^{-N\tilde Q(\mathbb I - g X) Q}\, ,\ee
we get the large $N$ expansion
\be\label{freeenergy} \mathsf F(t,g) = N^{2} W^{(\text s)}(t) + N\mathsf W(t,g) + O(N^{0})\, ,\ee
where $W^{(\text s)}(t)$ is the sphere generating function defined in \eqref{WLioudef}. 

So boundaries can be introduced either by inserting macroscopic loop operators $\tr X^{2n}$ or by introducing vector variables.

\subsection{\label{duallywSec}Dually weighted graphs}

\subsubsection{Basic idea}

When we write \eqref{Wbgene11}, we use the dual point of view, for which an interaction term $\tr X^{4}$ in the matrix model action yields a discretized surface built out of squares and the insertion of $\tr X^{2n}$ corresponds to a marked face of length $2n$. More generally, with an action
\be\label{mmActgen1} S(X;\{t\}) = N\tr\Bigl(\frac{1}{2} X^{2} - \sum_{p}\frac{t_{p}}{p}X^{p}\Bigr)\, ,\ee
we get discretized surfaces built out of $p$-gons, the number of $p$-gons being counted by the power of the coupling $t_{p}$ in the corresponding Feynman graph.

To build a matrix model for flat JT gravity, we want to fix the bulk curvature of the discretized surfaces to zero. According to Eq.\ \eqref{Rpqform}, we thus need to have control on both the order of the vertices and on the lengths of the faces that can appear in the matrix model graphs. Associated with the three regular tesselations of the flat Euclidean plane, we have three possibilities: if we use square tiles, $p=4$, only faces of length 4 are allowed; if we use hexagonal tiles, $p=6$, only faces of length 3 are allowed; and if we use triangular tiles, $p=3$, only faces of length 6 are allowed. 

The trick to do this is to introduce a new set of coupling constants $\tilde t_{q}$ in the model, in such a way that each face of length $q$ in the Feynman graph produces a factor $\tilde t_{q}$, in exactly the same way as each vertex of degree $p$ produces a factor $t_{p}$. This is the idea of ``dually weighted graphs,'' developed by Di Francesco, Itzykson \cite{difrancescoitzykson} and Kostov. One considers an external background Hermitian $N\times N$ matrix $A$ and sets
\be\label{tildetAdef} \tilde t_{q} = \frac{1}{N}\tr A^{q}\, .\ee
The action \eqref{mmActgen1} is generalized to
\be\label{mmActgen2} S(X;\{t\}, \{\tilde t\}) = N\tr\Bigl(\frac{1}{2} X^{2} - \sum_{p}\frac{t_{p}}{p}(AX)^{p}\Bigr)\, .\ee
It is then straightforward to check that each face of a Feynman graph now produces a factor of $\tr A^{q} = N\tilde t_{q}$ instead of the usual factor of $N$, exactly as looked for. If
\be\label{ZDFIdef} e^{\mathcal F(\{t\},\{\tilde t\})} = \int\!\d X\, e^{-S(X;\{t\}, \{\tilde t\})}\ee
one has the large $N$ expansion
\be\label{Wttildetexp} \mathcal F(\{t\},\{\tilde t\}) = \sum_{h\geq 0}N^{2-2h}\sum_{\{p_{i},q_{j}\}}\mathcal W_{\{p_{i},q_{j}\}}^{(h)} \prod_{i,j}t_{i}^{p_{i}}\tilde t_{j}^{q_{j}}\, ,\ee
where $\mathcal W_{\{p_{i},q_{j}\}}^{(h)}$ counts the number of genus $h$ Feynman graphs with $p_{i}$ vertices of degree $i$ and $q_{j}$ faces of length $j$ or, equivalently, in the dual picture, the number of discretized surfaces of genus $h$ made of $p_{i}$ $i$-gons with $q_{j}$ vertices of degree $j$.\footnote{As usual, the matrix integrals are normalized to one when the couplings are zero and the counting is made modulo the order of the automorphism groups of the graphs. See below for a discussion of these symmetry factors in the cases of interest for JT gravity.}

It is natural to introduce another background Hermitian matrix $\tilde A$ and write
\be\label{ttildeAdef} t_{p} = \frac{1}{N}\tr\tilde A^{p}\, .\ee
The free energy $\mathcal F$ is then a function of the external matrices $A$ and $\tilde A$. Standard manipulations show that
\be\label{dualWform}
\begin{split}
e^{\mathcal F(A,\tilde A)} & = \int\!\d X\d V\d V^{\dagger}\, e^{-\frac{N}{2}\tr X^{2} - N\tr V^{\dagger} V + N\tr V\tilde A V^{\dagger}AX}\\
& = \int\!\d \tilde X\d V\d V^{\dagger}\, e^{-\frac{N}{2}\tr \tilde X^{2} - N\tr V^{\dagger} V + N\tr V^{\dagger}A V\tilde A\tilde X}\\
& = \int\!\d V\d V^{\dagger}\, e^{-N\tr V^{\dagger}V+\frac{N}{2}\tr (V^{\dagger}AV\tilde A)^{2}}\, ,
\end{split}
\ee
where $V$ is a complex matrix and $\tilde X$ a ``dual'' Hermitian matrix. These formulas make manifest the duality 
\be\label{dualityact} A\leftrightarrow \tilde A\, ,\ X\leftrightarrow \tilde X\, ,\ V\leftrightarrow V^{\dagger}\ee
which exchanges vertices and faces. One has
\be\label{duality1} \mathcal F(A,\tilde A) = \mathcal F(\tilde A, A)\, .\ee
Eq.\ \eqref{dualWform} also makes manifest an underlying $\text U(N)\times\widetilde{\text U(N)}$ covariance,
\be\label{UNtimesUN} \mathcal F (A,\tilde A) = \mathcal W(\Omega A\Omega^{\dagger}, \tilde\Omega \tilde A\tilde\Omega^{\dagger})\, ,\ee
under which the variables transform as
\be\label{UNUNtransform} A\mapsto \Omega A\Omega^{\dagger}\, ,\ \tilde A\mapsto \tilde\Omega \tilde A \tilde\Omega^{\dagger}\, ,\ X\mapsto \Omega A\Omega^{\dagger}\, ,\
\tilde X\mapsto \tilde\Omega \tilde X\tilde\Omega^{\dagger}\, ,\ V\mapsto \Omega V \tilde\Omega^{\dagger}\, ,\ V^{\dagger}\mapsto \tilde\Omega V^{\dagger}\Omega^{\dagger}\, ,\ee
where $\Omega$ and $\tilde\Omega$ are independent unitary matrices. The two independent $\uN$ symmetries are associated with the strands forming the faces in the direct and dual representations of the Feynman graphs, respectively.

Let us remark that at finite $N$, the definitions \eqref{tildetAdef} and \eqref{ttildeAdef} imply that only $(\tilde t_{1},\ldots,\tilde t_{N})$ and $(t_{1},\ldots,t_{N})$ are independent couplings. But in the formal large $N$ expansion of the matrix integrals we are interested in, it is consistent to consider that all the couplings are independent. This is important for our purposes, since to impose the flatness condition in the bulk, we choose
\be\label{quadrangchoice} t_{p} = \delta_{p,p_{0}}t\, ,\ \tilde t_{q} = \delta_{q,q_{0}}\quad  \text{for}\quad (p_{0},q_{0}) = (4,4),\ (6,3)\ \text{or}\ (3,6)\, .\ee

\subsubsection{The Di Francesco-Itzykson formula}

The integral \eqref{ZDFIdef} can be expressed explicitly as a sum over Young tableaux, or partitions \cite{difrancescoitzykson}. Explicitly, for a partition $\la = (\la_{1},\ldots,\la_{N})$, $\la_{1}\geq\la_{2}\geq\cdots\geq\la_{N}\geq 0$, we define $h_{i} = N-i+\la_{i}$. We note $\La_{N}$ the set of partitions such that: if $N=2n$ is even, $n$ of the $h_{i}$s are even (noted $h_{i}^{\text e}$) and $n$ are odd (noted $h_{i}^{\text o}$); if $N=2n-1$ is odd, $n$ of the $h_{i}$s are even and $n-1$ are odd. We then have
\be\label{DIfIform} e^{\mathcal F(A,\tilde A)} = (-1)^{n(n-1)/2}\sum_{\la\in\La_{N}}
\frac{\prod_{i}(h_{i}^{\text e}-1)!!\prod_{j}h_{j}^{\text o}!!}{\prod_{i,j}(h_{j}^{\text o}-h_{i}^{\text e})} s_{\la}(A) s_{\la}(\tilde A)\, ,\ee
where $s_{\la}$ is the Schur polynomial. One can show that summing over the partitions of a given size $|\la| = \sum_{i}\la_{i}$ yields the sum over the Feynman graphs that have a given number of edges $E=|\la|/2$.

The main advantage of this formula is that the number of degrees of freedom is reduced from $N^{2}$ in the original matrix integral to $N$ in the sum over partitions $\la=(\la_{1},\ldots,\la_{N})$. This is technically similar to the usual diagonalization method, the partitions replacing the eigenvalues. It is then tempting to try to evaluate \eqref{DIfIform} in the large $N$ limit using a saddle point method, assuming that a very large partition will dominate the sum. This idea is at the basis of the results presented in \cite{KSW}.

Let us emphasize, however, that the assumption that the sum is dominated by a single partition is very non-trivial in itself, in particular because the sum in \eqref{DIfIform} involves terms of positive and negative signs. In relation to this, there are cases for which the free energy is easy to evaluate directly from the Feynman graphs, whereas the matrix integral, or the sum \eqref{DIfIform}, seem to be very complicated to compute directly. For instance, if we use flat quadrangulated surfaces, $(p_{0},q_{0}) = (4,4)$, then by Gauss-Bonnet we know that only the torus topology can contribute to $\mathcal F$ and one finds
\be\label{DiFIgenusone} e^{\mathcal F} = \frac{1}{\prod_{k\geq 1}(1-t^{k})^{1/4}}\ee
by examining the possible Feynman graphs. We have checked that this formula matches with \eqref{DIfIform} up to order $t^{12}$ in the small $t$ expansion, but we do not know a simple derivation of the identity \eqref{DiFIgenusone} starting from \eqref{DIfIform}, nor do we know how to obtain this result from a saddle point.

\subsubsection{\label{nogoSec}No-go theorems}

The flatness condition in the bulk is enforced by choosing the couplings as in \eqref{quadrangchoice}. The resulting models are rather trivial, as Eq.\ \eqref{DiFIgenusone} shows. To make the link with JT gravity, we need to introduce boundaries. 

It is very natural to try to do so by inserting macroscopic loop operators, which is the standard tool in Liouville gravity as we have reviewed in Section \ref{LiouclimitSec}. There are four natural loop operators that one may consider in the models, which are also the observables studied in \cite{KSW}:
\be\label{possibleloops} W_{L}^{(1)}(t) = \frac{1}{L}\bigl\langle\frac{1}{N}\tr X^{L}\bigr\rangle\, , \ W_{L}^{(2)}(t) = \frac{1}{L}\bigl\langle\frac{1}{N}\tr (AX)^{L}\bigr\rangle\, ,\ee
together with their duals. The aim of this subsection is to show that these standard loop operators are trivial in our case.

\begin{proposition} $W_{L}^{(1)}(t)$ is independent of $t$ and is thus equal to the $t=0$, Gaussian matrix model value
\be\label{GaussDII1} W_{2n}^{(1)}(t) = \frac{(2n)!}{2n(n+1)n!^{2}}\,\cvp\quad W_{2n+1}^{(1)}(t) = 0\, .\ee
\end{proposition}

\begin{proof} Let us analyse the possible graphs contributing to $W_{L}^{(1)}(t)$. We work in the direct representation. We have $V$ vertices of degree $p_{0}$, one marked vertex of degree $L$, $E$ edges and $F$ faces. It is useful to decompose $F=f+f_{\bullet} + f'_{\bullet}$, where $f$ is the number of faces that do not go through the marked vertex, $f_{\bullet}$ the number of faces that go only through the marked vertex and $f'_{\bullet}$ the number of faces that go through the marked vertex and through other vertices as well. We also denote by $f_{\bullet, k}$ the number of faces of length $k$ that go only through the marked vertex and by $f'_{\bullet, k}$ the number of faces that go through the marked vertex $k$ times and through other vertices as well. Note that the faces that do not go through the marked vertex must all have the same length $q_{0}$ and that the faces that go through the marked vertex $k$ times and also through other vertices must have length $k+q_{0}$. Standard arguments yield the relations
\begin{align}\label{brelp1}& 2 E = p_{0} V + L = q_{0}f + \sum_{k\geq 1 }kf_{\bullet,k}+\sum_{k\geq 1 }(q_{0}+k) f'_{\bullet,k}\\
& \label{brelp2}\sum_{k\geq 1} k(f_{\bullet,k}+f'_{\bullet,k}) = L\\
& \label{brelp3} F-E+V = 1\quad \text{(Euler)}\, .
\end{align}
Combining \eqref{brelp1} and \eqref{brelp2} yields
\be\label{brelp4} q_{0}(F-f_{\bullet}) = p_{0} V\, .\ee
We then get, using \eqref{brelp3},
\be\label{brelp5} \frac{4-(p_{0}-2)(q_{0}-2)}{2q_{0}} V + f_{\bullet} = \frac{L}{2}+1\, .\ee
But the prefactor in front of $V$ is proportional to the curvature \eqref{Rpqform} and thus vanishes for the allowed values of $p_{0}$ and $q_{0}$. This yields
\be\label{fbulletform} f_{\bullet} = \frac{L}{2} + 1\, .\ee
This formula proves that $L$ is even, $L=2n$, and that the number of faces that go through the marked vertex is $n+1$. Prop.\ \eqref{GaussDII1} then follows from the following elementary lemma.
\begin{lemma} Assume that a connected planar graph has a vertex of degree $2n$ such that there exists $n+1$ distinct faces that go through this vertex and only through this vertex. Then the planar graph has only this vertex.
\end{lemma}
The key to this result is to realize that there is always at least one face that goes only through the marked vertex and that has length one. Indeed, from \eqref{brelp2}, we get
\be\label{indeqbep} 2n \geq \sum_{k\geq 1}kf_{\bullet,k}\geq f_{\bullet,1}+2\sum_{k\geq 2}f_{\bullet, k} = -f_{\bullet,1}+2f_{\bullet}=-f_{\bullet,1}+2n+2\ee
and thus $f_{\bullet,1}\geq 2$. One then proves the lemma straightforwardly from a recursion on $n$, by using a simple operation of deletion of one of the faces of length one.
\end{proof}

\begin{proposition} $W_{L}^{(2)}(t)=0$.
\end{proposition}

\begin{proof} Working again in the direct representation, one observes that the faces of the graphs potentially contributing to $W_{L}^{(2)}$ all have length $q_{0}$ and that all the vertices have degree $p_{0}$ except one that has degree $L$. This marked vertex introduces a potential curvature defect on the otherwise flat discretized surface, with local Ricci scalar $R_{q_{0},L}$ given by \eqref{Rpqform}. The Gauss-Bonnet formula then yields $(2-q_{0}) L = 2q_{0}$ or equivalently $L=-p_{0}$, a condition that cannot be satisfied.
\end{proof}

\subsection{\label{VectorSec} JT gravity matrix models}

\subsubsection{Vector degrees of freedom}

As explained at the end of Section \ref{MMintSec}, in ordinary Liouville gravity, boundaries may be introduced either by using loop operators or by adding vector variables. The no-go theorems of the previous section show that the usual loop operators are not adequate for JT gravity. The simplest and most pedagogical way to proceed is then to introduce vector degrees of freedom. This will also produce automatically the correct loop operators.

It is convenient to work in the direct formulation, in which the Feynman graphs of the model correspond to the discretized surfaces.\footnote{Strict consistency with the notations used so far would require to put tildes on top of our matrix variable $X$, but this would be unnecessarily cumbersome.} The degrees of freedom are a $N\times N$ Hermitian matrix $X^{i}_{\ j}$ and a complex $N$-vector $(Q^{i},Q^{*}_{i})$. We introduce a polynomial
\be\label{Ppoldef} P(x) = \sum_{s\geq 2} g_{s}x^{s-2}\ee
and define the free energy as
\be\label{PFnewMM} e^{F(A,\tilde A,\{g_{q}\})} = \int\!\d X\d Q\d Q^{*}\, e^{- S( X, Q, Q^{*};A,\tilde A)}\, ,\ee
for an action
\be\label{actsnewMM} S( X, Q, Q^{*};A,\tilde A,\{g_{s}\}) = 
N\tr\Bigl(\frac{1}{2} X^{2} - \sum_{q}\frac{\tilde t_{q}}{q}(\tilde A X)^{q}\Bigr) + N Q^{*}\bigl(\mathbb I - P(\tilde A X)\tilde A\bigr) Q\, .\ee
The external matrices $A$ and $\tilde A$ are as in Eqs.\ \eqref{tildetAdef} and \eqref{ttildeAdef}, with the couplings \eqref{quadrangchoice}.

The matrix propagators form the bulk edges and the vector propagators form the boundary edges. The bulk vertices correspond to the interactions $\tr (\tilde A X)^{q}$ and the boundary vertices correspond to $ Q^{*} (\tilde A X)^{s-2}\tilde A Q$. The background matrix $\tilde A$ is inserted precisely in such a way that a face of length $p$ is weighted by the factor $\smash{t_{p}=\frac{1}{N}\tr\tilde A^{p}}$, both in the case of ``bulk'' faces, whose edges are all made of matrix propagators, and in the case of faces touching the boundary, that have one or more edges made of vector propagators. 

For the choices \eqref{quadrangchoice}, the graphs correspond to discretized surfaces for which there is no curvature at the bulk vertices and for which, according to \eqref{Kpqform}, the boundary has extrinsic curvature $k_{p_{0},q}$ at a vertex $Q^{*}(\tilde A X)^{q-2}\tilde A Q$. This is precisely what we need to build the distorted disk contributing to the JT gravity theory.

\begin{figure}
\centerline{\includegraphics[width=6in]{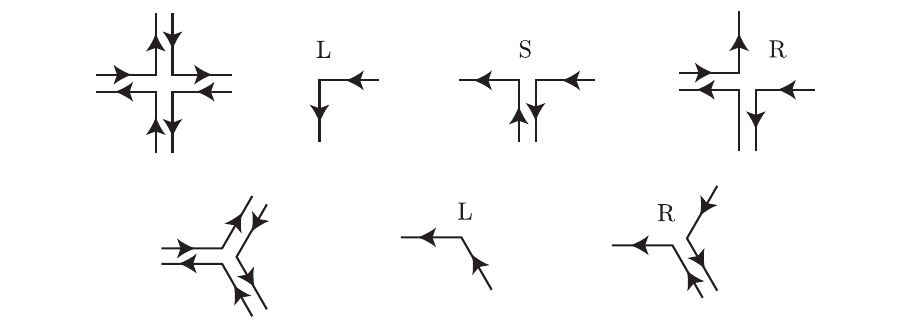}}
\caption{\label{FeynmanMMFig}Feynman rules for the $(p_{0},q_{0})=(4,4)$ and $(p_{0},q_{0})=(6,3)$ models \eqref{actMM4} and \eqref{actMM6}. Double lines correspond to the matrix degrees of freedom $X$ and simple lines, that form the boundary, to the vector degrees of freedom $(Q,Q^{*})$. The orientation is such that the faces that form the discretized surface are to the left. The vertices are drawn isometrically. For each model, there is a unique bulk vertex (to the left). The $(4,4)$ model has three boundary vertices L, S and R (upper right). The $(6,3)$ model has two boundary vertices L and R (lower right). At leading large $N$ order, the associated Feynman graphs are distorted disks, as in Fig.\ \ref{disdiskbasicFig}.}
\end{figure}

Explicitly, if we use quadrangulations, the generalized disk partition function $\mathscr W$ defined in \eqref{Zdefgeneralized}, in the case where the variable $u$ counting the multiplicity is set to one, is given by the leading large $N$ term in the large $N$ expansion of the free energy $F$,
\be\label{Fexp1} F = \frac{1}{4}N\mathscr W(t,g_{\text L},g_{\text S},g_{\text R},1) + O(N^{0})\, ,\ee
for the model with action
\be\label{actMM4} S = N\tr\Bigl(\frac{1}{2}  X^{2} - \frac{1}{4}(\tilde A X)^{4}\Bigr) + N Q^{*}\Bigl(\mathbb I -  \bigl(g_{\text L} + g_{\text S}\tilde A X + g_{\text R}(\tilde A X)^{2}\bigr)\tilde A\Bigr)  Q\ee
and background matrix $\tilde A$ satisfying
\be\label{tildeAex1} \tr\tilde A^{p} = Nt\, \delta_{p,4}\, .\ee
The factor $1/4$ in \eqref{Fexp1} comes from the fact that, by convention, we count configurations modulo lattice translations but not modulo lattice rotations, as explained in Sections \ref{Genfun1Sec} and \ref{symfactSec}. See also the next subsection for a discussion of the symmetry factors.

If we use a honeycomb lattice instead, $p_{0}=6$, we have only two types of boundary vertices, Left or Right, and Eq.\ \eqref{k4qform} is replaced by
\be\label{k6qform} k_{6,q} = 
\begin{cases}
\pi/3 & \text{if $q=2$ (left turn)},\\
-\pi/3  &\text{if $q=3$ (right turn)}.
\end{cases}
\ee
The factor $1/4$ in \eqref{Fexp1} is then replaced by a factor $1/6$ and the action is 
\be\label{actMM6} S = N\tr\Bigl(\frac{1}{2}  X^{2} - \frac{1}{3}(\tilde A X)^{3}\Bigr) + N Q^{*}\Bigl(\mathbb I -  \bigl(g_{\text L} + g_{\text R}\tilde A X \bigr)\tilde A\Bigr) Q\ee
with a background matrix $\tilde A$ satisfying
\be\label{tildeAex2} \tr\tilde A^{p} = Nt\, \delta_{p,6}\, .\ee
One may similarly consider the case of triangulations, $p_{0}=3$. The Feynman rules are depicted in Fig.\ \ref{FeynmanMMFig}.

Finally, note that considering the case of a general polynomial $P$ amounts to allowing branch points. These are clearly interesting models to consider as well, with potentially a variety of critical points and associated continuum limits. As briefly discussed in Section \ref{bdviewSec}, these models correspond to a generalization for which the boundaries are not standard self-overlapping curves but more general ``interior boundaries.''

\subsubsection{\label{SfactMMSec}Symmetry factors}

Symmetry factors associated to a Feynman graph correspond to the order of the automorphism group of the graph. We would like to show that this matches precisely with the counting rules explained in Sections \ref{Genfun1Sec} and \ref{bdcodeSec}. 

Let us deal with the case of quadrangulations and assume we have a graph with $v_{\text L}$, $v_{\text S}$ and $v_{\text R}$ Left, Straight and Right boundary vertices respectively and $V_{\text B}$ bulk vertices. The contribution of such a graph to the free energy $F$ is computed by adding up all the possible Wick contractions of the monomial
\be\label{Wicksymmono} \frac{g_{\text L}^{v_{\text L}}g_{\text S}^{v_{\text S}}g_{\text R}^{v_{\text R}}}{v_{\text L}!v_{\text S}!v_{\text R}!}\frac{1}{V_{\text B}!4^{V_{\text B}}}
\bigl(Q^{*}\tilde A Q\bigr)^{v_{\text L}}\bigl(Q^{*}\tilde A X \tilde A Q\bigr)^{v_{\text S}}
\bigl(Q^{*}\tilde A X \tilde A X\tilde A Q\bigr)^{v_{\text R}}\bigl(\tr (\tilde A X)^{4}\bigr)^{V_{\text B}}\ee
that produce the graph under consideration. The prefactor comes from the expansion of $e^{-S}$, where $S$ is given by \eqref{actMM4}, at the desired order.

Let us first consider the Wick contractions of the vector variables. The structure of the possible contractions is completely determined by the path code of the boundary, as defined in Section \ref{bdcodeSec}. Concretely, imagine that we start from a given Left term $Q^{*}\tilde A Q$ in \eqref{Wicksymmono}. Let's call this term the ``root'' (note that there are always at least four such terms, see Eq.\ \eqref{EulergenflatJTdisk}). This left term must correspond to a particular left vertex in the code. Let us pick one, that we denote by $v$. As explained in \ref{bdcodeSec}, to $v$ is associated a rooted code $B_{v}$. The allowed contractions are clearly completely fixed by the rooted code. At each step along the boundary, we must choose which vertex we contract with. This yields a total of $(v_{\text L}-1)!v_{\text S}!v_{\text R}!$ possibilities. Taking into account the $1/(v_{\text L}!v_{\text S}!v_{\text R}!)$ prefactor in \eqref{Wicksymmono}, we get a contribution $1/v_{\text L}$. According to Section \ref{bdcodeSec}, there are then three possibilities. In the generic case, the marked codes associated with the $v_{\text L}$ left vertices are all distinct. This means that the root term $Q^{*}\tilde A Q$ we started with can be associated with $v_{\text L}$ distinct patterns of contractions all yielding the same unrooted path code. Overall, we get a contribution $\frac{1}{v_{\text L}}\times v_{\text L} = 1$. In another case, the boundary code has a $\mathbb Z_{2}$ symmetry, which means that there are only $v_{\text L}/2$ distinct patterns of contractions and we get the symmetry factor $1/2$. In the last possible case, the boundary code has a $\mathbb Z_{4}$ symmetry, with only $v_{\text L}/4$ distinct patterns of contractions, yielding a symmetry factor $1/4$.

It remains to contract the matrix variables. Geometrically speaking, the procedure amounts to picking the vertices $\tr (\tilde A X)^{4}$ one by one and attaching them to the pre-existing boundary and bulk vertices, building up the distorted disk. Since the $V_{\text B}$ vertices $\tr (\tilde A X)^{4}$ are identical, we get $V_{\text B}!$ possible contraction patterns that yield precisely the same distorted disk. Moreover, each vertex can be attached in four equivalent ways, since they have a $\mathbb Z_{4}$ rotational symmetry. This yields an additional $4^{V_{\text B}}$ factor. We thus get a $V_{\text B}!4^{V_{\text B}}$ factor that precisely cancels the corresponding prefactor in \eqref{Wicksymmono}.

In conclusion, our matrix-vector model counts the distorted disks, or equivalently the SOP that form the boundaries, using precisely the same rules as those explained in Sections \ref{Genfun1Sec} and \ref{bdcodeSec}.

\subsubsection{\label{loopJTMMSec}JT gravity loop operators and arbitrary topologies}

The construction presented above works perfectly well for the disk topology. What about more general topologies?

The large $N$ expansion of the free energy defined in \eqref{PFnewMM} is of the form
\be\label{Fnewgenexp} F = \sum_{k\geq 1} N^{2-k} F_{k}\, ,\ee
where $F_{k}$ gets contributions from Feynman graphs of genus $h$ and $b$ boundaries, with 
\be\label{khbrel} k = 2 h +b\, .\ee
The leading term selects unambiguously the disk topology, $(h,b)=(0,1)$, with, according to Eq.\ \eqref{Fexp1}, $F_{1} = \frac{1}{4}\mathscr W$. However, the higher order terms pick contributions from several topologies. For example, at $k=2$, we get contributions from the torus, $(h,b) = (1,0)$, and from the annulus, $(h,b) = (0,2)$; at $k=3$, from $(h,b)=(0,3)$ and $(h,b)=(1,1)$; etc. In the case $k=2$, we can subtract the known torus contribution, Eq.\ \eqref{DiFIgenusone}, to single out the annulus contribution. But even in this simple case, we do not get the most general annulus partition function in this way, because the couplings $g_{\text L}$, $g_{\text S}$ and $g_{\text R}$ counts the vertices on both boundaries without distinction. This implies, for instance, that we cannot extract the partition function for which the length of the two boundaries are fixed independently of each other. When $k\geq 3$, the situation is even less favourable, because we contributions from various topologies cannot be disentangled.

This problem is easily solved in the case of Liouville theory, by using loop operators. For instance, generalizing Eq.\ \eqref{W2ndLMM}, the number $\mathsf W^{(h)}_{k_{1},\ldots,k_{b},p}$ of quadrangulations of genus $h$ surfaces with $b$ boundaries of lengths $k_{1},\ldots,k_{b}$ and area $p$ is obtained by computing the connected correlator of macroscopic loop operators,
\be\label{Wbgene} \Bigl\langle\frac{1}{k_{1}}\tr X^{k_{1}}\cdots \frac{1}{k_{b}}\tr X^{k_{b}}\Bigr\rangle_{\text c} = \sum_{h\geq 0}N^{2-2h-b}
\mathsf W^{(h)}_{k_{1},\dots,k_{b}}(t) = \sum_{h\geq 0}N^{2-2h-b}\sum_{p\geq 0} \mathsf W^{(h)}_{k_{1},\ldots,k_{b},p}t^{p}\ee
in the one-matrix matrix model with action $\smash{N\tr(\frac{1}{2}X^{2} - \frac{t}{4}X^{4})}$. 

For JT gravity, we have explained in \ref{nogoSec} that the naive ansatz for the loop operators do not work. But we now have all the ingredients to construct the correct loop operators.

To motivate the construction, let us start from the matrix-vector integral \eqref{PFnewMM} and let us integrate out the vector variables explicitly. For concreteness, we focus on the case of quadrangulations, with action \eqref{actMM4} and $\tr\tilde A^{p} = N t\delta_{p,4}$ and $\tr A^{p} = N\delta_{p,4}$ (the case of hexagons or triangles is very similar). We get
\begin{multline}\label{PFn2} e^{F} = \int\!\d X\, \frac{1}{\det (\mathbb I - P(\tilde A X)\tilde A)} e^{-N\tr[\frac{1}{2} X^{2}-\frac{1}{4}(\tilde A X)^{4}]}
\\= \int\!\d X\, e^{-N\tr[\frac{1}{2} X^{2}-\frac{1}{4}(\tilde A X)^{4}] + \sum_{k\geq 1}\frac{1}{k}\tr \phi(X)^{k}}
\end{multline}
where $\phi(X)$ is the Hermitian matrix defined by
\be\label{phidef} \phi(X) = P(\tilde A X)\tilde A = g_{\text L}\tilde A + g_{\text S}\tilde A X\tilde A + g_{\text R}\tilde A X\tilde A X\tilde A\, .\ee
At leading large $N$ order, we can write
\begin{multline}\label{PFn3} e^{NF_{1}+ O(N^{0})} = \frac{\int\!\d X\, e^{-N\tr[\frac{1}{2} X^{2}-\frac{1}{4}(\tilde A X)^{4}] + \sum_{k\geq 1}\frac{1}{k}\tr \phi(X)^{k}}}{\int\!\d X\, e^{-N\tr[\frac{1}{2} X^{2}-\frac{1}{4}(\tilde A X)^{4}]}}
\\ = \Bigl\langle e^{\sum_{k\geq 1}\frac{1}{k}\tr \phi^{k}}\Bigr\rangle
= e^{\sum_{k\geq 1}\frac{1}{k}\langle\tr \phi(X)^{k}\rangle + O(N^{0})}\, ,
\end{multline}
where in the last equality we have used the standard large $N$ factorization of correlation functions, $\langle \tr\phi^{k_{1}}\tr\phi^{k_{2}}\rangle = \langle \tr\phi^{k_{1}}\rangle\langle \tr\phi^{k_{2}}\rangle + O(1/N)$. This yields
\be\label{loopJT1} F_{1} = \frac{1}{4}\mathscr W = \sum_{k\geq 1}\frac{1}{k}\bigl\langle\frac{1}{N}\tr\phi(X)^{k}\bigr\rangle\ee
at leading large $N$ order. The disk generating function at fixed boundary length is thus
\be\label{Wtkmm} \frac{1}{4} \mathscr W_{k} = \frac{1}{k}\bigl\langle\frac{1}{N}\tr\phi(X)^{k}\bigr\rangle\quad \text{when}\ N\rightarrow\infty\, .\ee
This is the exact analogue of \eqref{W2ndLMM}, with a loop operator for JT gravity identified with $\frac{1}{k}\tr\phi(X)^{k}$. 

More generally, we can write the analogue of Eq.\ \eqref{Wbgene} for JT gravity,
\be\label{WbgenJT} \Bigl\langle\frac{1}{k_{1}}\tr \phi_{1}(X)^{k_{1}}\cdots \frac{1}{k_{b}}\tr \phi_{b}(X)^{k_{b}}\Bigr\rangle_{\text c} = \sum_{h\geq 0}N^{2-2h-b}
\mathscr W^{(h)}_{k_{1},\dots,k_{b}}\, , \ee
where $\mathscr W^{(h)}_{k_{1},\dots,k_{b}}$ is the generating function for the number of connected flat discretized surfaces of genus $h$ with $b$ boundaries. The operators $\phi_{\alpha}(X)$, $1\leq \alpha\leq b$, are given by
\be\label{phidefa} \phi_{\alpha}(X) = g_{\text L\alpha}\tilde A + g_{\text S\alpha}\tilde A X\tilde A + g_{\text R\alpha}\tilde A X\tilde A X\tilde A\, .\ee
If we do not want to keep track of the number of Left, Straight and Right boundary vertices, but only on the lengths of the boundary components, Eqs.\ \eqref{WbgenJT} and \eqref{phidefa} simplify to 
\be\label{WbgenJT2} \Bigl\langle\frac{1}{k_{1}}\tr \varphi(X)^{k_{1}}\cdots \frac{1}{k_{b}}\tr \varphi(X)^{k_{b}}\Bigr\rangle_{\text c} = \sum_{h\geq 0}N^{2-2h-b}
 W^{(h)}_{k_{1},\dots,k_{b}}(t)\, , \ee
for an operator
\be\label{phidef2} \varphi(X) = \tilde A + \tilde A X\tilde A + \tilde A X\tilde A X\tilde A\, .\ee
In the case of the honeycomb lattice, we work with $\tr\tilde A^{p} = Nt\delta_{p,6}$, $\tr A^{p} = N\delta_{p,3}$ and a simpler loop operator $\varphi = \tilde A + \tilde A X\tilde A$. 


\subsubsection{A matrix formula for the multiplicity}

An interesting application of the loop operator formalism is to derive a matrix integral formula for the multiplicity index $\mu_{\gamma}$ associated with an arbitrary polygon $\gamma$. Recall that $\mu_{\gamma}$ is a positive integer that counts the number of distinct distorted disks bounded by $\gamma$, see Section \ref{MilnorSec}. In particular, $\mu_{\gamma}$ is non-zero if and only of $\gamma$ is a self-overlapping polygon.

As usual, for concreteness, we work with the square lattice, the generalization to the hexagonal and triangular lattices being straightforward. We introduce 
\be\label{loopopmult} \varphi_{\text L} = \tilde A\, ,\quad \varphi_{\text S} = \tilde A X\tilde A\, ,\quad \varphi_{\text R} = \tilde A X\tilde A X\tilde A\, .\ee
Let $\gamma$ be an an arbitrary polygon, characterized by a certain path code, as defined in Section \ref{bdcodeSec}. Call $\gamma_{v}$ the polygon $\gamma$ rooted at a certain vertex $v$, with associated rooted code $B_{v}$, as in \ref{symfactSec}. To $\gamma_{v}$, we associate the matrix $\varphi_{\gamma_{v}}$, which is obtained by taking the product of $\varphi_{\text L}$, $\varphi_{\text S}$ and $\varphi_{\text R}$, ordered according to the rooted path code. For instance, for $B_{v}=\text{LSLSLLRL}$, we have $\varphi_{\gamma_{v}} = \varphi_{\text L}\varphi_{\text S}\varphi_{\text L}\varphi_{\text S}\varphi_{\text L}\varphi_{\text L}\varphi_{\text R}\varphi_{\text L}$. The operator $\tr\varphi_{\gamma_{v}}$ does not depend on the marked vertex, thank's to the invariance of the trace under cyclic permutations.
\begin{definition} The loop operator associated with a polygon $\gamma$ is defined to be $\tr\varphi_{\gamma_{v}}$, where $v$ is an arbitrary marked vertex on $\gamma$. Since the loop operator does not depend on the choice of $v$, we denote it by $\tr\varphi_{\gamma}$.
\end{definition}
\begin{theorem} We consider the large $N$ Hermitian matrix model with action 
\be\label{mmactmult} S = N\tr\Bigl(\frac{1}{2}X^{2} - \frac{1}{4}(\tilde A X)^{4}\Bigr)\, ,\ee
where the external Hermitian matrix $\tilde A$ is such that $\tr\tilde A^{p} = Nt\delta_{p,4}$. The expectation value of the loop operator associated with a polygon $\gamma$, satisfying the constraint $v_{\text L}-v_{\text R}=4$, Eq.\ \eqref{EulergenflatJTdisk}, in this matrix model, is given by 
\be\label{multasvev}  \langle\tr\varphi_{\gamma}\rangle = N \mu_{\gamma} t^{p_{\gamma}}\, ,\ee
where $\mu_{\gamma}$ and $p_{\gamma}$ are the multiplicity of $\gamma$ and the area of the distorted disks it bounds, respectively (recall from Section \ref{overlapSec}, Corrolary \ref{Areacorollary}, that the area depends only on $\gamma$, even in the cases where $\mu_{\gamma}>1$).
\end{theorem}
\begin{proof} The proof is based on a slight refinement of the discussion made in Section \ref{SfactMMSec}. A boundary code $\gamma$ being given, it was shown that the number of Wick contractions of the vector variables in \eqref{Wicksymmono} associated with $\gamma$ is $v_{\text L}!v_{\text S}!v_{\text R}!/\sigma_{\gamma}$, where $\sigma_{\gamma}$ is the symmetry factor associated with $\gamma$ (one for a generic polygon, two or four for a $\mathbb Z_{2}$ or $\mathbb Z_{4}$ symmetric polygon). The new piece of information we need is that these contractions precisely yield the loop operator $\tr\varphi_{\gamma}(X)$. This is straightforward to check. The subsequent contractions of the matrix variables and sum over the number $V_{\text B}$ of $\tr (\tilde AX)^{4}$ insertions amounts, on the one hand, to computing $\frac{1}{\sigma_{\gamma}}\langle\tr\varphi_{\gamma}\rangle$, and, on the other hand, to reconstructing all the distorted disks bounded by $\gamma$, yielding the contribution $Nt^{p_{\gamma}}\mu_{\gamma}/\sigma_{\gamma}$ to the free energy $F$. We conclude.
\end{proof}
Eq.\ \eqref{multasvev} is exact at any order in the large $N$ expansion; it thus implies that the expectation value of the loop operator gets contributions only from the disk topology. If we waive the constraint $v_{\text L}-v_{\text R}=4$, Eq.\ \eqref{EulergenflatJT} predicts that $\langle\tr\varphi_{\gamma}\rangle = 0$ if $v_{\text L}-v_{\text R}$ is not of the form $4(1-2h)$ for a positive integer $h$, and gets contribution only at genus $h$ if $v_{\text L}-v_{\text R}=4(1-2h)$. Similarly, we could compute connected correlators $\langle\tr\varphi_{\gamma_{1}}\cdots\tr\varphi_{\gamma_{b}}\rangle_{\text c}$ in the matrix model to decide whether the curves $\gamma_{i}$ bound a flat surface.

Investigating these questions and more generally finding a solution to JT gravity, on the disk and possibily on higher topologies as well, using the matrix model, is clearly a challenging and fascinating problem for the future.

\section{\label{nonzeroRSec}The cases of non-zero curvature}

We have already briefly discussed the case of non-zero curvature in Section \ref{posnegcurvesSec}. We can now go into more detail and attempt to paint a general qualitative picture of the properties expected for negative and positive curvatures.

\subsection{Qualitative features}

There are two important properties that we would like to emphasize.

\paragraph{The microscopic structure of the models cannot depend on the curvature of space-time.} 

If $L$ is the curvature length scale, defined by $R=\pm 2/L^{2}$, then both the positive and the negative curvature models will be well-approximated by the flat model on distance scales much less than $L$. When $\La\geq 0$, this will be the case as soon as the quantum length of the boundary $\beta_{\text q}$ is much less than $L^{2}$.\footnote{Recall that $\beta_{\text q}$ has the dimension of a length square, because the exponent $\nu=1/2$ in pure JT gravity, see Eq.\ \eqref{critexponentsJT} and \cite{ferrari}.} In the negative curvature model, if $\La$ is negative, one must also impose that $\beta_{\text q}|\La|$ is not too large. We shall argue that the positive curvature model is actually well-defined only for $\La\geq 0$. In terms of the partition functions, which were formally defined in Eq.\ \eqref{qJT1}, one must have
\be\label{flatlimit} \lim_{L\rightarrow +\infty} W^{\pm}(\beta_{\text q},\La,L) = W^{0}(\beta_{\text q},\La)\, ,\ee
the limit being taken for $\La\geq 0$ in positive curvature or for either $\La \geq 0$ or $\La<0$ and $\beta_{\text q}|\La|$ not too large in negative curvature.

The fact that all three models have the same short distance behaviour is equivalent to saying that they have the same UV-completion. 

\paragraph{The models in negative, zero and positive curvature differ in very important ways on large distance scales.}

These differences will be made very precise below, in particular by offering conjectures on the large area asymptotics of the area probability densities. At a qualitative level, we can understand the very different behaviour of the three models at long distances by referring to the existence, or non-existence, of isoperimetric inequalities. In negative curvature, for a disk having a smooth boundary of length $\ell$, the area is bounded above by the area of the round disk,
\be\label{isoneg} A\leq  2\pi L^{2}\Biggl[-1+\sqrt{1+\frac{\ell^{2}}{4\pi^{2}L^{2}}}\Biggr]\quad\text{when $R=-2/L^{2}$.}\ee
For large configurations, the bound grows linearly with the boundary length, $A\leq L\ell$. In zero curvature, the upper bound is still given by the round disk,
\be\label{isozero} A\leq  \frac{\ell^{2}}{4\pi}\quad\text{when $R=0$.}\ee
It now grows quadratically with the boundary length. In positive curvature, there is no bound. This will be illustrated explicitly below, in Section \ref{poscurvlastSec}. 

If the boundaries were smooth in the quantum theory, the above bounds would imply that the area probability densities in negative and zero curvature would have compact support. Instead, the boundary is a fractal, with a fixed quantum boundary length $\beta_{\text q}$ but an infinite smooth length $\ell$. There is thus no bound on the possible areas in the quantum theory, and the probability densities can have tails that extend to infinity, even in negative or zero curvature. It is natural to expect that the densities will decay much faster in negative curvature than in zero curvature, since the smooth bound as a function of the boundary length is stricter in negative curvature, and much faster in zero curvature than in positive curvature, since there is no smooth bound in the latter case. The decay laws of the area probability densities can be interpreted as quantum versions of the classical isoperimetric inequalities. Conjectures about their explicit form are given below, see \ref{areazeroconj1}, \ref{areazeroconj2} and \ref{areanegconj}, \ref{areaposconj}. 

The rapid decay in negative curvature implies that large areas are strongly suppressed; this ensures the existence of the theory for any value of the cosmological constant, including the regime of large negative $\La$ where the effective Schwarzian description is emerging. In zero curvature, we conjecture an exponential decay, implying that the pure gravity model exists for a range of cosmological constant $\La\geq\La_{\text c}$, where $\La_{\text c}$ is finite and negative. In positive curvature, we conjecture a power-law decay. Large area configurations are then ubiquitous and the theory exists only for positive cosmological constants.

\subsection{\label{negcurvlastSec}Negative curvature}

\subsubsection{Area distribution law} 

\begin{conjecture}\label{areanegconj} The area distribution of JT gravity in negative curvature coupled to $c \leq 0$ conformal matter, with critical exponent $\nu$ given by Eq.\ \eqref{critexponentsJT}, has the asymptotic behaviour
\be\label{rhodecaynegc} \ln\rho^{-}_{\beta_{\text q}}(A)\underset{A\rightarrow +\infty}{\sim} - \kappa_{\nu}^{-} \biggl(\frac{A}{\beta_{\text q}^{\nu} L}\biggr)^{\chi^{-}_{\nu}}\, ,\ee
where $\kappa_{\nu}^{-}$ is a strictly positive numerical constant and $\chi_{\nu}^{-}$ a critical exponent, which is a strictly increasing function of $\nu$ satisfying $\lim_{\nu\rightarrow \frac{1}{2}^{+}}\chi_{\nu}^{-} =2$ and $\lim_{\nu\rightarrow 1^{-}}\chi_{\nu}^{-} = +\infty$.
\end{conjecture}
The fact that $\chi_{\nu}^{-}\geq 2>1$ ensures the existence of the theory for all $\La$, including in the case of the pure gravity theory, for which $\chi_{\nu}^{-}$ is conjectured to be 2. The limit when $\nu\rightarrow 1^{-}$ follows from the semiclassical analysis of \cite{Loopcalc}. The limit $\nu\rightarrow \frac{1}{2}^{+}$ is motivated in part by a similar behaviour found for purely random Brownian paths in hyperbolic space \cite{RDFollowup1}, see also the discussion in \ref{areaSchSec} below.

In a very similar way to the discussion in Section \ref{areaSOPcomnjSec}, we observe that Eq.\ \eqref{rhodecaynegc} is equivalent to the asymptotic behaviour
\be\label{Wminasympconj} \ln W^{-}(\beta_{\text q},\La,L) \underset{\La\rightarrow -\infty}{\sim} \frac{\chi_{\nu}^{-}-1}{\chi_{\nu}^{-}} \bigl(\chi_{\nu}^{-}\kappa_{\nu}^{-}\bigr)^{-\frac{1}{\chi_{\nu}^{-}-1}}
\biggl(\frac{|\La|\beta_{\text q}^{\nu}L}{16\pi}\biggr)^{\frac{\chi_{\nu}^{-}}{\chi_{\nu}^{-}-1}}\, ,\ee
for the JT gravity partition function in negative curvature or, similarly, to the asymptotic behaviour
\be\label{minusexpareaconj} \langle A\rangle \underset{\La\rightarrow -\infty}{\sim} 
\biggl(\frac{|\La|\beta_{\text q}^{\nu}L}{16\pi\chi_{\nu}^{-}\kappa_{\nu}^{-}}\biggr)^{\frac{1}{\chi_{\nu}^{-}-1}} \beta_{\text q}^{\nu} L\ee
for the expected area.

These formulas as consistent with the $c\rightarrow -\infty$ analysis presented in \cite{Loopcalc}, with
\be\label{chiminas} \chi_{\nu}^{-}\underset{c\rightarrow -\infty}{\sim} \frac{|c|}{6} 
\underset{\nu\rightarrow 1^{-}}{\sim}\frac{1}{1-\nu}\,\cdotp\ee
They are also consistent with the discussion in Section \ref{areaSchSec} below. We conjecture that the asymptotic law \eqref{chiminas}, as a function of $\nu$, is exact.
\begin{conjecture} The exponent $\chi_{\nu}^{-}$ governing the asymptotic behaviour of the area distribution, Eq.\ \eqref{rhodecaynegc}, is given by 
\be\label{chinyconjf2} \chi_{\nu}^{-} = \frac{1}{1-\nu}\,\cdotp\ee
\end{conjecture}

\subsubsection{\label{areaSchSec}Area distribution, ballistic behaviour and the Schwarzian limit}

As already mentioned repeatedly, e.g.\ in Section \ref{repaSec}, one expects that the negative curvature JT gravity theory has an effective description in the $\La\rightarrow -\infty$ limit, valid on scales much larger than the curvature length scale $L$, in terms of near-hyperbolic geometries described by the reparameterization ansatz and the Schwarzian action. These effective geometries have a smooth boundary, with a fixed length $\ell$. This length is an effective macroscopic parameter which can be expressed, in principle, to the microscopic parameters $\beta_{\text q}$, $\La$ and $L$. We now have the necessary ingredients to derive this relation.

But first things first, in order to get an intuitive understanding of why and how the Schwarzian description may emerge, it is useful to recall the behaviour of Brownian motion in hyperbolic space, see e.g.\ \cite{hyperbrownian}. If we let a Brownian particle move for a diffusion (``quantum'') time $\beta_{\text q}$, then the average geodesic distance between the particle and its initial position will start to grow as the square root of $\beta_{\text q}$, as it does in flat space. Precisely, with the conventions used in Section \ref{RPexSec} to define the Brownian bridge, this average distance is $\smash{\sqrt{\pi\beta_{\text q}/2}}$. The square root behaviour is directly related to the fact that the exponent $\nu$ is $1/2$ for Brownian motion. In flat space, this behaviour remains valid for all $\beta_{\text q}$. But in hyperbolic space, the situation is drastically different. The square root growth is valid as long as $\beta_{\text q}\ll L^{2}$, but when $\beta_{\text q}\gg L^{2}$ one finds a completely different ``ballistic'' regime for which the average geodesic distance grows linearly, as $\beta_{\text q}/(2 L)$. The linear growth with $\beta_{\text q}$ indicates that, at large distances, an effective description emerges, with an effective exponent $\nu_{\text{eff}}=1$ corresponding to a smooth trajectory.

We claim that the transition from the flat space description to the Schwarzian description in negative curvature JT gravity is of a similar nature.\footnote{And a similar mechanism is undoubtedly responsible for the fact that the Kitaev-Suh model for JT gravity, which is based on Brownian motion in hyperbolic space, also has a Schwarzian description at long distances \cite{KitaevSuh,RDFollowup1}.} The long distance, large geometry limit is obtained when the cosmological constant $\La\rightarrow -\infty$. The emergent macroscopic length parameter $\ell$ can be naturally defined by the asymptotic expected area, consistently with the leading term in Eq.\ \eqref{Arearepansatz},
\be\label{schareaminas} \langle A\rangle \underset{\La\rightarrow -\infty}{\sim} 
\ell L\, .\ee
The fact that the boundary will be effectively smooth means that $\ell$ will be proportional to $\beta_{\text q}$, exactly as for the ballistic regime of Brownian motion. Assuming a simple power-law dependence in $\La$ and $L$, dimensional analysis implies that
\be\label{ellans1} \ell = a \beta_{\text q}|\La|^{\epsilon}L^{1-1/\nu+2\epsilon}\, ,\ee
where $a$ is a numerical constant of proportionality and $\epsilon$ an exponent.
The results of \cite{Loopcalc} yield the $c\rightarrow -\infty$ behaviour
\be\label{eps12sc}  \epsilon =\frac{6}{|c|} + O\bigl(|c|^{-2}\bigr) \, .\ee

The linearity of $\ell$ with $\beta_{\text q}$ also has the following nice consequence. At the microscopic level, there is a boundary cosmological constant counterterm that is proportional to the quantum boundary length $\beta_{\text q}$.\footnote{At a fundamental level, this comes from the exponential growth of the number of configurations with the discretized boundary length $n$ and the fact that $\beta_{\text q}$ is proportional to $n$, see Section \ref{pathcontSec} and in particular Eq.\ \eqref{contpathdef}.} This means that the logarithm of the partition function of the theory is defined modulo the addition of terms that are proportional to $\beta_{\text q}$. In the effective Schwarzian description of the theory, such a boundary cosmological counterterm is usually used to cancel the diverging leading $\ell L$ term in the Schwarzian area expansion, Eq.\ \eqref{Arearepansatz}. But a term proportional to the macroscopic length parameter $\ell$ in the effective description is actually an allowed counterterm from the fundamental, microscopic point of view, only if $\ell$ is itself proportional to $\beta_{\text q}$. The ballistic property $\ell\propto\beta_{\text q}$ ensures that this is indeed the case.

Note that Eqs.\ \eqref{schareaminas} and \eqref{ellans1} are consistent with \eqref{minusexpareaconj}, with $\epsilon = 1/(\chi_{\nu}^{-} - 1)$. Using the conjecture \eqref{chinyconjf2}, we get the prediction
\be\label{epsminexpconj}\epsilon = \frac{1}{\nu}-1\, .\ee
In the pure gravity case, $\nu=1/2$, $\epsilon = 1$ and Eq.\ \eqref{ellans1} yields $\ell\propto \beta_{\text q}|\La|L$. In particular, the dimensionless Schwarzian limit parameter $\beta_{\text S}$ defined in Eq.\ \eqref{Schlimit} turns out to be simply proportional to $\beta_{\text q}/L^{2}$, in pure gravity.

\subsubsection{Lattice formulation and the Schwarzian limit}

Up to now, we have (extensively) discussed the lattice formulation of JT gravity only in zero curvature. It is of course natural to look for a similar formulation with non-zero curvature and in particular for negative curvature, which is the case we are interested in in the present section. This turns out to be more subtle and interesting than one might have expected, and sheds an unexpected light on the Schwarzian regime.

A straightforward generalization of the lattice formulation in zero curvature, that uses regular tesselations of the Euclidean plane, is to consider the self-overlapping polygon model on regular tesselations of hyperbolic space, also called ``hyperbolic lattices.'' As in flat space, these tesselations are characterized by the number of edges $p$ of the regular polygons forming the lattice tiles and the degree $q$ of the vertices at which the tiles meet. The integers $p$ and $q$ must satisfy the inequality \eqref{hyperbolictesscons}, already mentioned in Section \ref{posnegcurvesSec}, that follows from the formula \eqref{Rpqform} for the curvature and the constraint of negative curvature.

A crucial difference with the case of flat space is that the regular tessellations of hyperbolic space cannot be scaled down to take a continuum limit. The edge length and polygon area are indeed fixed in terms of the curvature length scale and the values of $p$ and $q$. For instance, in the construction of Section \ref{diskviewSec}, for which the polygons are flat, of edge length $\ell_{0}$, and the curvature is localized at the vertices, the area of a regular $p$-gon is
\be\label{Anot} A_{0} = \frac{p\ell_{0}^{2}}{4\tan(\pi/p)}\ee
and the area per vertices is $q A_{0}/p$, since a given vertex belongs to $q$ different $p$-gons. If we denote by $-2/L^{2}$ the average curvature, the integral of the curvature per vertex is thus $-2 q A_{0}/(p L^{2})$. On the other hand, this integral is also given by $R_{p,q}$, as defined in \eqref{Rpqform}. This yields
\be\label{Anot2} A_{0} = \frac{\pi L^{2}}{q}\bigl((p-2)(q-2)-4\bigr)\, .\ee
Thus the area of the basic building block of the lattice is proportional to $L^{2}$ and a continuum limit in which this area goes to zero does not exist. Of course, the edge length $\ell_{0}$ itself can be made very small when $p$ is very large at fixed $q$ and $L$, Eqs.\ \eqref{Anot} and \eqref{Anot2} implying that $\ell_{0}^{2}\sim 4\pi^{2}\frac{q-2}{q}\frac{1}{p}L^{2}$, but at the same time $A_{0}$ goes to infinity, proportionally to $p$. Thus, this is not suitable for a continuum limit. It is actually easy to check that the polygon area is bounded below by $\frac{1}{7}\pi L^{2}$, a bound saturated for $(p,q)=(3,7)$.

The conclusion is that a lattice model for negative curvature JT gravity based on regular tesselations (or similar tesselations, like the semi-regular tesselations, for which the microscopic length and area are also fixed by $L$) cannot yield a microscopic description of the theory, valid on arbitrarily small length scales, but at best an effective description, valid on length scales much greater than $L$. Obviously, we have a very natural candidate for this effective theory, for which the ``UV'' cut-off is $L$: it is the Schwarzian theory.
\begin{conjecture} The Schwarzian effective description of negative curvature JT gravity is the continuum limit of lattice models of self-overlapping polygons, defined on regular or semi-regular tesselations of hyperbolic space. This continuum limit corresponds to taking the curvature length scale $L\propto\ell_{0}$ to zero and the polygon length $n$ to infinity, with the combination $n L$ held fixed and proportional to the macroscopic length parameter $\ell$. In particular, the critical exponent $\nu$ and the Hausdorff dimension $1/\nu$ are one in this case.
\end{conjecture}
One may actually conjecture that the Schwarzian description is universal and is not limited to the SOP model or JT gravity. Any model of random closed curves in hyperbolic space having a suitable macroscopic limit will likely be described by the Schwarzian on length scales much greater than the curvature length scale. Strong evidence for this in the case of purely random curves was discussed above and given in \cite{KitaevSuh}. Self-avoiding polygons on hyperbolic lattices have also been studied in the literature \cite{SAWhyper}. The ballistic property was proven in \cite{SAWhyperball}, consistently with the idea that a Schwarzian description will also emerge for these models in the limit $\La\rightarrow -\infty$. 

The intuitive reason explaining why ballisticity seems to be generic in hyperbolic space is that the number of lattice vertices accessible after $n$ steps from a given vertex grows exponentially with $n$ for hyperbolic lattices, instead of quadratically for flat space lattices.\footnote{This is the lattice version of the well-known exponential growth of a round disk area $A$ as a function of its geodesic radius $r$, for large $r$, $A \sim \pi L^{2}e^{r/L}$.} A random path starting from a given point has thus many more possibilities to extend far away in hyperbolic space than in flat space.

To obtain a microscopic definition of JT gravity in negative curvature by using a lattice formulation, one thus needs to refine the regular tesselations in one way or another, in order to create a lattice mesh of arbitrarily small length scales. There are many ways to do this, cutting the polygons of the regular tesselations into smaller pieces, see e.g.\ the discussion in Section 5 of the first reference in \cite{hyplatticeapp}. 

\subsubsection{Matrix model}

The matrix model formulation described in Section \ref{mmSec} in flat space can be straightforwardly generalized to the case of negative curvature. One simply chooses the matrices $A$ and $\tilde A$ such that
\be\label{AAtildehyper} \tr\tilde A^{k} = N t \delta_{k,p}\, ,\quad \tr A^{k} = N \delta_{k,q}\ee
for integers $p$ and $q$ satisfying the inequality \ref{hyperbolictesscons}, $(p-2)(q-2)>4$. The coupling to the vector variables is as in Eq.\ \eqref{actsnewMM} and the loop operator is $\phi(X) = P(\tilde A X)\tilde A$, along the lines of Section \ref{loopJTMMSec}. The polynomial $P$ can be naturally chosen to be of degree $q-2$. The matrix model partition function is then the generating function for self-overlapping polygons on the hyperbolic lattice $(p,q)$, the allowed polygons being defined as in Section \ref{bdviewSec}. The fact that $\deg P=q-2$ ensures that the only constraint on successive points $P_{k}$ and $P_{k+1}$ along the polygon is that the edge joining $P_{k}$ and $P_{k+1}$ is not the same as the edge joining $P_{k-1}$ and $P_{k}$. 

As mentioned in the previous subsection, such a model has a short-distance cut-off set by the curvature length scale $L$. Its ``continuum'' limit is expected to yield the Schwarzian theory.

However, it might be the case that the matrix model is more powerful than this and could give access to the full UV-complete theory in curved space. The idea is that there may exist a natural and physically relevant way to analytically continue the matrix model partition function and correlators as a function of the parameters $p$ and $q$. Actually, only the analytic continuation in one of the parameter is needed, e.g.\ in $q$. Let us then set $p=3$, $4$ or $6$ and consider $q$ to be very near $6$, $4$ or $3$, such that the curvature $R_{p,q}$, Eq.\ \eqref{Rpqform}, can be made arbitrarily small. This trick will make it possible to fix the curvature so that it is very small in Planck units. For concreteness, let us work with a square lattice $p=4$ as we have done in most of this paper and let us write
\be\label{qscalingdef} q = 4 (1+\varepsilon)\, ,\ee
where $\varepsilon$ is small and positive. To leading order in $\varepsilon$, Eqs.\ \eqref{Anot} and \eqref{Anot2} yield (for $p=4$)
\be\label{lzeroqsca} \ell_{0}^{2}=2\pi L^{2}\varepsilon\, .\ee
By taking $\varepsilon\rightarrow 0^{+}$, we can therefore obtain, in principle, a genuine continuous limit for the theory of negative curvature. This is summarized by the following (bold) conjecture.
\begin{conjecture}\label{MMneglimconj}
The disk partition function
\be\label{diskpfscaconj} 
W_{n}(t;p=4,q) = \frac{4}{N}\Bigl\langle\frac{1}{N}\tr\varphi(X)^{n}\Bigr\rangle\ee
of the matrix model with action
\be\label{MMconjnegact} S = N\tr\Bigl(\frac{1}{2}X^{2} -\frac{1}{q}\bigl(\tilde A X\bigr)^{q}\Bigr)\, ,\ee
for $\tr\tilde A^{k}= Nt\delta_{k,4}$ and loop operator
\be\label{loopopconjneg} \varphi(X) = \sum_{s=0}^{q-2} \bigl(\tilde A X\bigr)^{s}\tilde A\ee
admits an analytic continuation in the variable $q$, in the vicinity of $q=4$, such that the limit
\be\label{contlimitnegMMconj} n\rightarrow\infty\, ,\ \ell_{0}^{2}=\frac{\pi L^{2}}{2}\bigl(q-4\bigr)\rightarrow 0^{+}\, ,\ t\rightarrow 1\, ,\ n\ell_{0}^{2} = \beta_{\text q}\ \text{and}\ \La=\frac{16\pi}{\ell_{0}^{2}}\bigl(1-t\bigr)\ \text{fixed}\ee
yields the UV-complete disk partition function of pure JT gravity in negative curvature $R=-2/L^{2}$.
\end{conjecture}
Of course, similar conjectures can be made when $p=3$ or $6$.

\subsection{\label{poscurvlastSec}Positive curvature}

\subsubsection{Area distribution law}

\begin{figure}
\centerline{\includegraphics[width=6in]{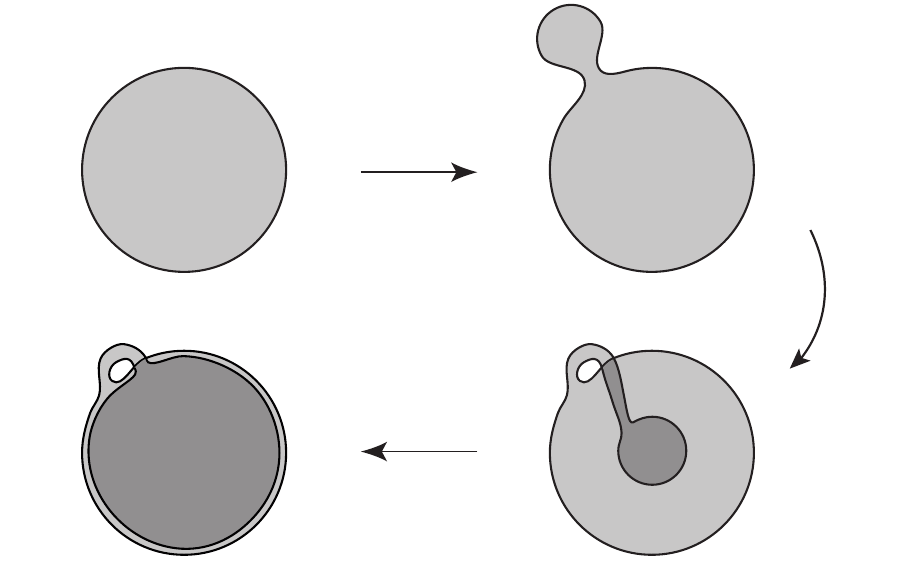}}
\caption{\label{noisofig}One step in the construction of a metric of constant positive curvature on the disk having arbitrarily large area and arbitrarily small boundary length. The distorted disks are drawn on the two-sphere in projective coordinates, with metric given by Eq.\ \eqref{deltaStwo}. In the process depicted in the figure, the boundary length and the area are both multiplied by a number that can be made arbitrarily close to two. By repeating this operation $n$ times and making the figure as large as we want in the last step, we construct a metric of area $\sim 4\pi n$ and arbitrarily small boundary length. See the main text for more details.}
\end{figure}

A crucial qualitative feature of the positive curvature theory is the existence of smooth geometries that can have arbitrarily large areas and arbitrarily small boundary lengths at the same time. A step in the construction of such a geometry is illustrated in Fig.\ \ref{noisofig}. The drawings depict immersions of the disk on the sphere, which is represented in projective coordinates with metric given by Eq.\ \eqref{deltaStwo} (the units are such that the curvature length scale is $L=1$). Starting from a round disk of projective coordinates radius $R$, boundary length $\ell(R)$ and area $A(R)$,
\be\label{ellAposcurv}\ell = \frac{4\pi R}{1+R^{2}}\underset{R\rightarrow\infty}{\sim}\frac{4\pi}{R}\, \cvp\quad A(R) = R\ell(R)\underset{R\rightarrow\infty}{\sim} 4\pi\, ,\ee
we get a distorted disk of approximately twice the original boundary length and area. By repeating the same procedure $n$ times, we obtain a distorted disk which overlaps $n$ times with itself over a region covering a fraction of the entire two-sphere. The associated metric has a boundary length $\ell_{n}$ which is aproximately $n\ell(R)$ and an area $A_{n}$ which is approximately $n A(R)$. By taking $R\gg 1$, the overlapping region covers almost the entire two-sphere, which implies that $A_{n}\sim 4\pi n$, whereas the boundary is concentrated in a tiny region around the south pole, with $\ell_{n}\sim 4\pi n/R$. If we scale $R$ with $n$ in such a way that $n/R\rightarrow 0$ when $n\rightarrow\infty$, for instance by taking $R\sim n^{2}$, we obtain a constant positive curvature metric on the disk that has arbitrarily large area and arbitrarily small boundary length.

This example illustrates a fundamental feature of the positive curvature JT gravity theory, that large geometries will be ubiquitous. This yields the following conjecture on the large area behaviour of the area distribution law for positive curvature JT.

\begin{conjecture}\label{areaposconj} The area distribution of JT gravity in positive curvature coupled to $c \leq 0$ conformal matter, with critical exponent $\nu$ given by Eq.\ \eqref{critexponentsJT}, has the asymptotic behaviour
\be\label{rhodecayposc} \rho^{+}_{\beta_{\text q}}(A)\underset{A\rightarrow +\infty}{\sim}\frac{\kappa_{\nu}^{+}(\beta_{\text q}^{\nu}/L)}{L^{2}}\biggl(\frac{L^{2}}{A}\biggr)^{\chi^{+}_{\nu}}\, ,\ee
where $\kappa_{\nu}^{+}$ is a strictly positive function of the dimensionless parameter $\beta_{\text q}^{\nu}/L$ and $\chi_{\nu}^{+}$ a critical exponent, which is a strictly increasing function of $\nu$ satisfying $\smash{\lim_{\nu\rightarrow \frac{1}{2}^{+}}\chi_{\nu}^{+} =2}$ and $\smash{\lim_{\nu\rightarrow 1^{-}}\chi_{\nu}^{+} = +\infty}$.
\end{conjecture}
The asymptotic power-law behaviour \eqref{rhodecaynegc} implies that the positive curvature JT gravity theory exists only for positive cosmological constant. Moreover, only a finite number of moments of the area at zero cosmological constant exists, $\langle A^{k}\rangle <\infty$ for $k< \chi^{+}_{\nu}-1$. Note that these properties are shared by Liouville gravity. Actually, in the case of Liouville, it is straightforward to derive the area distribution law $\rho^{\text L}$ from the exact partition function given in Eq.\ \eqref{exactWLioudiskc}. One finds
\be\label{rhoL} \rho^{\text L}(A) = \frac{1}{\Gamma(1-\gamma_{\text{str}})}\biggl(\frac{\tilde\ell^{2}}{4}\biggr)^{1-\gamma_{\text{str}}} A^{\gamma_{\text{str}}-2}e^{-\frac{\tilde\ell^{2}}{4 A}}\ee
for
\be\label{ltildeLdef} \tilde\ell = \frac{\ell}{\sqrt{\sin\frac{\pi}{1-\gamma_{\text{str}}}}}\,\cdotp\ee
The large area behaviour of $\rho^{\text L}$ is thus also power-like. The moments $\langle A^{k}\rangle$ exist only when $k<1-\gamma_{\text{str}}$ and the theory is well-defined only in positive cosmological constant. 

The conjectured behaviour $\lim_{\nu\rightarrow 1^{-}}\chi_{\nu}^{+} = +\infty$ is consistent with the idea that this limit is semi-classical and all the moments of the area then automatically exist, with $\langle A^{k}\rangle = A_{*}^{k}$, where $A_{*}$ is the area of the relevant classical background. In fact, this argument is not without its subtleties, because the simple round disk classical background in positive curvature exists only when the boundary length $\ell$ is less than $2\pi L$. The behaviour in the $c\rightarrow -\infty$ limit is not known if this condition is not satisfied \cite{Loopcalc}. On may thus expect a non-trivial behaviour of the function $\kappa_{\nu}^{+}$ in Eq.\ \eqref{rhodecaynegc} in the limit $\smash{\nu\rightarrow 1^{-}}$ if $\beta_{\text q}>2\pi L$. The conjectured behaviour $\smash{\lim_{\nu\rightarrow \frac{1}{2}^{+}}\chi_{\nu}^{+} = 2}$ is motivated by the fact that the Brownian bridge model in positive curvature (that is to say, the generalization of the Kitaev-Suh model \cite{KitaevSuh} to the case of positive curvature) is found to have an area distribution that goes as $1/A^{2}$ asymptotically \cite{RDFollowup1} and it is plausible that JT gravity will behave similarly for very large areas.

\subsubsection{Lattice formulation, matrix model}

Regular tesselations of the two-sphere must satisfy the inequality $(p-2)(q-2)<4$. The only solutions are the five Platonic solids, $(p,q)=(3,3)$, $(3,4)$, $(3,5)$, $(4,3)$ or $(5,3)$, with a finite number of vertices and tiles. As in the case of negative curvature, the area of the $p$-gons are of the order of $L^{2}$, the precise formula being exactly as in Eq.\ \eqref{Anot2} with a global sign change in the right-hand side of this equation. The random SOP model on such lattices seems too trivial to have any interesting physical application; in particular, contrary to the negative curvature model, there is no Schwarzian-like limit in positive curvature. To obtain an interesting lattice model for positive curvature JT, one must thus refine the Platonic solids and discretize the two-sphere with tiles that may be made arbitrarily small; a problem well-studied in wheather forecasting and climate modeling.

The matrix model may also be conjectured to be relevant, mimicking Conjecture \ref{MMneglimconj}. The only change with the case of negative curvature is that the small parameter $\varepsilon$ in Eq.\ \eqref{qscalingdef} is now taken to be negative and the formula \eqref{lzeroqsca} for the ``Planck length'' becomes $\ell_{0}^{2} = -2\pi L^{2}\varepsilon=\frac{\pi L^{2}}{2}(4-q)$, which goes to $0^{+}$ in the continuum limit.

\section{\label{OutSec}Summary and outlook}

\subsection{\label{OutSec1}Upshots}

We have proposed a UV-complete definition of Euclidean Jackiw-Teitelboim quantum gravity, starting from the fundamental postulate that two-dimensional quantum gravity is a diffeomorphism-invariant theory of random metrics. The model on the disk has a dual formulation in terms of a theory of random self-overlapping loops that describes the fluctuating boundary. Self-overlapping loops must bound a distorted disk, a highly non-trivial constraint that can be analysed algorithmically using Blank cuts. A surprising aspect of this formulation is that some self-overlapping curves bound several distinct disks and therefore must be counted with a non-trivial multiplicity, being associated with diffeomorphism-inequivalent metrics on the disk.

At the microscopic level, JT gravity can thus be reduced to a purely combinatorial problem, which amounts to counting self-overlapping polygons with the appropriate multiplicity. This is very similar to, but probably harder than, the counting of ``maps'' (polygonizations) of surfaces in Liouville gravity. We have constructed a matrix model that does the counting, using an upgraded version of the method of dually weighted graphs. All these ideas open the way to studying the JT theories using techniques of enumerative geometry, matrix models and numerical simulations, which in the past have been instrumental in gaining insight into non-trivial random polygon models, such as the self-avoiding polygon model. 

We have learnt that the space of constant bulk curvature metrics one must consider is highly singular on the boundary, albeit perfectly smooth in the bulk. Viewed as closed curves immersed in the canonical spaces of constant curvature (hyperbolic space, Euclidean space or two-sphere), the boundaries are fractals with a non-trivial Hausdorff dimension $d_{\text H}$, which is conjectured to be two in the case of the pure gravity model. In particular, the boundaries do not have a smooth geometric length. Instead, one can associate to them a quantum length $\beta_{\text q}$, which has geometric length dimension $d_{\text H}$. This quantum length is analogous to the diffusion time, or Euclidean quantum mechanical time, that is associated to Brownian paths. At the microscopic level, it is proportional to the number of discretized, Planck-length edges that form the boundary, but in the continuum limit it takes on an anomalous dimension, reflecting the very singular nature of the boundary metric. As is made clear elsewhere \cite{ferrari,ferraJTconfgauge}, the boundary metric is actually distribution-valued. 

The short distance properties of the three models of JT gravity (negative, zero or positive curvature) are  identical: they have the same UV-completion. On the other hand, the IR properties drastically differ. The large distance behaviour is governed by quantum versions of the classical isopetrimetric inequalities. In negative curvature, the probability density for the area at fixed quantum boundary length $\beta_{\text q}$ is conjectured to decay as fast as $\exp (-\kappa^{-} A^{2}/(\beta_{\text q}L^{2}))$, where $\kappa_{-}$ is a strictly positive numerical constant and $L$ the curvature length scale, whereas in zero curvature a simple exponential decay is expected. In positive curvature, the absence of a classical isoperimetric inequality is conjectured to imply a power-law decay. Large area configurations become ubiquitous and the model is well-defined only for a positive cosmological constant, as is also the case for the Liouville theory.

The usual formulation of the negative curvature model in terms of the Schwarzian theory emerges when the cosmological constant $\La\rightarrow -\infty$. It is an approximate, IR effective description of JT gravity, valid on distance scales that are much larger than the curvature length scale $L$. In particular, the usual smooth boundary length $\ell$ is an effective, long distance parameter that can be expressed in terms of the microscopic parameters in the model, $\beta_{\text q}$, $\La$ and $L$, as in Eq.\ \eqref{ellans1}. The Schwarzian theory therefore plays a role, with respect to JT, that is similar to the one it plays with respect to the SYK model. We believe that the Schwarzian theory is a universal large-distance effective theory for many different random loop models in hyperbolic space, including models that do not have a metric or quantum gravity interpretation. Its emergence and universality is tightly related to the asymptotic ballistic property of random paths models in hyperbolic space. 

\subsection{\label{OutSec2}Comments on general topologies}

In the framework of the Schwarzian description, the construction of the negative curvature JT gravity on a two-dimensional manifold of arbitrary topology was presented in \cite{SSS}. The results are particularly simple and elegant. They rely on two basic assumptions.

The first assumption is to consider that the contributing geometries in genus $h$ with $b$ boundaries can be built as follows. One considers a ``core'' geometry of genus $h$ with $b$ geodesic boundaries of lengths $l_{1},\ldots, l_{b}$. These are precisely the geometries that are considered in topological gravity. They yield the moduli space volumes $V_{h,b}(l_{1},\ldots,l_{b})$ that are discussed in \cite{Mirzakhani}. On these core geometries are glued ``trumpets'' that correspond to annulus partition functions for which one of the boundary is a geodesic of fixed length whereas the other is a JT gravity boundary, which means that it has a fixed length but need not be a geodesic. Let us call $\smash{W^{\text{trumpet}}_{\text{Sch}}(\beta_{\text S},l)}$ the associated partition function, where $l$ is the geodesic boundary length and $\beta_{\text S}$ the Schwarzian limit parameter for the JT boundary, defined in Eq.\ \eqref{Schlimit}. The full JT gravity partition function, that follows from this gluing construction, is given by \cite{SSS}
\be\label{SSSpartgen} W^{(h,b)}_{\text{Sch}}(\beta_{\text{S},1},\ldots,\beta_{\text{S},b}) = 
\int_{[0,+\infty[^{b}} V_{h,b}(l_{1},\ldots,l_{b})
\prod_{i=1}^{b}W^{\text{trumpet}}_{\text{Sch}}(\beta_{\text{S},i},l_{i})\,  l_{i}\d l_{i}\, .\ee
The case of the annulus, $h=0$, $b=2$, is special, because there is no core manifold with this topology and geodesic boundaries. The partition function is then given by the gluing of two trumpets,
\be\label{SSSdbletrump}  W_{\text{Sch}}^{(0,2)}(\beta_{\text S,1},\beta_{\text S,2}) = \int_{0}^{\infty} W^{\text{trumpet}}_{\text{Sch}}(\beta_{\text{S},1},l)W^{\text{trumpet}}_{\text{Sch}}(\beta_{\text{S},2},l)\, l\d l\, .\ee
The measure of integration $l\d l$ comes from the usual Weil-Pertersson measure, the factor of $l$ being associated with the choice of twisting one has to make when one glues the trumpet on the core surface. 

The second assumption is that the trumpet partition function, in the Schwarzian limit, is obtained from an argument following closely the derivation of the disk partition function, Eq.\ \eqref{ZRAexact}, in the same limit. It is then found to be one-loop exact and given by
\be\label{Ztrumexact}  W^{\text{trumpet}}_{\text{Sch}}(\beta_{\text{S}},l) = 
\beta_{\text{S}}^{-1/2}e^{-\frac{l^{2}}{16\pi\beta_{\text S}}}\, .\ee

Using the equivalence between the Mirzakhani recursion relations \cite{Mirzakhani} and the Eynard topological recursion \cite{Eynardtoporec}, Eqs.\ \eqref{SSSpartgen} and \eqref{SSSdbletrump} imply that the JT gravity partition functions in the Schwarzian limit are computed by a standard Hermitian matrix model. The explicit spectral curve for the matrix model is derived straightforwardly from Eq.\ \eqref{Ztrumexact}.

In the microscopic formulation of JT gravity, the partition function in genus $h$ and $b$ boundaries is a function $W_{h,b}(\beta_{\text q,1},\ldots,\beta_{\text q,b})$ of the $b$ quantum boundary lengths $\beta_{\text q,i}$. In the Schwarzian limit $\La\rightarrow -\infty$, this microscopic partition function should match with $W^{(h,b)}_{\text{Sch}}$, with an appropriate identification between the microscopic parameters, $\beta_{\text q,i}$, $\La$ and $L$, and the Schwarzian parameters $\beta_{\text S,i}$. The discussion at the end of Section \ref{areaSchSec} suggests that $\beta_{\text S,i}$ is simply proportional to $\beta_{\text q,i}/L^{2}$ in the pure gravity theory. Even though it remains to be derived rigorously, we strongly believe that the matching between $W_{h,b}$ and $W^{(h,b)}_{\text{Sch}}$ in the Schwarzian limit is valid, and so is the analysis presented in \cite{SSS} in the same limit. The physics underlying the emergence of the Schwarzian/matrix model description is the same as for the disk topology, as explained in section \ref{nonzeroRSec}.

A fundamental open problem is the computation of $W_{h,b}(\beta_{\text q,1},\ldots,\beta_{\text q,b})$ for finite values of the microscopic parameters. Can we hope that the Saad-Shenker-Stanford framework can be suitably generalized in this case? Is it possible to maintain the relationship with a matrix model? The answer to these questions seems to be a definite no. Indeed, it is straightforward to construct many geometries contributing to $W_{h,b}$ that do not admit the simple ``trumpet decomposition'' that yields the formula \eqref{SSSpartgen}, because they do not contain a geodesic on which the gluing can be made. Let us illustrate this on the simplest case of the annulus. 

\begin{figure}
\centerline{\includegraphics[width=6in]{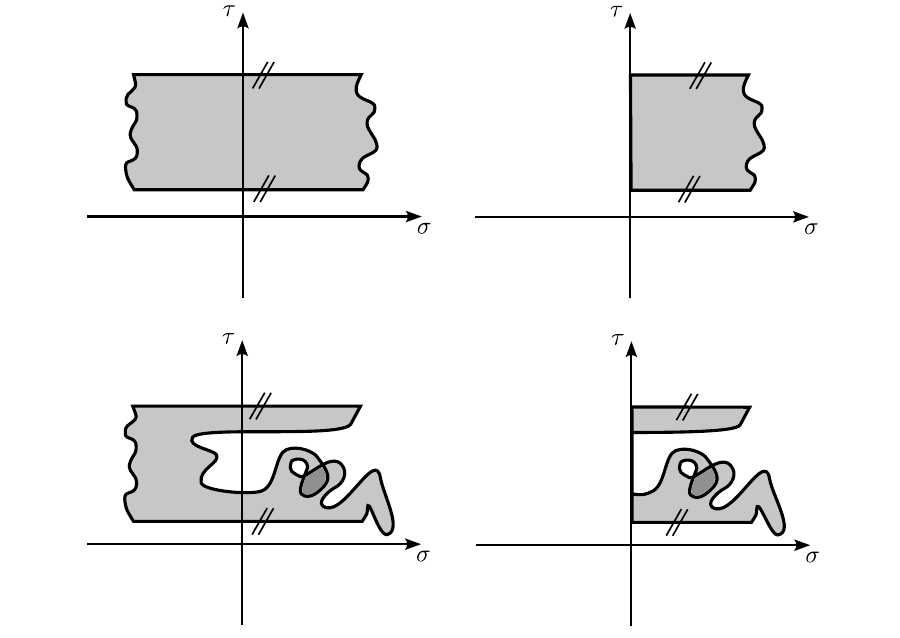}}
\caption{\label{annulusFig} Annuli of constant negative curvature depicted as immersions into $\Htwo$ in the strip representation, with metric \eqref{metH2stri} and identification $\tau\equiv\tau + l$. The identified sides are crossed out twice. Upper-left inset: a double-trumpet geometry, with its smooth and gently wiggling boundaries. It contains a closed geodesic at $\sigma=0$. Upper-right inset: one of the trumpet used to build the geometry on the upper-left inset. It has a geodesic boundary on the left and a fluctuating JT boundary on the right. Lower-left inset: a geometry with a hyperbolic holonomy but a wildly fluctuating boundary on the right, that also contains an overlapping region for the sake of illustration. The geometry does not contain a closed geodesic. Lower-right inset: a geometry a priori contributing to the ``trumpet'' partition function in the microscopic theory, with a geodesic boundary and a fluctuating JT boundary. The fluctuating boundary may have several components that end on the geodesic. Note that, in the microscopic theory, the fluctuating boundaries are fractal curves, exactly as in the case of the disk topology.}
\end{figure}

The ``double-trumpet'' geometries that contribute in \eqref{SSSdbletrump} are characterized by the fact that the holonomy of the associated flat $\PslR$ connexion along the non-trivial cycle of the annulus is a so-called hyperbolic element, conjugate to a matrix of the form
\be\label{cclasshyper} \begin{pmatrix} e^{l/2} & 0\\ 0 & e^{-l/2}\end{pmatrix}\, .\ee
The parameter $l$ is the length of the geodesic along which the gluing of the trumpets is made. These double-trumpet geometries can be explicitly described by an immersion of the annulus into hyperbolic space with an appropriate identification. The simplest way to depict this immersion is to use the strip representation of $\Htwo$, with metric, in units for which $L=1$, given by
\be\label{metH2stri} \d s^{2} = \frac{\d\tau^{2} + \d\sigma^{2}}{\cos^{2}\sigma}\, \cvp\quad -\frac{\pi}{2}<\sigma<\frac{\pi}{2}\, \cvp\ee
and identify $\tau\equiv\tau + a$. The geodesic of length $a$ is at $\sigma = 0$; see the upper-left inset in Fig.\ \ref{annulusFig}.

But there are myriads of annulus geometries that are not of this type, including geometries that can have arbitrarily large smooth boundary lengths. For instance, the annulus represented on the lower-left inset of Fig.\ \ref{annulusFig} does not contain a geodesic, even though its holonomy is hyperbolic. One can also use immersions into $\Htwo$ in the Poincar\'e disk representation, with metric given by Eq.\ \eqref{deltaHtwo2}, for instance of the type represented in Fig.\ \ref{cylFig}. The holonomy of such configurations is trivial. Configurations with a more general elliptic holonomy of the form 
\be\label{cclassellip} \begin{pmatrix} \cos (2\pi\alpha) & \sin (2\pi\alpha)\\ -\sin(2\pi\alpha) & \cos (2\pi\alpha)\end{pmatrix}\ee
are obtained if one identifies $\phi \equiv \phi +2\pi\alpha$ in \eqref{deltaHtwo2}. Annulus with a parabolic holonomy can be constructed using similar ideas, considering immersions into $\Htwo$ in the Poincar\'e half-plane representation, with metric $\d s^{2} = (\d x^{2} + \d z^{2})/z^{2}$ and identification $x\equiv x + \text{constant}$. The existence of all these geometries clearly invalidate the decomposition \eqref{SSSdbletrump} away from the Schwarzian limit. Similar constructions similarly invalidate \eqref{SSSpartgen}. Let us also note that the trumpet partition function itself, or similar partition functions involving mixed boundary conditions (some boundaries being geodesics, other being fluctuating JT boundaries), involve rather singular geometries in the microscopic theory, as illustrated in the lower-right inset of Fig.\ \ref{annulusFig}. The JT boundaries may indeed fluctuate wildly and touch the geodesics, yielding fluctuating boundaries with several components and degenerate, infinitely thin annulus portions.

\subsection{\label{OutSec3}Future research directions}

There are many possible future research directions. Let us briefly outline a few of them. 

We have already mentioned the new interesting combinatorial questions in relation with the self-overlapping polygon model, its matrix model formulation and the relevance of numerical methods. One would like to obtain an exact solution in the continuum limit and answer interesting conceptual questions about the model. For instance, is the measure of the set of configurations having a multiplicity index strictly greater than one zero, or is it strictly positive?

The study of the microscopic formulation of JT gravity on any topology is an outstanding and non-trivial problem. We have emphasized that the usual picture in terms of the trumpet decomposition and the associated matrix model formalism \`a la Saad-Shenker-Stanford does not generalize in any straightforward way. Yet, the understanding of the contributions of arbitrary topologies is important for many physical applications. On the one hand, these contributions are responsible for the so-called factorization problem, which makes it impossible to obtain a standard holographic interpretation in terms of a dual quantum theory. On the other hand, they seem to be essential for guaranteeing the full consistency of the quantum theory, in particular in relation with the black hole information problem.

In any case, our work indicates that the real-time quantization of JT gravity needs to be revisited. The UV-complete Euclidean theory we have developed in the present paper from the discretized point of view can also be formulated directly in the continuum by using the conformal gauge \cite{ferrari,ferraJTconfgauge}. This approach strongly suggests that the correct degree of freedom is the boundary conformal factor and the action to start with when quantizing has to include a ``Liouville-Hilbert'' term that is not present in the classical JT gravity action. It also suggests that, if space-time has the topology of a strip $I\times\mathbb R$, for a finite-size spatial interval $I$, boundary time reparameterization invariance remains a good unbroken gauge symmetry. This is the counterpart of the fact that, in the microscopic Euclidean theory, reparameterizations of the boundary coordinate are perfectly good unbroken gauge symmetries. This implies that there cannot be a non-trivial boundary Hamiltonian. In this sense, the model with boundaries is not so different than a model for a closed universe, with space-time topology $\text S^{1}\times\mathbb R$. Both the bulk time and the boundary time seem to be emergent concepts in this framework. The boundary time may emerge as an effective parameter in the near-AdS limit, possibly in a way similar to the emergence of the smooth boundary length parameter in the near-hyperbolic limit of the Euclidean theory.

Finally, we have explained that the positive curvature JT quantum gravity theory is a perfectly consistent model for positive cosmological constant, even though it has been very little studied up to now. In view of its relation with cosmological models, it may turn out to be the most interesting JT model for physics, especially if the real-time quantization can be studied, potentially shedding light on quantum gravity in de Sitter space and de Sitter holography.

\section*{Acknowledgements}

This work has a very long and tortuous history, tightly intertwined with the sad period of the Covid pandemic. Some results, including the role played by self-overlapping curves in JT gravity, were sketched as early as 2021 during a seminar at the Institute for Advanced Study in Dublin, Ireland and at the workshop on Quantum Geometry, Field Theory and Gravity in Corfu, Greece. I am particularly grateful to the organisers of the Quantum Gravity, Random Geometry and Holography programme at the Institut Henri Poincar\'e in Paris, France (J.~Barrett, D.~Benedetti, J.~Ben Geloun and R.~Loll, January and February 2023), of the workshop on Random Geometry in Math and Physics in Nijmegen, The Netherlands (T.~Budd, March 2023) and to the APCTP in Pohang, South Korea (J.~Yoon, April and May 2023, May 2024), for providing me with an exceptional scientific and working environment, which has enabled the work presented here to mature. I would also like to thank Soumyadeep Chaudhuri for a stimulating collaboration and useful discussions.

This work was supported in part by the International Solvay Institutes and the Asia Pacific Center for Theoretical Physics (APCTP) for participating in the APCTP Focus Program, ``Entanglement, Large N, and Black Hole 2024.'' It was completed during my stay at the APCTP in Pohang, South Korea, and at the Aspen Center for Physics, Aspen, Colorado, USA.

\appendix\clearpage

\section{\label{SemiclassApp}Classical limits of Liouville and JT gravity on the disk}

The continuum generating function (also called partition function in this context) $\mathsf W(\ell,\La)$ for Liouville quantum gravity on the disk, discussed in Section \ref{LiouclimitSec}, has a formal path integral representation
\be\label{Lioupathin1} \mathsf W(\ell,\La) = \frac{1}{\text{Vol}(\diff(\disk))}\int\! D g\, e^{-\frac{\La}{16\pi} A[g]}Z_{\text{CFT}}[g]\delta\Bigl(\ell - \oint \d s\Bigr)\, .\ee
The integral is over all the metrics on the disk, with the constraint, imposed by the $\delta$-function, that the disk boundary length is fixed and equal to $\ell$. The factor $Z_{\text{CFT}}[g]$ represents the partition function of an arbitrary CFT coupled to gravity. The cosmological constant $\La$ couples to the disk area $A[g]$. In conformal gauge, the metrics take the form
\be\label{confgaugeApp} g = e^{2\sigma}\delta\ee
for an arbitrary background disk metric $\delta$. Standard manipulations then yield\footnote{Basic references for Liouville gravity include \cite{Polyakov,Liouvillesuccess1,Liouvillesuccess2,gravityreviews}.} 
\be\label{Lioupathin2}  \mathsf W(\ell,\La) = \frac{Z_{\text{CFT}}[\delta]}{\text{Vol}(\PslR)}\int\! D \sigma\, e^{\frac{c-26}{24\pi}S_{\text L}[\delta,\sigma]-\frac{\La}{16\pi} A[g]}\delta\Bigl(\ell - \oint \d s\Bigr)\, .\ee
The partition function $Z_{\text{CFT}}[\delta]$ is an overall irrelevant, parameter-independent, constant. One must still formally divide by the (infinite) volume $\text{Vol}(\PslR)$ because there is a $\PslR$ subgroup of $\diff(\disk)$ that remains unbroken in conformal gauge. The measure $D\sigma$ over the Liouville field $\sigma$ is a formal background-independent non-linear integration measure. The action $S_{\text L}$ is the  Liouville action
\be\label{SLiouville} S_{\text L}[\delta,\sigma] = \int_{\disk}\!\d^{2}x\sqrt{\delta}\bigl(\delta^{ab}\partial_{a}\sigma\partial_{b}\sigma + R[\delta]\sigma\bigr) + 2\oint_{\partial\disk}\!\d s_{\delta}\, k[\delta]\sigma\, .\ee
The quantities $R[\delta]$, $s_{\delta}$ and $k[\delta]$ refer to the Ricci scalar, arc-length coordinate and extrinsic curvature of the boundary computed with the background metric $\delta$, respectively. The Liouville action captures the metric dependence of the CFT partition function $Z_{\text{CFT}}[g]$ and of the Fadeev-Popov ghosts determinants that appear when the gauge \eqref{confgaugeApp} is imposed.

Our goal is to discuss the semi-classical limit of the theory, defined in Eq.\ \eqref{Lioucllimit}, to leading $|c|\rightarrow\infty$ order. To do this, the precise definition of the measure $D\sigma$, or the treatment of the unbroken $\PslR$ gauge symmetry, are irrelevant. All we need is to solve the classical equations of motion associated with the classical action
\be\label{SclLiou}S_{\text{cl}} = \frac{|c|}{24\pi}\biggl(S_{\text L}[\delta,\sigma] + \mu A[g]\biggr)\, ,\ee
where $\mu =\frac{3}{2|c|} \La$ is the rescaled cosmological constant defined in \eqref{Lioucllimit}, taking into account the fact that the boundary length is fixed,
\be\label{bdconsclL} \oint\d s = \oint \d s_{\delta}\, e^{\sigma} = \ell\, .\ee
The classical action evaluated on the classical solution, $S_{\text{cl}}=S_{\text{cl}}^{*}$,  then yields the leading large $|c|$ behaviour of the generating function according to
\be\label{leadWLiou} \lim_{c\rightarrow -\infty}\frac{1}{|c|}\ln\mathsf W(\ell,\La) = -S_{\text{cl}}^{*}\, .\ee

Dealing with the constraint \eqref{bdconsclL} with the help of a Lagrange multiplier $\zeta$, it is straightforward to show that the equations of motion derived from \eqref{SclLiou} are equivalent to
\be\label{Lcleom} R = -2\tilde\La\, ,\quad k = \zeta\, ,\ee
where $R$ and $k$ are the Ricci scalar and extrinsic curvature of the disk boundary computed with the metric $g = e^{2\sigma}\delta$, respectively. The two equations in \eqref{Lcleom} come from the vanishing of the bulk and boundary terms in the variation of the action \eqref{SclLiou}.

The fact that the Ricci curvature is a negative constant and the extrinsic curvature of the boundary is constant as well implies that the solution is a disk embedded in hyperbolic space. Choosing the background metric $\delta$ to be the flat Euclidean metric for a disk of area $\pi$, the conformal factor of the solution in polar coordinates is given by
\be\label{sigclLiou} e^{\sigma} = \frac{2}{\sqrt{\mu}}\frac{r_{0}}{1-r_{0}^{2}\rho^{2}}\,\cvp\quad 0\leq 1\leq \rho\, .\ee
The constraint \eqref{bdconsclL} yields
\be\label{Lclrzeroell} r_{0} = \frac{2\pi}{\sqrt{\mu}\,\ell}\Biggl[-1+\sqrt{1+\frac{\mu\ell^{2}}{4\pi^{2}}}\Biggr]\, .\ee
On this solution, the area and the Liouville action evaluate to
\be\label{areaLstar} A^{*} = \frac{4\pi}{\mu}\frac{r_{0}^{2}}{1-r_{0}^{2}} = \frac{\ell r_{0}}{\sqrt{\mu}} \, \cvp\quad S_{\text L}^{*} = \frac{4\pi r_{0}^{2}}{1-r_{0}^{2}} + 4\pi\ln\frac{2r_{0}}{\sqrt{\mu}}\,\cdotp\ee
Introducing the variable
\be\label{zLcldef}z = \frac{\ell\sqrt{\mu}}{2\pi}\ee
and using \eqref{Lclrzeroell}, we get
\be\label{areaLstar2} \frac{r_{0}^{2}}{1-r_{0}^{2}}=\frac{1}{2}\bigl(-1+\sqrt{1+z^{2}}\bigr)\, ,\quad \frac{\sqrt{\mu}}{2r_{0}} = \frac{\pi}{\ell}\bigl(1+\sqrt{1+z^{2}}\bigr)\, .\ee
Combining \eqref{SclLiou}, \eqref{leadWLiou} and \eqref{areaLstar} with the above formulas, we finally obtain
\be\label{lead2W} \lim_{c\rightarrow -\infty}\frac{1}{|c|}\ln\mathsf W(\ell,\La) = -\frac{1}{6}\ln\frac{\ell}{2\pi} -\frac{1}{6}\bigl(\sqrt{1+z^{2}}-1\bigr)+\frac{1}{6}\ln\frac{1+\sqrt{1+z^{2}}}{2}\, \cdotp\ee
This result matches with the predictions of Eq.\ \eqref{WLioZamoexp}, \eqref{gammaexpZ} and \eqref{WLioZloop0} in the main text.

In the case of JT gravity in negative curvature, the constraint $R=-2$ is imposed by definition of the model. The semi-classical limit corresponds to $\La\rightarrow -\infty$ at fixed $\ell$ and was discussed in Section \ref{JTptSec}. Since $\La<0$ in this case, the classical solution, minimizing the action, corresponds to a disk of maximal area, at fixed boundary length $\ell$. As is well-known in the theory of extremal surfaces, this implies that the extrinsic curvature of the boundary must be constant. One thus gets exactly the same classical equations \eqref{Lcleom} as in the case of Liouville gravity, with $\mu = 1$. Using the value for the on-shell area given by \eqref{areaLstar}, \eqref{zLcldef} and \eqref{areaLstar2}, we obtain the leading behaviour of the partition function,
\be\label{leadZJTApp} \lim_{\La\rightarrow -\infty}\frac{1}{|\La|}\ln Z(\ell,\La) = \frac{A^{*}}{16\pi} = \frac{1}{8}\sqrt{1+\Bigl(\frac{\ell}{2\pi}\Bigr)^{2}} - \frac{1}{8} = \frac{\ell}{16\pi}\sqrt{1+\Bigl(\frac{2\pi}{\ell}\Bigr)^{2}}-\frac{1}{8}\, .\ee
This matches with Eqs.\ \eqref{ZIKTVLoop} and \eqref{treeJTneg} modulo the usual counterterms, as discussed below Eq.\ \eqref{treeJTneg}.

\section{\label{SAPasymp}Asymptotic analysis of the SAP generating function}

The aim of this Appendix is to derive the asymptotic formulas \eqref{WSAPasym1} for the SAP generating function \eqref{WSAPexact} in the limit
\be\label{limitApp} \La\beta^{3/2}\rightarrow - \infty\ee
of large negative cosmological constant. 

We are going to use a nice integral formula \cite{KMAiry} for the coefficients $\varphi_{r}$ of the asymptotic expansion of the logarithmic derivative of the Airy function, Eq.\ \eqref{Airylogasy},
\be\label{vphirAiform} \varphi_{r} = \frac{3}{2\pi^{2}}(-1)^{r+1}\int_{0}^{\infty}t^{3(r-1)/2} A(t)\,\d t\quad\text{for any $r\geq 1$.}\ee
where
\be\label{AAirydef} A(t) = \frac{1}{\Ai^{2}(t) + \Bi^{2}(t)}\,\cdotp\ee
Let us introduce the functions $\psi$ and $\tilde\psi$,
\be\label{psidefApp} \psi(x) = -\sum_{r=1}^{\infty}\frac{\varphi_{r}}{\Gamma\bigl((3r-1)/2\bigr)} x^{r} = \frac{3x}{2\pi^{2}}\tilde\psi(x)\, ,\ee
in terms of which the exact formula \eqref{WSAPexact} for the generating function reads
\be\label{WSAPAppex} W^{\text{SAP}}(\beta,\La) = 2\sqrt{\pi}a_{0}\frac{e^{-\beta\la}}{\beta^{5/2}}\Bigl(\frac{1}{2\sqrt{\pi}} + \psi(x)\Bigr)\, ,\quad\text{for}\ x=\frac{\La\beta^{3/2}}{32\pi^{5/2}a_{0}}\,\cdotp\ee

Note that the asymptotic behaviour of $\varphi_{r}$ can be straightforwardly obtained from \eqref{vphirAiform} by using the saddle point method and the known large $t$ behaviour of the Airy functions. One gets
\be\label{vphirasym} \varphi_{r}\underset{r\rightarrow \infty}{\sim}  (-1)^{r+1}
\sqrt{\frac{2}{\pi}}\Bigl(\frac{3}{4}\Bigr)^{r}r^{r-\frac{1}{2}}e^{-r}\, .\ee
This implies in particular that the series defining $\psi$ or equivalenlly $W^{\text{SAP}}$ have an infinite radius of convergence. 

Using \eqref{vphirAiform} and permuting the infinite sum and integral signs, we can perform explicitly the infinite sum. We get in this way
\be\label{tildepsif} \tilde\psi(x) = \int_{0}^{\infty}A(t) f\bigl(t^{3/2}x\bigr)\,\d t\ee
for
\be\label{fOFAiry} f(u) = -\frac{1}{3}\Bigl( e^{|u|^{2/3}} + 2 \cos\bigl(\frac{\sqrt{3}}{2}|u|^{2/3}\bigr) e^{-\frac{1}{2}|u|^{2/3}}\Bigr) + \frac{1}{2\sqrt{\pi} u}\Bigl( 1 - F\bigl(u^{2}/27\bigr)\Bigr)\, ,\ee
where the function $F$ is a generalised hypergeometric function,
\be\label{FgenhyperfApp} F(z) = {}_{1}F_{3}\bigl(\{1\},\{-1/6,1/6,1/2\};z\bigr)\, .\ee
When $u\rightarrow -\infty$, which is the relevant regime to study the limit \eqref{limitApp}, $f$ has an asymptotic expansion of the form
\be\label{fgenhyperasymp} f(u) = -\frac{2}{3}e^{|u|^{2/3}}+ \frac{1}{2\sqrt{\pi}u} + O\bigl(1/u^{3}\bigr)\, .\ee
This asymptotic form becomes valid when $u$ is of order one and larger. In the integral \eqref{tildepsif}, this is the region where $t$ is of the order $1/|x|^{2/3}$ or larger. Introducing an arbitrary $M>0$, we thus write \eqref{tildepsif} as a sum of two terms,
\be\label{tildepsif2} \tilde\psi(x) = \int_{0}^{M/|x|^{2/3}}A(t) f\bigl(t^{3/2}x\bigr)\, \d t+\int_{M/|x|^{2/3}}^{\infty}A(t) f\bigl(t^{3/2}x\bigr)\, \d t\, .\ee
The first term on the right-hand side of this equation can be evaluated when $|x|\rightarrow\infty$ by making the change of variable $t' = |x|^{2/3}t$ and then expanding the function $A$ near zero. This yields contributions of order $1/|x|^{2/3}$, with expansion parameter $1/|x|^{2/3}\propto 1/(|\La|^{2/3}\beta)$. Using \eqref{fgenhyperasymp}, the second integral can approximated as
\be\label{psif4til} \int_{M/|x|^{2/3}}^{\infty}f\bigl(t^{3/2}x\bigr)A(t)\,\d t \simeq -\frac{2}{3}\int_{M/|x|^{2/3}}^{\infty}A(t) e^{|x|^{2/3}t} \,\d t\,\ee
up to exponentially subleading terms. The integral \eqref{psif4til} can be evaluated  systematically by using the saddle-point method. The saddle point $t_{\text s}$ satisfies
\be\label{saddleApp} \frac{A'(t_{\text s})}{A(t_{\text s})} = -|x|^{2/3}\, .\ee
Using the asymptotic expansion
\be\label{AasymApp} A(t)\underset{t\rightarrow +\infty}{=} \pi e^{-\frac{4}{3}t^{3/2}}\sqrt{t}\Bigl(1-\frac{5}{24 t^{3/2}} + O\bigl(1/t^{3}\bigr)\Bigr) + O\bigl(e^{-\frac{8}{3}t^{3/2}}\bigr)\ee
we find that $t_{\text s}=\frac{1}{4}|x|^{4/3}$. It is thus natural to set $t = \frac{1}{4}|x|^{4/3}\tau$. Up to exponentially small terms in the limit $|x|\rightarrow\infty$, the integral \eqref{psif4til} is of the form
\be\label{psif5til} -\frac{|x|^{4/3}}{6}\int_{0}^{\infty} e^{-|x|^{2}\mathcal S(\tau)}\d\tau\ee
for 
\be\label{ScalApp}\mathcal S(\tau) = -\frac{1}{4}|x|^{2}\tau - \ln A\bigl(\frac{1}{4}|x|^{4/3}\tau\bigr) = \frac{1}{6}|x|^{2}\bigl(\tau^{3/2} -\frac{3}{2}\tau\bigr)
-\ln\frac{\pi|x|^{2/3}}{2}-\frac{1}{2}\ln\tau + O\bigl(1/|x|^{2}\bigr)\, .\ee
On this form, a systematic expansion in powers of $1/|x|^{2}$ around the saddle point $\tau=1$  can be obtained straightforwardly, to any desired order, by setting $\tau = 1 + \tau'/|x|$ and expanding. One finds in this way
\be\label{psitildeAppfinal} \tilde\psi(x) \underset{x\rightarrow -\infty}{=} -\frac{\pi^{3/2}}{3}|x|e^{x^{2}/12}\bigl(1 + O(1/x^{2})\bigr)\ee
Using \eqref{WSAPAppex}, this yields \eqref{SAPareaasymp}.

\end{document}